\documentclass[11pt]{article}
\usepackage{amsmath,amssymb,graphicx,float,epsf,tikz}
\def\RR{\hbox{I\kern-.2em\hbox{R}}}
\newcommand{\qed}{\hbox to 0pt{}\hfill$\rlap{$\sqcap$}\sqcup$ \vspace{3mm}}
\numberwithin{equation}{section}
\restylefloat{figure}
\usepackage{subfigure}
\usepackage{graphicx} 
\usepackage{amsmath}
\usepackage{amsthm}
\usepackage{amsfonts}
\usepackage{subfigure}
\usepackage{wrapfig}
\usepackage{enumerate}
\usepackage{enumitem}
\usepackage{array}
\usepackage{caption}
\usepackage[margin=1cm]{caption}
\usepackage{authblk}

\usepackage[utf8]{inputenc}

\newtheorem{theorem}{Theorem}
\newtheorem{lemma}[theorem]{Lemma}

\usepackage{graphicx}
\usepackage{tikz}
\usetikzlibrary{backgrounds, shadows}
\usetikzlibrary{arrows, positioning, shapes.geometric}
\usepackage{graphicx}
\allowdisplaybreaks
\usetikzlibrary{mindmap,backgrounds,arrows.meta}
\usetikzlibrary{shapes.geometric, arrows}
\tikzstyle{rect} = [draw, rectangle, fill=blue!20, text width=6em, text centered, minimum height=2em]
\tikzstyle{elli} = [draw, ellipse, fill=red!20, minimum height=2em]
\tikzstyle{circ} = [draw, circle, fill=white!20, minimum width=8pt, inner sep=5pt]
\tikzstyle{diam} = [draw, diamond, fill=white!20, text width=6em, text badly centered, inner sep=0pt]
\tikzstyle{line} = [draw, -latex']
\usepackage{a4wide} 
\usepackage{authblk} 

\usepackage{varioref}
\usepackage{xcolor}
\usepackage{hyperref}
\hypersetup{
	colorlinks=true,
	linkcolor=violet,
	citecolor=blue,
	urlcolor=teal
}
\makeatletter
\renewcommand{\@cite}[2]{\textcolor{blue}{[}\textcolor{blue}{#1}\textcolor{blue}{]}}
\makeatother
\usepackage{booktabs}

\date{}
\graphicspath{ {Images/} }

\begin{document}
	
		\title{Comparative Analysis of Stochastic and Predictable Models in the HIV Epidemic Across Genders}

	\author[1]{\small Nuzhat Nuari Khan Rivu\thanks{Email: nuzhatnuarikhan-rivu@uiowa.edu }}
	\author[2]{\small Md Kamrujjaman\thanks{Corresponding author Email: kamrujjaman@du.ac.bd}}
	\author[3]{\small  Shohel Ahmed\thanks{Email: shohel2@ualberta.ca }}

\affil[1]{\footnotesize Department of Mathematics, University of Iowa, Iowa City, IA 52242, USA}
\affil[2]{\footnotesize Department of Mathematics, University of Dhaka, Dhaka 1000, Bangladesh}
\affil[3]{\footnotesize Department of Mathematical and Statistical Sciences, University of Alberta, Edmonton, AB, Canada}

\maketitle
\vspace{-1.0cm}
\noindent\rule{6.35in}{0.02in}\\
\noindent {\bf Abstract}\\
This study conducts a comparative analysis of stochastic and deterministic models to better understand the dynamics of the HIV epidemic across genders. By incorporating gender-specific transmission probabilities and treatment uptake rates, the research addresses gaps in existing models that often overlook these critical factors. The introduction of gender-specific treatment, where only one gender receives treatment, allows for a detailed examination of its effects on both male and female populations. Two compartmental models, divided by gender, are analyzed in parallel to identify the parameters that most significantly impact the control of infected populations and the number of treated females. Stochastic methods, including the Euler, Runge-Kutta, and Non-Standard Finite Difference (SNSFD) approaches, demonstrate that stochastic models provide a more accurate and realistic portrayal of HIV transmission and progression compared to deterministic models. Key findings reveal that the stochastic Runge-Kutta method is particularly effective in capturing the epidemic's complex dynamics, such as subtle fluctuations in transmission and population changes. The study also emphasizes the crucial role of transmission probabilities and treatment rates in shaping the epidemic’s trajectory, highlighting their importance for optimizing public health interventions. The research concludes that advanced stochastic modeling is essential for improving public health policies and responses, especially in resource-constrained settings.
\\

\noindent\rule{6.35in}{0.02in}

\clearpage
\section*{Highlights}
\begin{enumerate}
	\item Incorporates gender-specific transmission probabilities and treatment uptake to address gaps in traditional HIV epidemic models. 
	\item Conducts a parallel analysis of gender-based compartmental models, identifying key parameters influencing the control of infected populations (both male and female) and the number of treated females.
	 \item Highlights the significant impact of varying transmission probabilities and treatment rates on epidemic behavior, underscoring the critical role of antiretroviral therapy (ART).
	\item Introduces gender-specific treatment, where only one gender receives treatment, and examines its effects on both genders. 
	\item Emphasizes the necessity of advanced stochastic modeling in epidemiology to enhance public health policy and response strategies.
\end{enumerate}

\clearpage
\section{Introduction}
\label{sec:introduction}
Human immunodeficiency virus (HIV) continues to be a critical global health issue, affecting approximately 42.3 million people globally, with 1.3 million new infections and nearly 1 million AIDS-related deaths recorded in 2023 alone \cite{who2024}. The origins of the HIV/AIDS pandemic can be traced back to June 1981, when the Centers for Disease Control and Prevention (CDC) published a report documenting the first cases of a rare lung infection, Pneumocystis carinii pneumonia, in five homosexual men in Los Angeles, later identified as the first known instances of AIDS \cite{cdcgov}. Initially perceived as confined to high-risk groups, such as homosexual men and intravenous drug users, HIV quickly transcended these boundaries, affecting diverse populations worldwide, with sub-Saharan Africa being particularly devastated by the disease, accounting for nearly two-thirds of global infections \cite{unaids_africa, buve2002spread}. The spread of HIV has been shaped by a complex interplay of social, economic, and cultural factors, which influenced both transmission patterns and the effectiveness of public health responses. Early efforts to combat the epidemic were often hampered by insufficient knowledge, a lack of resources, and deeply ingrained stigma and discrimination, which not only marginalized those most affected but also delayed the global response in terms of education, prevention, and treatment efforts \cite{herek2002,mohammad2024wiener,hassan2023mathematical}. International initiatives like the President’s Emergency Plan for AIDS Relief (PEPFAR) and the Global Fund to Fight AIDS, Tuberculosis, and Malaria have been instrumental in mobilizing resources and implementing large-scale interventions, leading to a substantial reduction in new HIV infections and AIDS-related deaths globally \cite{pepfar2023, globalfund2023}. Despite these advances, challenges remain, particularly in addressing the social and economic inequities that continue to fuel the epidemic, ensuring sustained access to treatment, and working toward global health targets such as ending AIDS as a public health threat by 2030 \cite{unaids2022}. Continued global cooperation, scientific innovation, and efforts to combat stigma and discrimination will be critical in achieving these goals, as HIV/AIDS remains one of the most significant public health challenges of our time.

The introduction of antiretroviral therapy (ART) in the mid-1990s marked a major turning point in reducing HIV-related mortality and transmission rates, significantly improving life expectancy and quality of life for people living with HIV \cite{zolopa2010evolution}. Clinical data reveal that infection is established early \cite{chun1998early}, and Archin et al. \cite{archin2012early} demonstrated that infected cells are primarily generated during the primary infection stage. Early initiation of ART can significantly reduce the number of infected cells, suggesting the possibility of limiting or even eradicating the virus. Over the past 20 years, many epidemiological compartment models have been developed to study the dynamics of HIV infection within the host and its transmission among populations. These models assist in decision-making for public health policies by simulating various scenarios and interventions. Early mathematical models were deterministic, offering simplified insights into HIV transmission and progression \cite{nowak2000virus, pankavich2016latent, kim2006latent, rong2009viral, perelson1993dynamics}. The inclusion of ART and its impact on transmission dynamics has also been extensively explored in the literature \cite{kim2006latent, hattaf2012optimal, ogunlaran2016management, ahmed2020dynamics, ahmed2021optimal}. However, these models often fall short of capturing the inherent randomness and variability of real-world epidemics. More recent advancements have introduced stochastic models that incorporate random variations, providing a more complex understanding of disease spread. Studies by Baleanu et al. \cite{baleanu2019competitive} and Kimbir et al. \cite{kimbir2012mathematical} have highlighted the utility of stochastic modeling in predicting epidemic outcomes and shaping public health strategies. This study builds on previous work by comparing the effectiveness of deterministic and stochastic models in capturing the complexities of the HIV epidemic, particularly across genders. Stochastic models consider the uncertainty and unpredictability inherent in the spread of disease, offering a more comprehensive understanding of potential outcomes and the efficacy of different control measures \cite{minnis2005effectiveness,mohammad2025stochastic,akhi2024seasonal}.

While significant progress has been made in modeling the HIV epidemic, there remains a gap in comparative analyses between stochastic and deterministic models specifically tailored to gender differences \cite{raza2019reliable}. Our study undertakes a comprehensive analysis to assess the reliability of a stochastic HIV model within a two-sex population, considering the impact of counseling and antiretroviral therapy (ART). Additionally, we examine scenarios where only females receive ART, an aspect often neglected by existing models, which tend to overlook gender-specific transmission probabilities and treatment uptake. Moreover, the need for more robust numerical methods to enhance the accuracy of stochastic models is evident, particularly in reflecting the real-world uncertainties and variabilities in HIV transmission and treatment efficacy. The primary objective of this study is to conduct a comparative analysis of stochastic and deterministic models of the HIV epidemic, focusing on gender differences, and demonstrating that stochastic models provide a more realistic depiction of disease dynamics by accounting for inherent randomness in transmission and treatment processes. This comparison leverages the threshold number $H^*$ to explore the conditions under which the disease can be controlled or persists within a population. We critique traditional numerical methods such as stochastic Euler and stochastic Runge–Kutta, highlighting their limitations with large time step sizes, and introduce the Stochastic Non-Standard Finite Difference (SNSFD) scheme \cite{baleanu2019competitive}. This novel scheme successfully preserves critical properties like positivity, boundedness, and dynamical consistency, providing a robust framework for analyzing the dynamics of the HIV epidemic \cite{ahmed2021optimal}.

The paper is organized as follows: Section \ref{sec:model} presents the deterministic HIV model, outlining the state variables, parameters, and differential equations that describe the progression of the disease. In section \ref{sec:steady_states}, we examine the model's behavior at both the disease-free and endemic equilibria. Following this, section \ref{sec:stochasticity} introduces stochastic elements, incorporating randomness into the model through stochastic differential equations. The subsequent section details the computational methods and numerical techniques used, including the stochastic Euler, stochastic Runge-Kutta, and Stochastic Non-Standard Finite Difference (SNSFD) methods. A comparative analysis is then provided, with graphical comparisons of the deterministic and stochastic models, showcasing the additional variability captured by the stochastic approach. In section \ref{sec:results}, the results and discussions emphasize the advantages of stochastic modeling in accurately depicting the dynamics of the HIV epidemic and the effects of antiretroviral therapy (ART). The paper concludes by highlighting the importance of stochastic modeling in informing public health strategies and suggesting areas for future research.

\section{Predictable HIV Model}\label{sec:model}
We propose a Predictable HIV model that incorporates both sexes, while accounting for the effects of counseling and antiretroviral therapy (ART). The state variables in the HIV model represent the different compartments within the population for modeling the progression and treatment of HIV: \(S_m\) for susceptible males, \(I_m\) for infected males, \(T_m\) for treated males, \(S_f\) for susceptible females, \(I_f\) for infected females, and \(T_f\) for treated females. We describe the  dynamics using the following differential equations::
\begin{equation}
\left\{
\begin{aligned}
\frac{dS_{m}}{dt} &= bN_m - P_{m}S_{m}(t) - \mu S_{m}(t) \\
\frac{dI_{m}}{dt} &= P_{m}S_{m}(t) - (\mu + \sigma_0 + \delta_m)I_{m}(t) \\
\frac{dT_{m}}{dt} &= \delta_m I_{m}(t) - (\mu + \sigma)T_{m}(t) \\
\frac{dS_{f}}{dt} &= bN_f - P_{f}S_{f}(t) - \mu S_{f}(t) \\
\frac{dI_{f}}{dt} &= P_{f}S_{f}(t) - (\mu + \sigma_0 + \delta_f)I_{f}(t) \\
\frac{dT_{f}}{dt} &= \delta_f I_{f}(t) - (\mu + \sigma)T_{f}(t)
\end{aligned}
\right.
\tag{1}
\label{Model_4}
\end{equation}

\noindent where $N_m$ and $N_f$ represent the overall number of males and females, respectively, i.e.
\begin{equation}
    N_m(t)=S_m(t)+I_m(t)+T_m(t)  \tag{2} 
\end{equation}
\begin{equation}
    N_f(t)=S_f(t)+I_f(t)+T_f(t)  \tag{3} 
\end{equation}

The \(P_m(t)\) and \(P_f(t)\) are frequency rates of infections of both males and females, respectively, and are as follows:\\
\begin{equation}
P_m(t) = \frac{c_m(t)\gamma_f I_{f}(t) - c_m^*(t)\gamma_f^*T_{f}(t)}{N_f} \tag{4}
\end{equation}
\begin{equation}
P_f(t) = \frac{c_f(t)\gamma_m I_{m}(t) - c_f^*(t)\gamma_m^*T_{m}(t)}{N_m} \tag{5}
\end{equation}

This paper describes several parameters that are crucial for modeling the progression and treatment of HIV. The parameters include \(\delta_m\) and \(\delta_f\), which represent the proportion of infected males and females receiving antiretroviral therapy (ART), respectively. The transmission probabilities \(\gamma_m\) and \(\gamma_f\) denote the likelihood of HIV transmission by an infected male and female, respectively. Additionally, \(\gamma^*_m\) and \(\gamma^*_f\) represent the transmission probabilities by males and females undergoing ART. The average number of sexual contacts between males and females, represented by \(c_m\) and \(c_f\), and those receiving ART, represented by \(c^*_m\) and \(c^*_f\), are also considered. The parameters \(\mu\), \(\sigma_0\), and \(\sigma\) denote the death rates of both genders, the mortality rate of infected individuals without ART, and the mortality rate of those receiving ART, in that sequence. The birth rate \(b\) and the life duration after infection \(T\) are included as well. The randomness in the system is captured by parameters \(\delta_1\), \(\delta_2\), and \(\delta_3\), which introduce stochastic elements into the model. These parameters are significant as they allow the model to account for the inherent variability and uncertainty in disease transmission and treatment processes, providing a more accurate depiction of the epidemic's dynamics. The impact of ART on reducing transmission rates and improving survival rates highlights the importance of these parameters in public health interventions. A detailed description of the state variable and the different parameters for the model \ref{Model_4} is listed in table \ref{table:state} and table \ref{table:parameter}, respectively. A schematic representation of the model \ref{Model_4} is given in figure \ref{fig:model}.

\begin{figure}[H]
	\centering
	\begin{tikzpicture}[
		compartment/.style={draw, rectangle, minimum height=1.2cm, minimum width=2cm, align=center, rounded corners, drop shadow},
		label/.style={align=center},
		arrow/.style={-Stealth, thick, draw=black},
		in_arrow/.style={-Stealth, thick, draw=black, dashed}
		]
		
		\node[compartment, fill=yellow!60] (Hm1) {$S_m(t)$}; 
		\node[compartment, fill=red!50, right=3cm of Hm1] (Hm2) {$I_m(t)$}; 
		\node[compartment, fill=green!50, right=3cm of Hm2] (Hm3) {$T_m(t)$}; 
		
		\node[label, below=1.2cm of Hm1] (mu1) {$\mu$};
		\node[label, below=1.2cm of Hm2] (mu_a0) {$\mu + \sigma_0$};
		\node[label, below=1.2cm of Hm3] (mu_a) {$\mu + \sigma$};
		
		\node[label, left=1.3cm of Hm1] (in_rate) {};
		
		\draw[arrow] (Hm1) -- (Hm2) node[midway, above, sloped] {\large $P_m(t)$};
		\draw[arrow] (Hm2) -- (Hm3) node[midway, above, sloped] {\large $\delta_m(t)$};
		
		\draw[arrow, bend right] (Hm1) to (mu1);
		\draw[arrow, bend right] (Hm2) to (mu_a0);
		\draw[arrow, bend right] (Hm3) to (mu_a);
		
		\draw[in_arrow] (in_rate) -- (Hm1) node[midway, above, sloped] {};
		
		 \begin{scope}[on background layer]
			\filldraw[fill=yellow!20, draw=yellow!60] (Hm1) circle[radius=1.5cm];
			\filldraw[fill=red!20, draw=red!50] (Hm2) circle[radius=1.5cm];
			\filldraw[fill=green!20, draw=green!50] (Hm3) circle[radius=1.5cm];
		\end{scope}
	
	\end{tikzpicture}
\end{figure}
\vspace{-12mm}
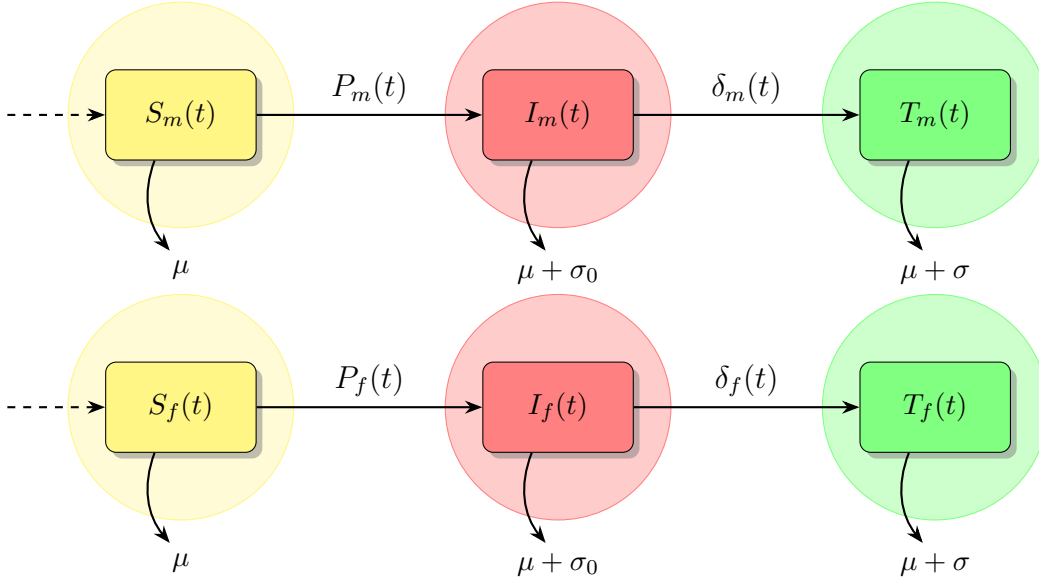
\begin{figure}[H]
\centering
\begin{tikzpicture}[
    compartment/.style={draw, rectangle, minimum height=1.2cm, minimum width=2cm, align=center, rounded corners, drop shadow},
    label/.style={align=center},
    arrow/.style={-Stealth, thick, draw=black},
    in_arrow/.style={-Stealth, thick, draw=black, dashed}
]

    \node[compartment, fill=yellow!60] (Hf1) {$S_f(t)$};
    \node[compartment, fill=red!50, right=3cm of Hf1] (Hf2) {$I_f(t)$};
    \node[compartment, fill=green!50, right=3cm of Hf2] (Hf3) {$T_f(t)$};

    \node[label, below=1.2cm of Hf1] (mu1) {$\mu$};
    \node[label, below=1.2cm of Hf2] (mu_a0) {$\mu + \sigma_0$};
    \node[label, below=1.2cm of Hf3] (mu_a) {$\mu + \sigma$};

    \node[label, left=1.3cm of Hf1] (in_rate) {};

    \draw[arrow] (Hf1) -- (Hf2) node[midway, above, sloped] {\large $P_f(t)$};
    \draw[arrow] (Hf2) -- (Hf3) node[midway, above, sloped] {\large $\delta_f(t)$};

    \draw[arrow, bend right] (Hf1) to (mu1);
    \draw[arrow, bend right] (Hf2) to (mu_a0);
    \draw[arrow, bend right] (Hf3) to (mu_a);

    \draw[in_arrow] (in_rate) -- (Hf1) node[midway, above, sloped] {};

    \begin{scope}[on background layer]
        \filldraw[fill=yellow!20, draw=yellow!60] (Hf1) circle[radius=1.5cm];
        \filldraw[fill=red!20, draw=red!50] (Hf2) circle[radius=1.5cm];
        \filldraw[fill=green!20, draw=green!50] (Hf3) circle[radius=1.5cm];
    \end{scope}

\end{tikzpicture}
\caption{Dynamic compartmental diagram representing the progression of HIV within a two-sex population.}
\label{fig:model}
\end{figure}

\begin{table}[ht]
\centering
\caption{The state variables for Model (\ref{fig:model}).}
\label{tab:5.1}
\begin{tabular}{@{}ll@{}}
State variable & Description                           \\ \midrule
$S_m$               & The number of susceptible males       \\
$I_m$               & The number of infected males        \\
$T_m$            & The number of infected males who recovered by getting ART \\
$S_f$             & The number of susceptible females   \\
$I_f$               & The number of infected females \\
$T_f$               & The number of infected females who recovered by getting ART  \\ \bottomrule
\end{tabular}
\label{table:state}
\end{table}

\begin{table}[htbp]
\centering
\caption{Description of parameters in Model (\ref{fig:model}).}
\label{tab:5.2}
\resizebox{\textwidth}{!}{%
\begin{tabular}{@{}clccc@{}}
\toprule
\textbf{Parameter} & \textbf{Description} & \multicolumn{2}{c}{\textbf{Values}} & \textbf{Source} \\
\cmidrule(lr){3-4}
& & \textbf{DFE} & \textbf{EE} & \\ 
\midrule
$\delta_m$ & The portion of infected males getting ART & 0.15 & 0.10 & \cite{kimbir2012mathematical} \\
$\delta_f$ & The portion of infected females getting ART & 0.90 & 0.35 & \cite{kimbir2012mathematical} \\
$\gamma_m$ & The probability of transmission by a male who is infected & 0.30 & 0.30 & \cite{raza2019competitive} \\
$\gamma_f$ & The probability of transmission by a female who is infected & 0.15 & 0.15 & \cite{raza2019competitive} \\
$\gamma^*_m$ & The probability of transmission by a male who is getting ART & 0.25 & 0.25 & Assumed \\
$\gamma^*_f$ & The probability of transmission by a female who is getting ART & 0.12 & 0.12 & Assumed \\
$c_m$ & The average number of sexual cohorts between males who are infected and females & 3 & 8 & Estimated \\
$c_f$ & The average number of sexual cohorts between females who are infected and males & 3 & 8 & Estimated \\
$c^*_m$ & The average number of sexual cohorts between males who are receiving ART and females & 1 & 1 & Assumed \\
$c^*_f$ & The average number of sexual cohorts between females who are receiving ART and males & 1 & 1 & Assumed \\
$T$ & The duration of life for males and females after infection & 10 & 10 & \cite{raza2019competitive} \\
$K$ & The value of ART for genders & 1 & 1 & \cite{raza2019competitive} \\ 
$b$ & The birth rate of both genders & 0.5 & 0.5 & Estimated \\
$\mu$ & The death rate of both genders & 0.5 & 0.5 & Estimated \\
$\sigma_0$ & The mortality rate of infected genders who do not receive ART & 0.20 & 0.20 & \cite{raza2019competitive} \\
$\sigma$ & The mortality rate of infected genders who receive ART & 0.13 & 0.13 & \cite{raza2019competitive} \\
$\delta_1$ & The randomness of equation - of the system & 0.4 & 0.4 & Estimated \\
$\delta_2$ & The randomness of equation - of the system & 0.3 & 0.3 & Estimated \\
$\delta_3$ & The randomness of equation - of the system & 0.2 & 0.2 & Estimated \\
\bottomrule
\end{tabular}
}
\label{table:parameter}
\end{table}

\noindent The HIV model (\ref{Model_4}) is shown below in its normalised and simplified form \cite{raza2019competitive}:
\begin{flalign}
\frac{di_{m}(t)}{dt} &= (c_{m}\gamma_f i_{f}(t) - c_{m}^*\gamma_f^*t_{f}(t))(1 - i_{m}(t) - t_{m}(t))- (b + \delta_m)i_{m}(t)+ \sigma t_{m}(t)i_{m}(t) && \nonumber \\
&\quad + \sigma_0 i_{m}^2(t), &\tag{6} \label{eq:15} \\
\frac{di_{f}(t)}{dt} &= (c_{f}\gamma_m i_{m}(t) - c_{f}^*\gamma_m^* t_{m}(t))(1 - i_{f}(t) - t_{f}(t))- (b + \delta_f)i_{f}(t)+ \sigma t_{f}(t)i_{f}(t) && \nonumber \\
&\quad + \sigma_0 i_{f}^2(t), &\tag{7} \\
\frac{dt_{f}(t)}{dt} &= \delta_f i_{f}(t) - (b + \sigma)t_{f}(t) + \sigma_0 i_{f}(t)t_{f}(t) + \sigma t_{f}^2(t). &\tag{8} \label{eq:18}
\end{flalign}
where,
\[
s_{m}(t) = \frac{S_{m}(t)}{N_m}, \quad i_{m}(t) = \frac{I_{m}(t)}{N_m}, \quad t_{m}(t) = \frac{T_{m}(t)}{N_m},
\]
\[
s_{f}(t) = \frac{S_{f}(t)}{N_f}, \quad i_{f}(t) = \frac{I_{f}(t)}{N_f}, \quad t_{f}(t) = \frac{T_{f}(t)}{N_f},
\]
with conditions,
\[
s_{m}(t) + i_{m}(t) + t_{m}(t) = 1 \quad \text{and} \quad s_{f}(t) + i_{f}(t) + t_{f}(t) = 1.
\]

\subsection*{Basic Reproduction Number ($H^*$):}
To calculate the basic reproduction number ($H^*$), first of all we need to find the infection rates (\(\mathcal{F}\)) and transition rates (\(\mathcal{V}\)) \cite{dietz1993estimation}. \\
\textbf{Infection Rates (\(\mathcal{F}\)):}
\begin{align*}
	\mathcal{F}_m &= c_m \gamma_f \frac{I_f}{N_f} S_m \\
	\mathcal{F}_f &= c_f \gamma_m \frac{I_m}{N_m} S_f 
\end{align*}

\textbf{Transition Rates (\(\mathcal{V}\)):}
\begin{align*}
	\mathcal{V}_m &= (\mu + \sigma_0 + \delta_m) I_m \\
	\mathcal{V}_f &= (\mu + \sigma_0 + \delta_f) I_f 
\end{align*}

Then, we need to construct the next generation matrix ($K$). For that, we get,\\
\textbf{New Infections:}
\[
\mathbf{F} = \begin{bmatrix}
	0 & c_m \gamma_f \\
	c_f \gamma_m & 0 \\
\end{bmatrix}
\]

\textbf{Transitions:}
\[
\mathbf{V} = \begin{bmatrix}
	\mu + \sigma_0 + \delta_m & 0 \\
	0 & \mu + \sigma_0 + \delta_f \\
\end{bmatrix}
\]

Then,

\[
\mathbf{V}^{-1} = \begin{bmatrix}
	\frac{1}{\mu + \sigma_0 + \delta_m} & 0 \\
	0 & \frac{1}{\mu + \sigma_0 + \delta_f} \\
\end{bmatrix}
\]

So, the \textbf{Next Generation Matrix \(\mathbf{K}:\)}
\[
\mathbf{K} = \mathbf{F} \mathbf{V}^{-1} = \begin{bmatrix}
	0 & c_m \gamma_f \\
	c_f \gamma_m & 0 \\
\end{bmatrix}
\begin{bmatrix}
	\frac{1}{\mu + \sigma_0 + \delta_m} & 0 \\
	0 & \frac{1}{\mu + \sigma_0 + \delta_f} \\
\end{bmatrix}
= \begin{bmatrix}
	0 & \frac{c_m \gamma_f}{\mu + \sigma_0 + \delta_f} \\
	\frac{c_f \gamma_m}{\mu + \sigma_0 + \delta_m} & 0 \\
\end{bmatrix}
\]

\noindent Now, to calculate the basic reproduction number, $H^*$ is the spectral radius (dominant eigenvalue) of the next-generation matrix \(\mathbf{K}\).

The eigenvalues \(\lambda\) of \(\mathbf{K}\) are determined by solving:
\[
\text{det}(\mathbf{K} - \lambda \mathbf{I}) = 0
\]

For matrix \(\mathbf{K}\):
\[
\begin{vmatrix}
	-\lambda & \frac{c_m \gamma_f}{\mu + \sigma_0 + \delta_f} \\
	\frac{c_f \gamma_m}{\mu + \sigma_0 + \delta_m} & -\lambda \\
\end{vmatrix}
= 0
\]

This simplifies to:
\[
\lambda^2 = \frac{c_m \gamma_f c_f \gamma_m}{(\mu + \sigma_0 + \delta_f)(\mu + \sigma_0 + \delta_m)}
\]

Therefore, the Basic Reproduction Number ($H^*$) is:
\[
H^* = \sqrt{\frac{c_m \gamma_f c_f \gamma_m}{(\mu + \sigma_0 + \delta_f)(\mu + \sigma_0 + \delta_m)}}
\]

\subsection*{Calculation and Significance of \(H^*\) for DFE and EE:}
From table \ref{tab:5.2}, we have the following parameter values:\\
\noindent For Disease-Free Equilibrium (DFE):
\[
c_m = 3, \, c_f = 3, \, \gamma_m = 0.30, \, \gamma_f = 0.15, \, \mu = 0.5, \, \sigma_0 = 0.20, \, \delta_m = 0.15, \, \delta_f = 0.90,
\]
which gives $H^*= 0.545$.\\
\noindent For Endemic Equilibrium (EE):
\[
c_m = 8, \, c_f = 8, \, \gamma_m = 0.30, \, \gamma_f = 0.15, \, \mu = 0.5, \, \sigma_0 = 0.20, \, \delta_m = 0.10, \, \delta_f = 0.35,
\]
which gives $H^*= 1.85$.

\noindent \textbf{Situational Explanation:}
\begin{itemize}
	\item \textbf{Disease-Free Equilibrium (DFE):} The \( H^*\) value for DFE is approximately 0.545. This means that, on average, each infected individual will infect less than one other person. When \( H^* < 1 \), the infection cannot spread in the population, and the disease will eventually die out. This represents a scenario where the disease is controlled, and new infections are not sufficient to sustain an outbreak.
	
\item	\textbf{Endemic Equilibrium (EE):} The \( H^* \) value for EE is approximately 1.85. This means that, on average, each infected individual will infect more than one other person. When \( H^* > 1 \), the infection can spread in the population, leading to a sustained outbreak or endemic situation. This represents a scenario where the disease is prevalent in the population, with continuous transmission occurring.
\end{itemize}

\noindent These values indicate that the disease dynamics can vary significantly based on the underlying parameters and the equilibrium state of the population \cite{van2008further}. Control measures, such as increasing treatment rates or reducing transmission probabilities, are crucial in moving the population from an endemic to a disease-free state.

\section{The HIV model (when $\delta_f > 0, \delta_m=0$)}
\label{sec:steady_states}
Our approach exclusively focuses on infected females who are undergoing antiretroviral therapy (ART). For this, we place $\delta_f > 0, \delta_m=0$ in the given equations (\ref{eq:15}-\ref{eq:18}) as follows:
\begin{flalign}
&\frac{di_{m}(t)}{dt} = c_{m}\gamma_{f}i_{f}(t)(1 - i_{m}(t)) - b i_{m}(t) + \sigma_{0}i_{m}^{2}(t) \tag{9} \label{eq:19} &\\
&\frac{di_{f}(t)}{dt} = c_{f}\gamma_{m}i_{m}(t)(1 - i_{f}(t) - t_{f}(t)) - (b + \delta_{f})i_{f}(t) + \sigma i_{f}(t)t_{f}(t) + \sigma_{0} i_{f}^{2}(t) \tag{10} &\\
&\frac{dt_{f}(t)}{dt} = \delta_{f} i_{f}(t) - (b + \sigma)t_{f}(t)+ \sigma_{0}i_{f}(t) t_{f}(t) + \sigma t_{f}^{2}(t) \tag{11} \label{eq:21} &
\end{flalign}

\noindent Here, we have only explored the equations in a deterministic environment. In the next section, we have introduced stochasticity to the equations where only the females are receiving ART ($\delta_f > 0$).

\section{The HIV Model with Stochasticity (when $\delta_f > 0, \delta_m=0$)}
\label{sec:stochasticity}
The following is a list of the stochastic HIV model equations that correlate to the deterministic model equations (\ref{eq:19}-\ref{eq:21}) for the HIV epidemic \cite{raza2019reliable,arif2019reliable}:
\begin{flalign}
&i_{m}(t) = (c_{m}\gamma_{f}i_{f}(t)(1 - i_{m}(t)) - b i_{m}(t) 
+ \sigma_{0}i_{m}^{2}(t))\,dt + \delta_{1}dB_{1}(t)i_{m}(t) \tag{12} \label{eq:22} &
\end{flalign}
\begin{flalign}
&i_{f}(t) = (c_{f}\gamma_{m}i_{m}(t)(1 - i_{f}(t) - t_{f}(t)) - (b + \delta_{f})i_{f}(t) + \sigma i_{f}(t)t_{f}(t) + \sigma_{0}i_{f}^{2}(t))\,dt \notag & \\
&\quad + \delta_{2}dB_{2}(t)i_{f}(t) \tag{13}&
\end{flalign}
\begin{flalign}
&t_{f}(t) = (\delta_{f}i_{f}(t) - (b + \sigma)t_{f}(t) + \sigma_{0}i_{f}(t)t_{f}(t) + \sigma t_{f}^{2}(t))\,dt + \delta_{3}dB_{3}(t)t_{f}(t) \tag{14}&
\end{flalign}
The Brownian motion is denoted by \( B_{k}(t) \), \( (k = 1, 2, 3) \), and \( \delta_{1} \), \( \delta_{2} \), and \( \delta_{3} \) are the randomness of each equation (\ref{eq:19}-\ref{eq:21}), respectively. The estimated values of these parameters are given in table \ref{tab:5.2}.

\section{Computational Methods}\label{sec:methods}
For computational results and visualization, we used the Stochastic Euler Method, Stochastic Runge-Kutta Method, and Stochastic Non-Standard Finite Difference Method. Detailed discussions and the discretization process for these methods, 
applied to solve the proposed problems, are provided in Appendix \ref{appen}.

\section{Comparative analysis: Predictable vs Stochastic}\label{sec:graphical}
\subsection{Deterministic}
In this subsection, we have adjusted three parameters: $\gamma_f$, $\gamma_m$, and $\delta_f$ for Endemic Equilibrium. We are examining how these changes affect three population types: $i_m$, $i_f$, and $t_f$.

In figure \ref{fig:sub11}, we can see that, for lower values of the probability of transmission by a female who is infected ($\gamma_f$), the number of infected males would be less. In figure \ref{fig:sub12}, we observe the same result. Also, it is crystal clear from the analysis that the number infected females with ART ($t_f$) would be less if the probability of transmission by a female who is infected ($\gamma_f$) is less (see figure \ref{fig:sub13}).

\begin{figure}[H]
	\centering
	\subfigure[$i_m$ vs $t$.]{
		\includegraphics[width=0.3\linewidth]{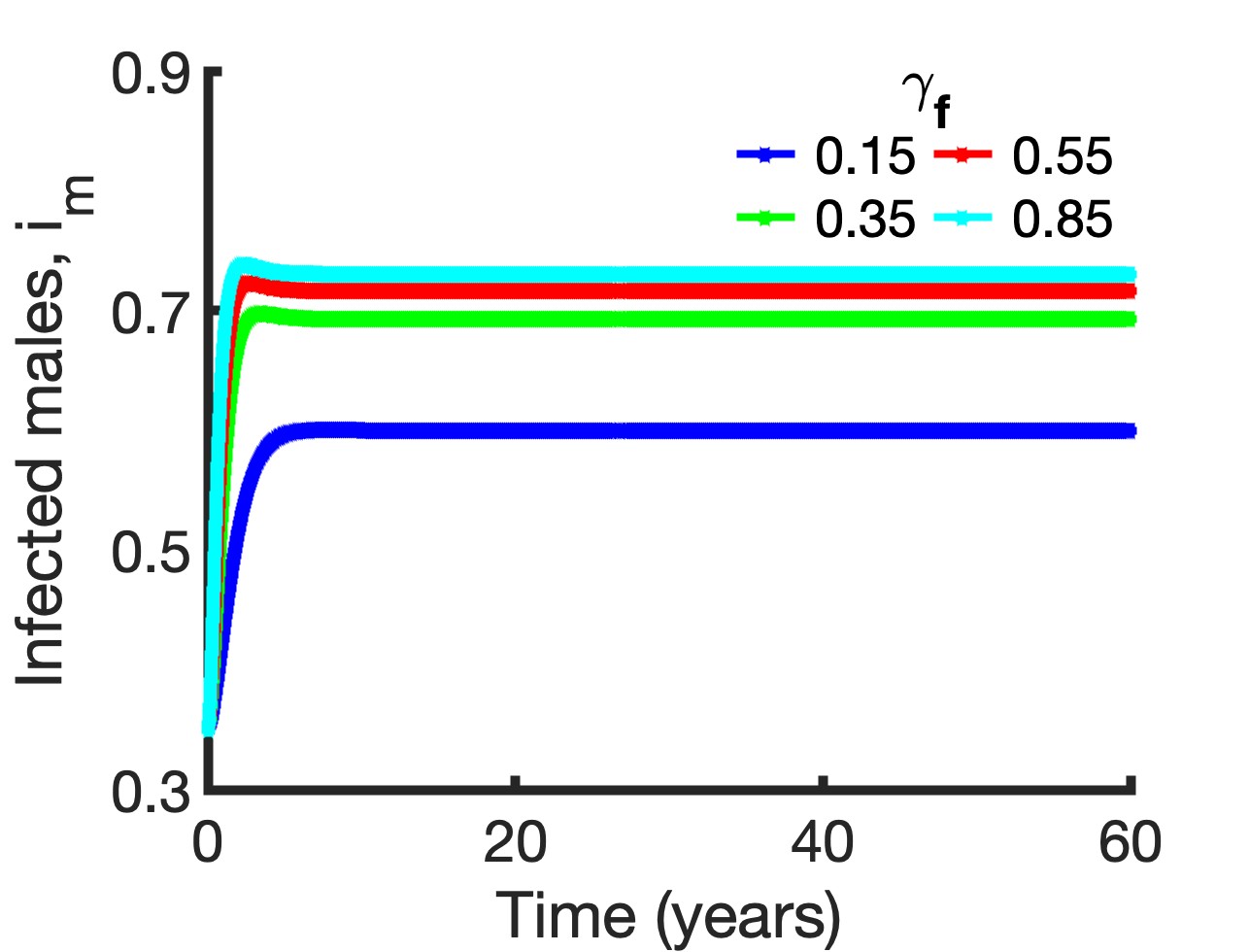}
		\label{fig:sub11}
	}
	\subfigure[$i_f$ vs $t$.]{
		\includegraphics[width=0.3\linewidth]{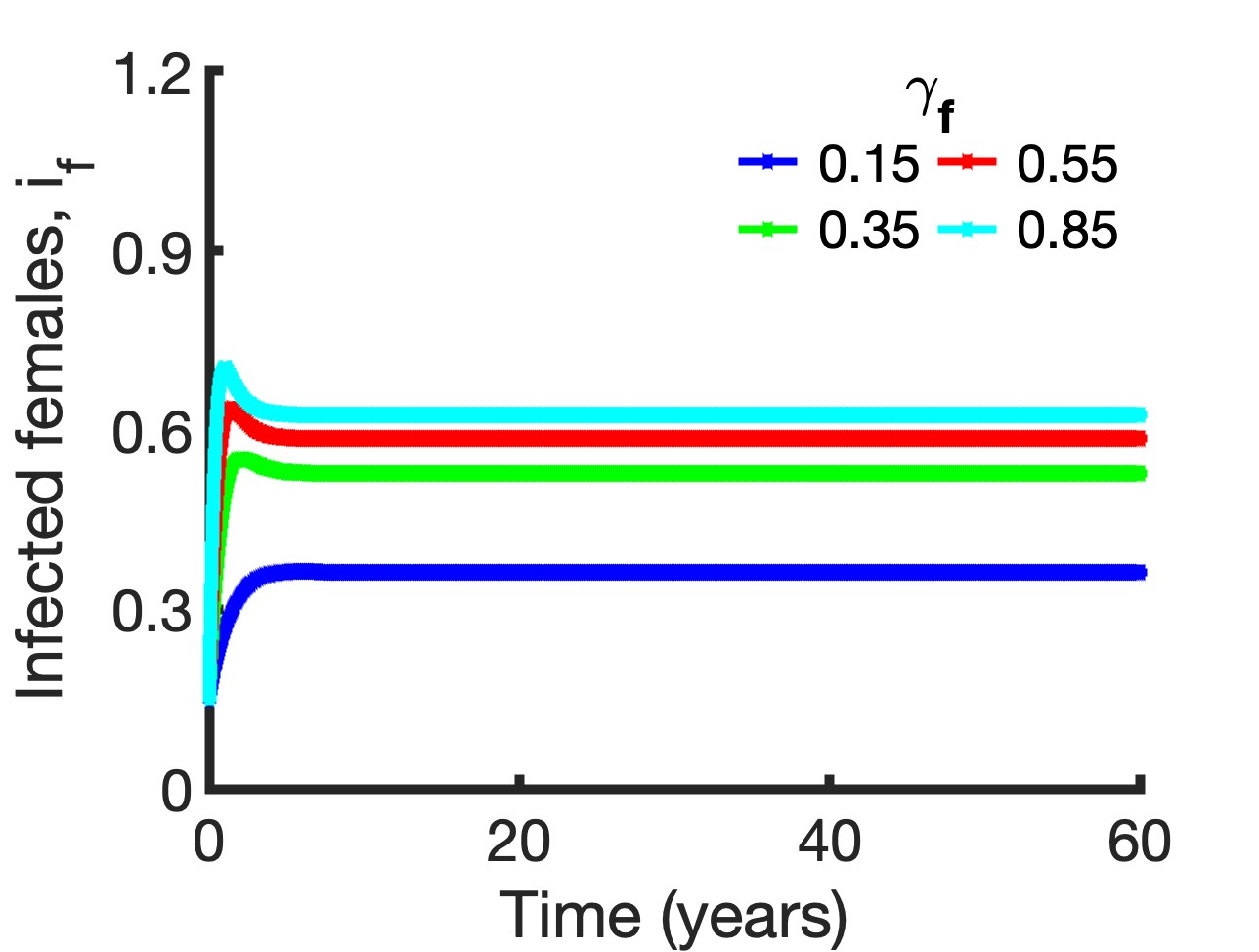}
		\label{fig:sub12}
	}
	\subfigure[$t_f$ vs $t$.]{
		\includegraphics[width=0.3\linewidth]{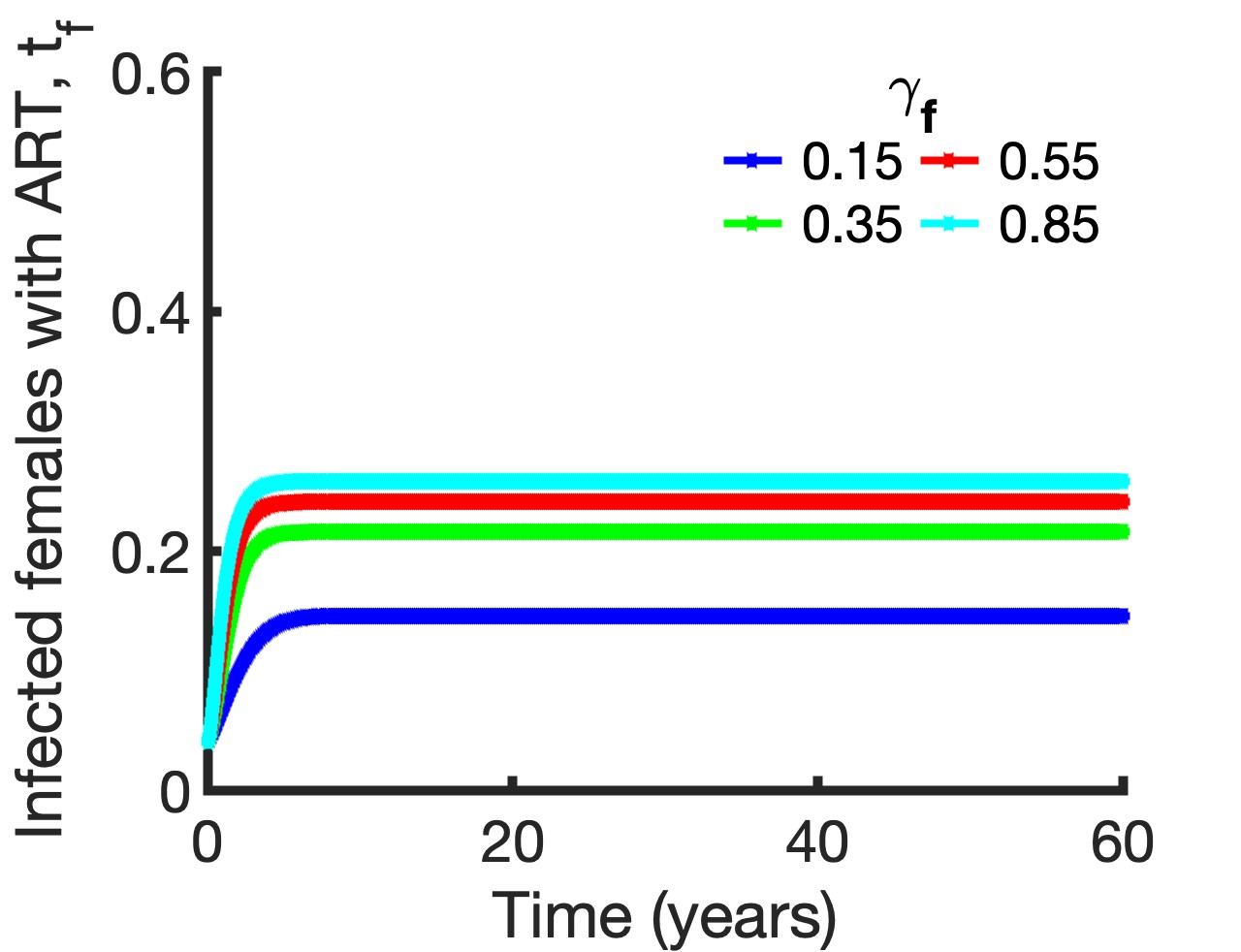}
		\label{fig:sub13}
	}
	\caption{Endemic Equilibrium (EE) of the deterministic model for varying values of the probability of transmission by a female who is infected  ($\gamma_f$).}
	\label{fig:EE_Deterministic_gamma_f}
\end{figure}

\noindent Additionally, we explored how the transmission probability by an infected male ($\gamma_m$) impacts three population groups. Figures \ref{fig:sub21}, \ref{fig:sub22} and \ref{fig:sub23} illustrate that lower values of $\gamma_m$ correspond to reduced numbers of $i_m$, $i_f$, and $t_f$ respectively.

\begin{figure}[H]
	\centering
	\subfigure[$i_m$ vs $t$.]{\includegraphics[width=0.3\linewidth]{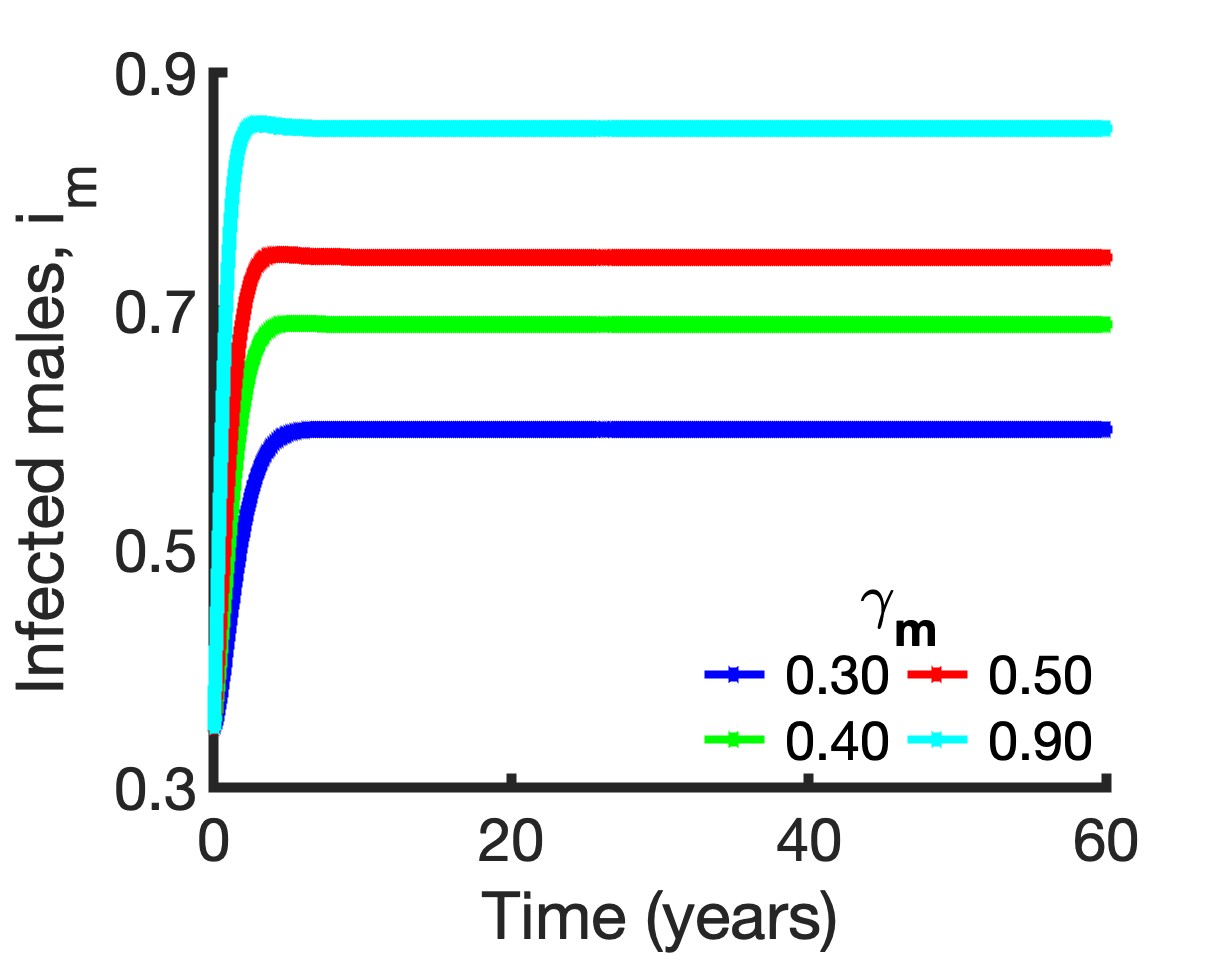}\label{fig:sub21}}
	\subfigure[$i_f$ vs $t$.]{\includegraphics[width=0.3\linewidth]{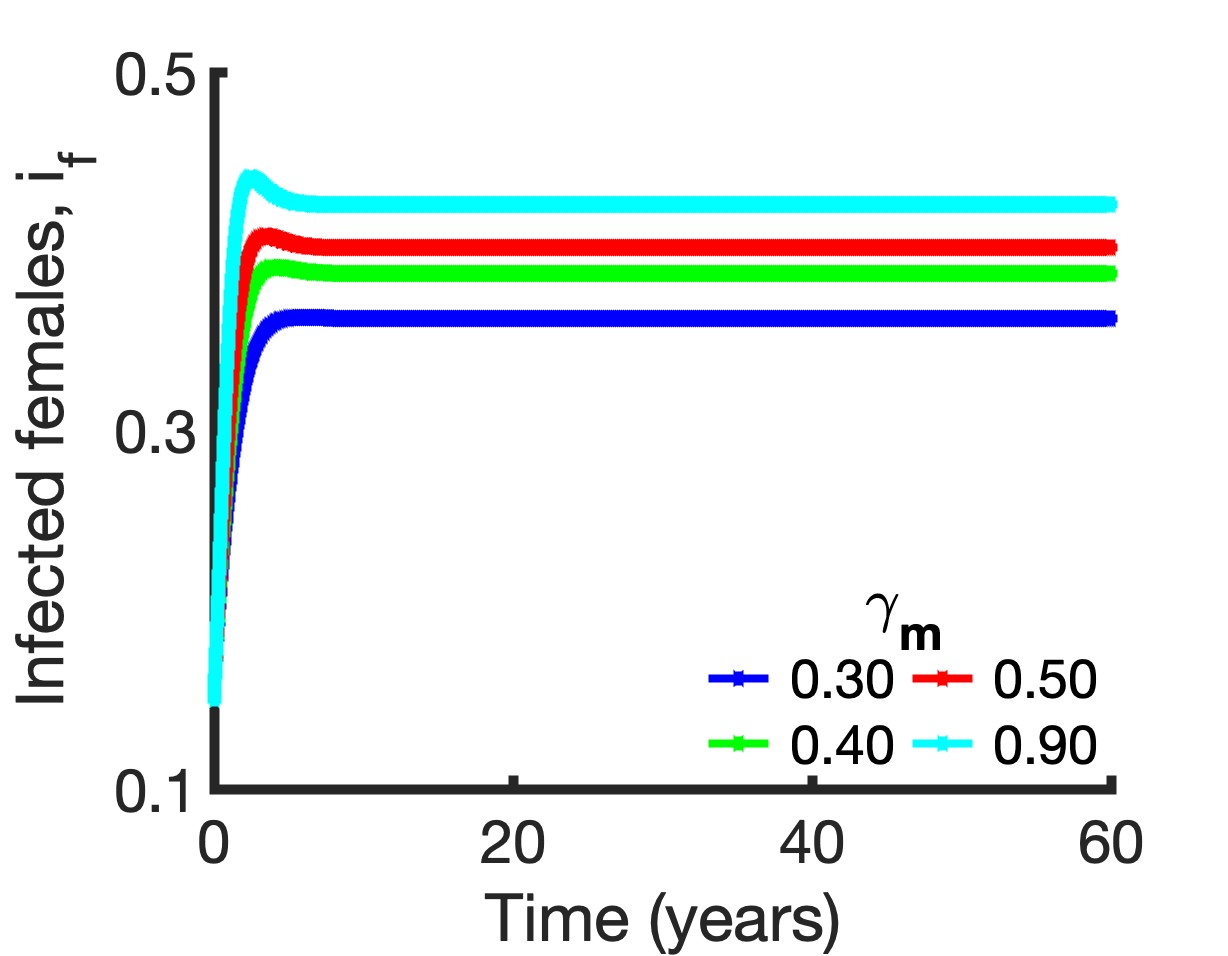}\label{fig:sub22}}
	\subfigure[$t_f$ vs $t$.]{\includegraphics[width=0.3\linewidth]{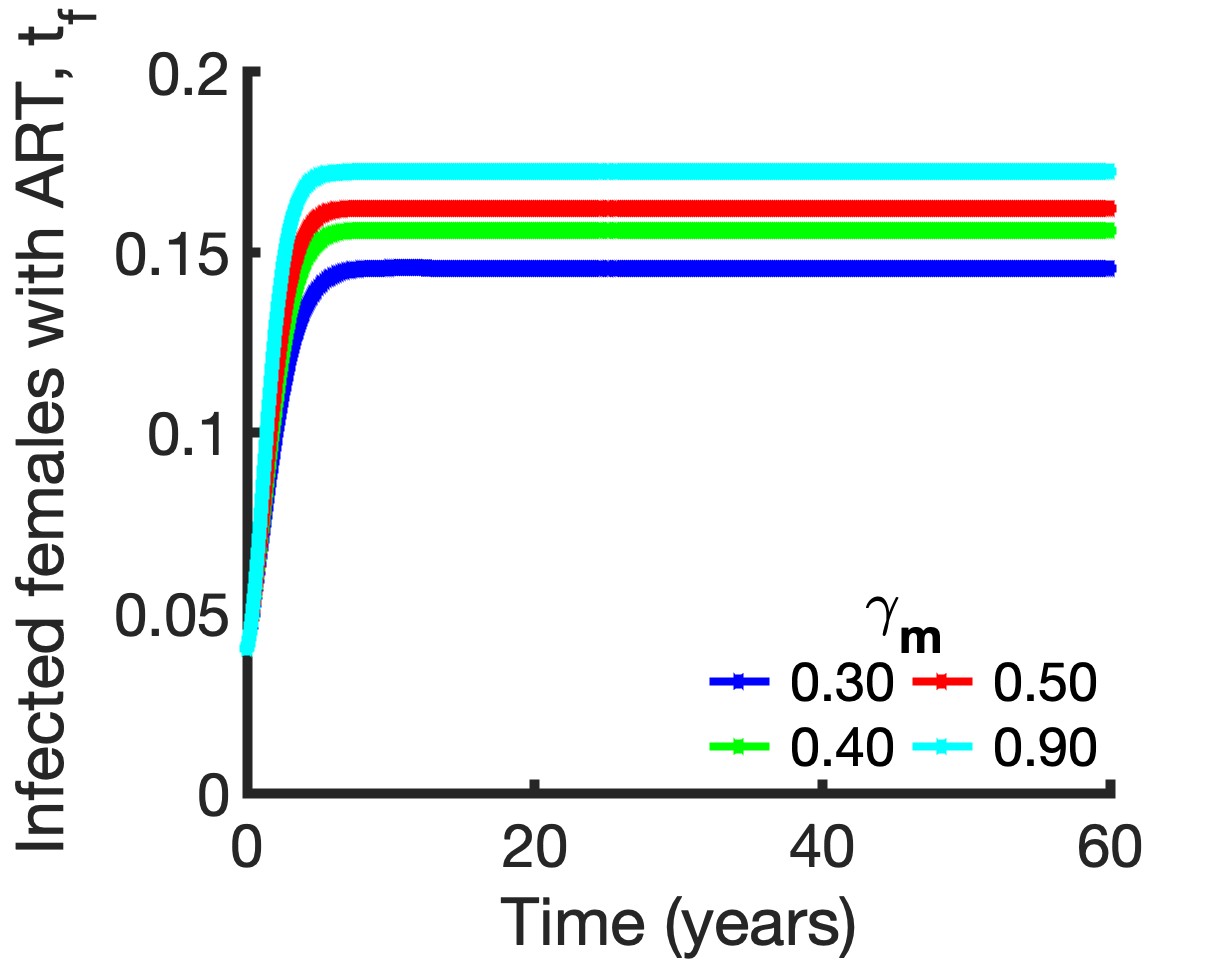}\label{fig:sub23}}
	\caption{Endemic Equilibrium (EE) of the deterministic model for varying values of the probability of transmission by a male who is infected  ($\gamma_m$).}
	\label{fig:EE_Deterministic_gamma_m}
\end{figure}

Furthermore, we analyzed the impact of antiretroviral treatment (ART) provided to infected females ($\delta_f$) on three demographic groups. Increasing the treatment rate results in fewer infected males ($i_m$) and females ($i_f$) (see figures \ref{fig:sub31} and \ref{fig:sub32}). Naturally, as the proportion of infected females receiving ART ($\delta_f$) increases, the number of females treated with ART ($t_f$) also rises which is clearly visible in figure \ref{fig:sub33}.

\begin{figure}[H]
	\centering
	\subfigure[$i_m$ vs $t$.]{\includegraphics[width=0.3\linewidth]{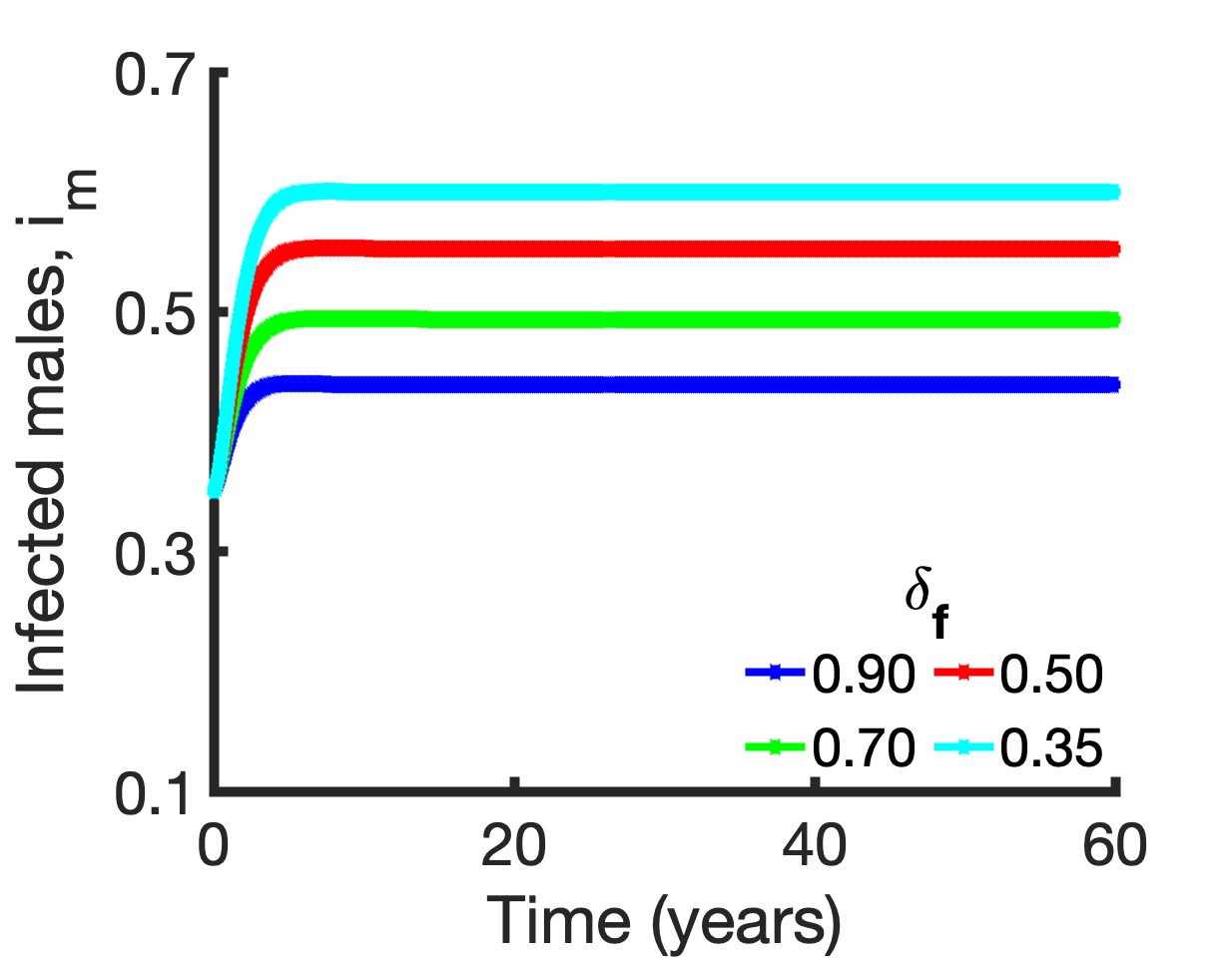}\label{fig:sub31}}
	\subfigure[$i_f$ vs $t$.]{\includegraphics[width=0.3\linewidth]{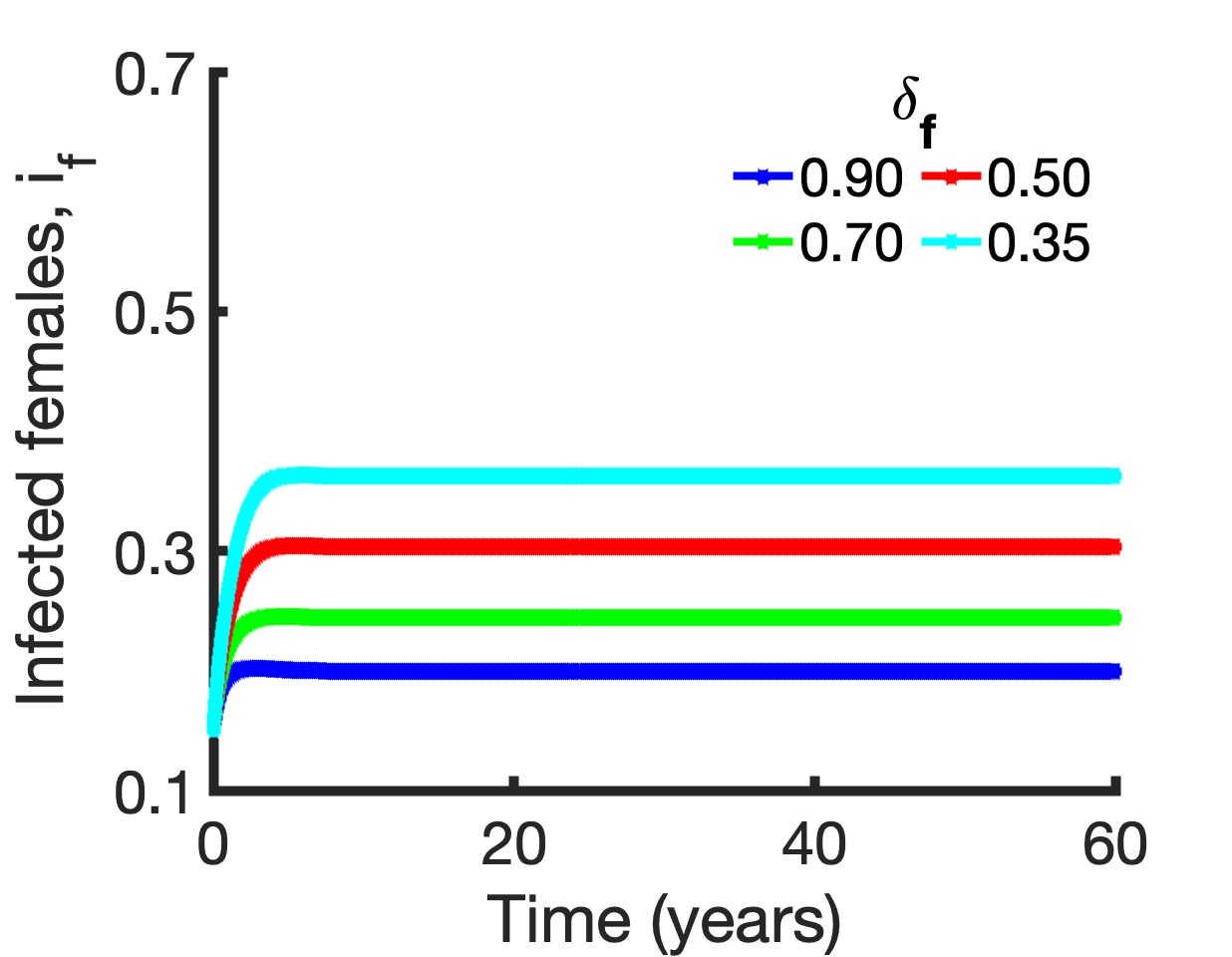}\label{fig:sub32}}
	\subfigure[$t_f$ vs $t$.]{\includegraphics[width=0.3\linewidth]{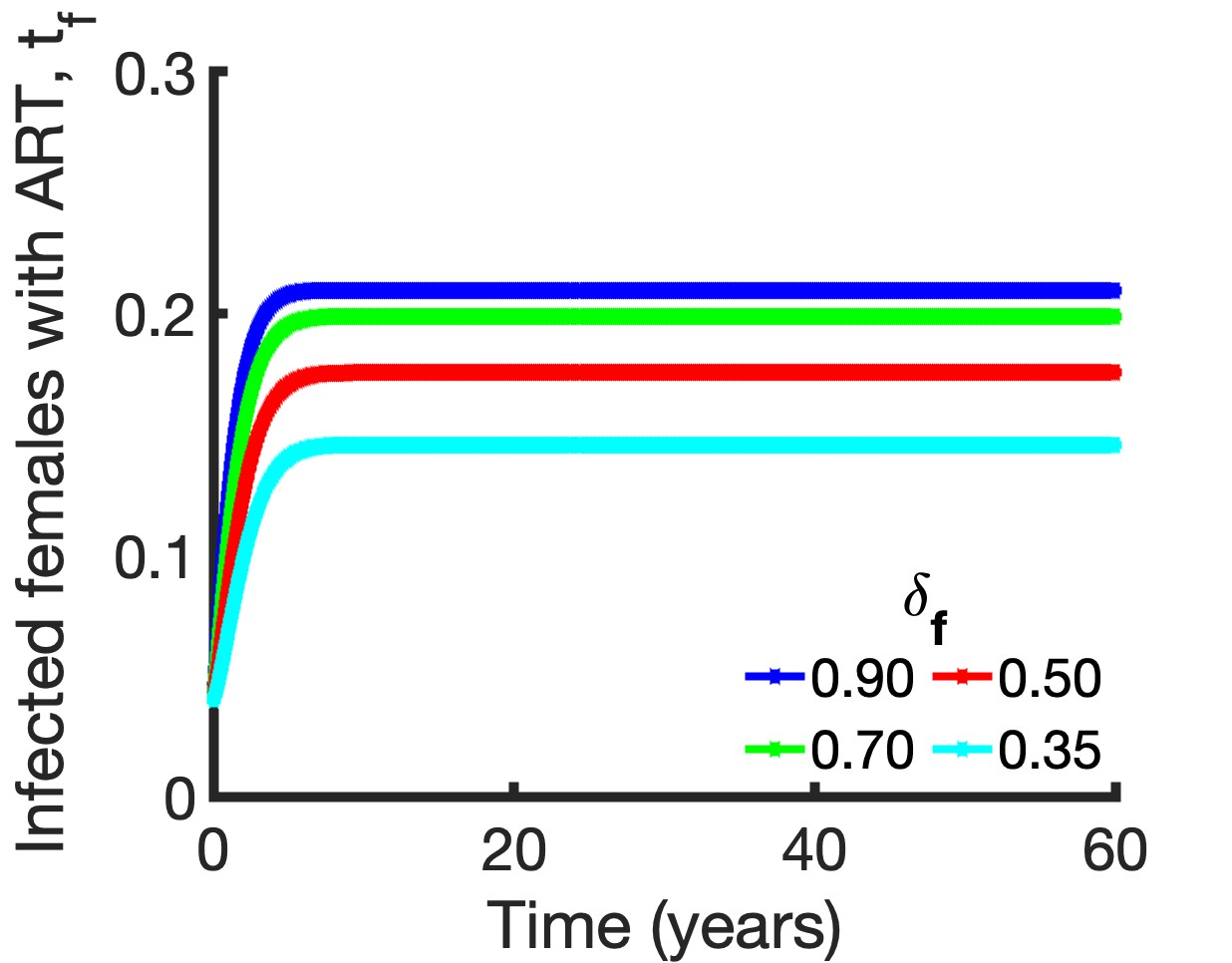}\label{fig:sub33}}
 \caption{Endemic Equilibrium (EE) of the deterministic model for varying values of the portion of the infected females getting ART ($\delta_f$).}
 \label{fig:EE_Deterministic_delta_f}
\end{figure}

In this subsection, we have adjusted three parameters: $\gamma_f$, $\gamma_m$, and $\delta_f$ for Disease-Free Equilibrium. We are examining how these changes affect three population types: $i_m$, $i_f$, and $t_f$.\\
\begin{figure}[H]
	\centering
	\subfigure[$i_m$ vs $t$.]{\includegraphics[width=0.3\linewidth]{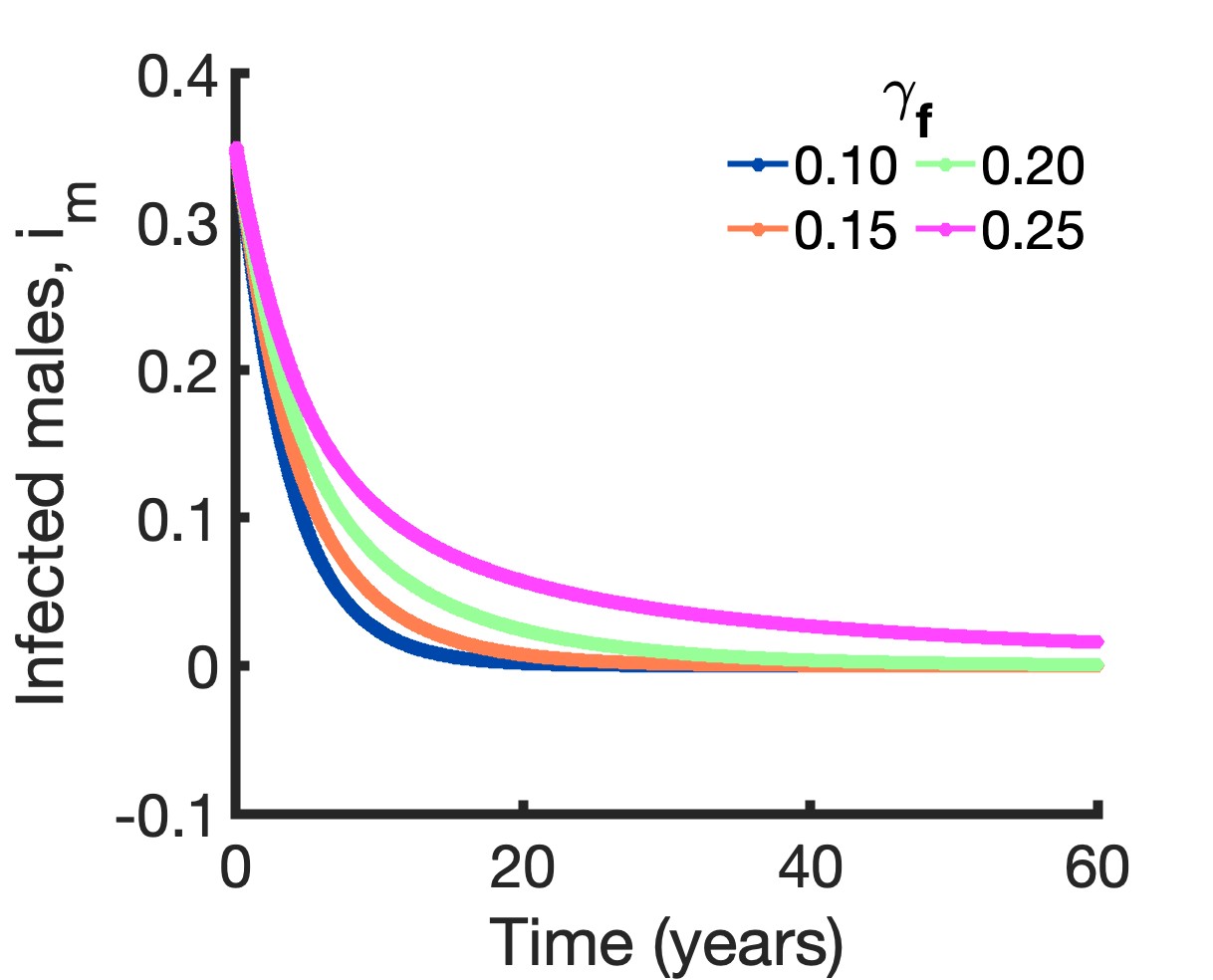}\label{fig:sub114}}
	\subfigure[$i_f$ vs $t$.]{\includegraphics[width=0.3\linewidth]{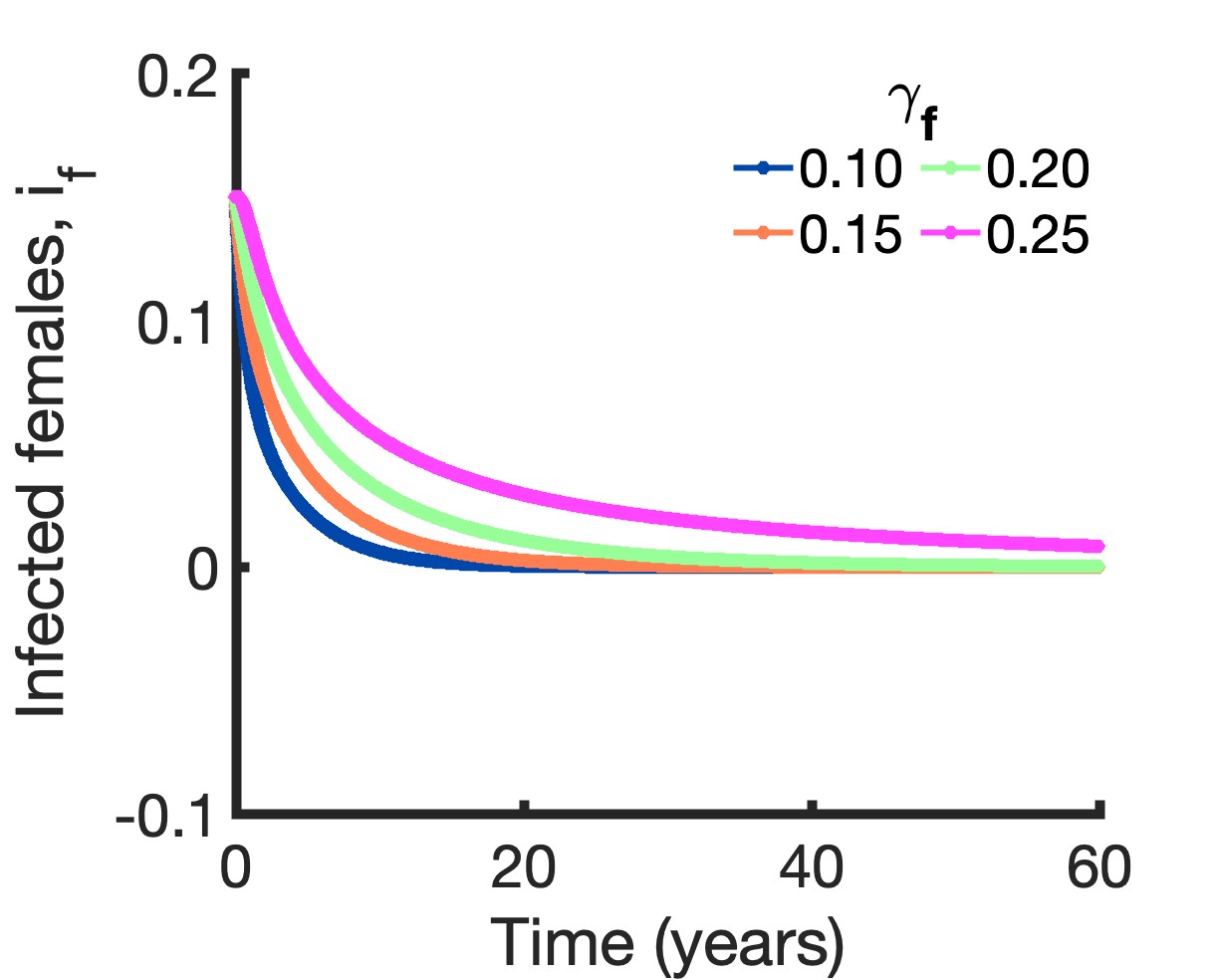}\label{fig:sub124}}
	\subfigure[$t_f$ vs $t$.]{\includegraphics[width=0.3\linewidth]{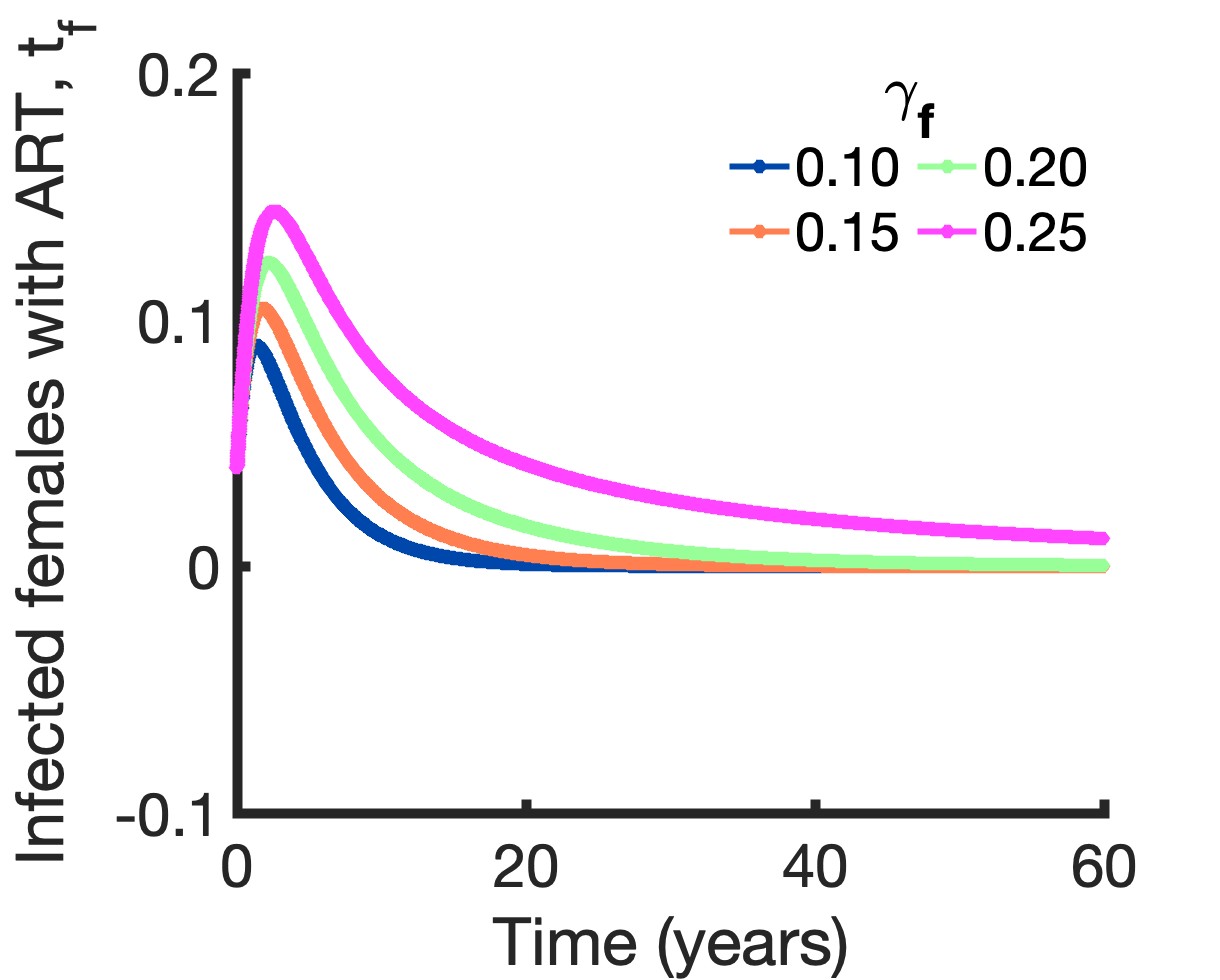}\label{fig:sub134}}
  \caption{Disease-Free Equilibrium (DFE) of the deterministic model for varying values of the probability of transmission by a female who is infected  ($\gamma_f$).}
\label{fig:DFE_Deterministic_gamma_f}
\end{figure}

\noindent In figures \ref{fig:sub114}, \ref{fig:sub124} and \ref{fig:sub134}, it is evident that a lower transmission probability from an infected female ($\gamma_f$) corresponds to reduced numbers of infected males ($i_m$), infected females ($i_f$), and infected females receiving antiretroviral therapy (ART) ($t_f$) respectively.

\noindent According to figure \ref{fig:sub214}, the deterministic model shows that a lower probability of transmission by an infected male ($\gamma_m$) results in a significant reduction in the number of infected males ($i_m$). This indicates that controlling the transmission rate among males can effectively decrease the overall male infection rate in the population. Figure \ref{fig:sub224} illustrates that a decreased probability of transmission by an infected male ($\gamma_m$) also leads to a reduction in the number of infected females ($i_f$). This suggests that the transmission dynamics between genders are interconnected, and reducing male-to-female transmission can help in lowering the female infection rate. As shown in figure \ref{fig:sub234}, the number of infected females receiving antiretroviral therapy (ART) ($t_f$) decreases with a lower transmission probability from an infected male ($\gamma_m$). This highlights that interventions aimed at reducing male transmission rates can have a positive impact on reducing the number of females needing ART, thereby potentially lowering the overall treatment burden.

\begin{figure}[H]
	\centering
	\subfigure[$i_m$ vs $t$.]{\includegraphics[width=0.3\linewidth]{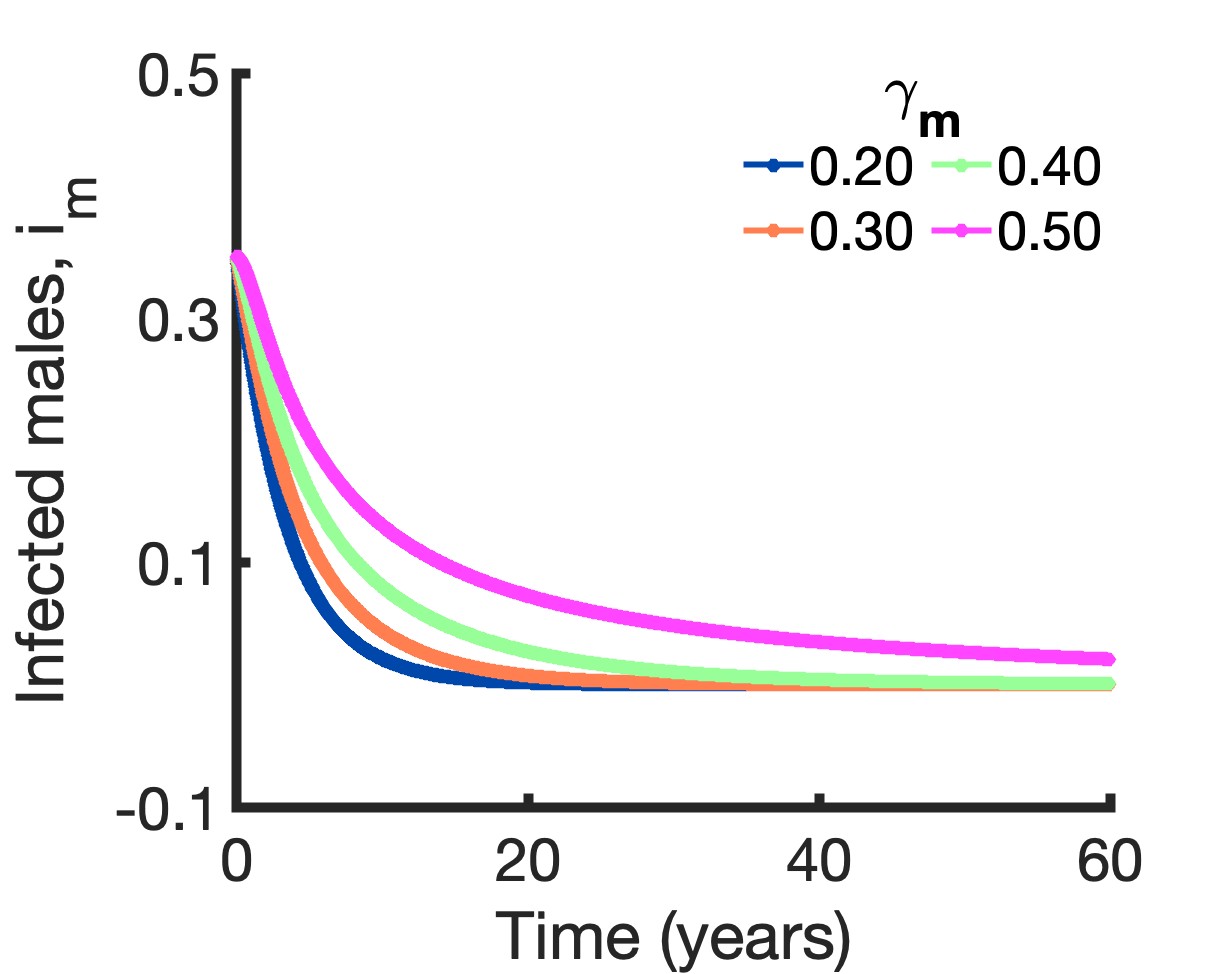}\label{fig:sub214}}
	\subfigure[$i_f$ vs $t$.]{\includegraphics[width=0.3\linewidth]{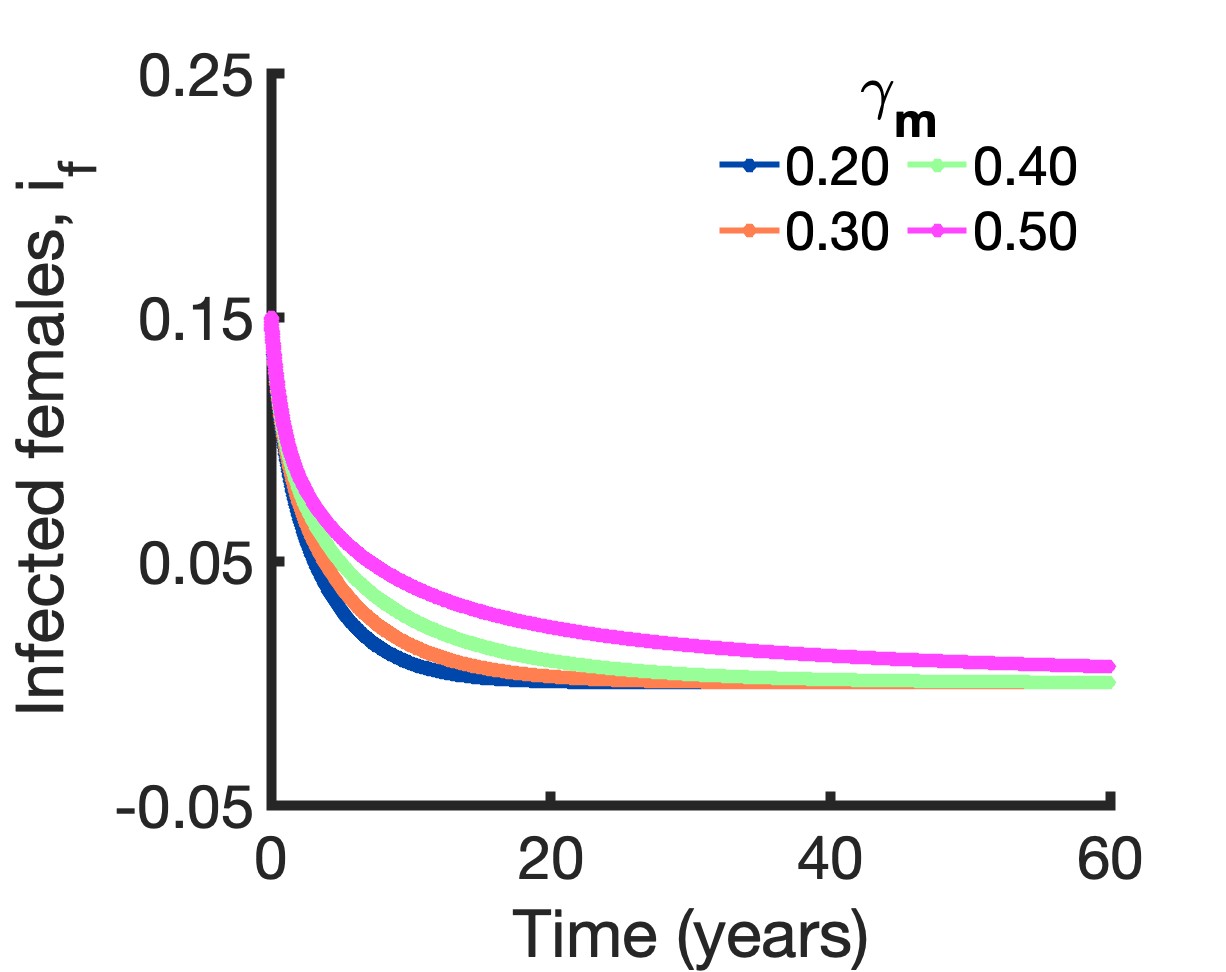}\label{fig:sub224}}
	\subfigure[$t_f$ vs $t$.]{\includegraphics[width=0.3\linewidth]{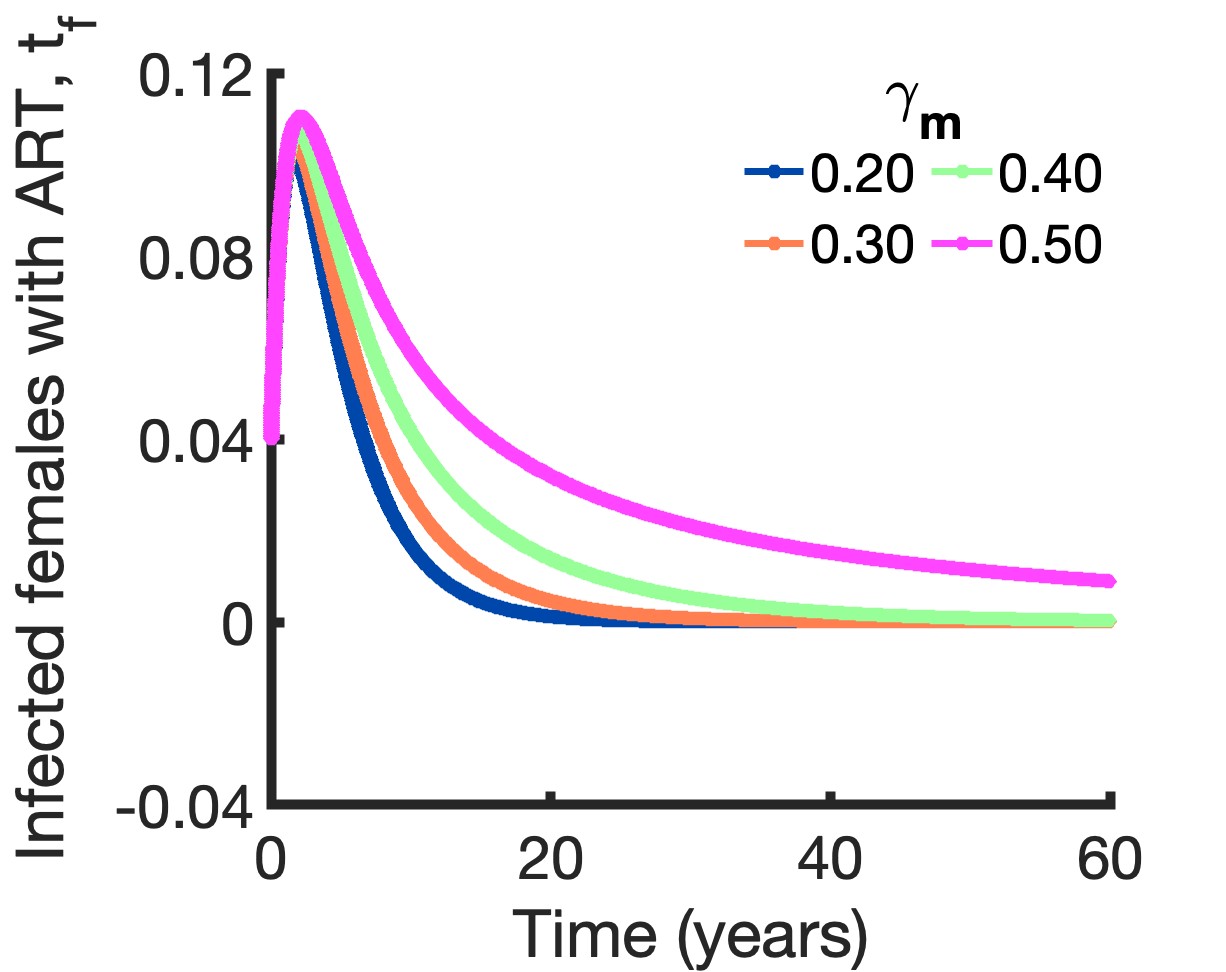}\label{fig:sub234}}
	\caption{Disease-Free Equilibrium (DFE) of the deterministic model for varying values of the probability of transmission by a male who is infected ($\gamma_m$).}
	\label{fig:DFE_Deterministic_gamma_m}
\end{figure}

\noindent Since we have already covered the disease-free equilibrium, it stands to reason that for certain parameter values, the proportion of infected males and females will approach zero. Therefore, the number of infected males and females would increase as the proportion of infected females receiving ART decreased (see figures \ref{fig:sub314} and \ref{fig:sub324}). It stands to reason from figure \ref{fig:sub334} that the number of infected females receiving ART will rise in direct proportion to the number of females receiving ART.

\begin{figure}[H]
	\centering
	\subfigure[$i_m$ vs $t$.]{\includegraphics[width=0.3\linewidth]{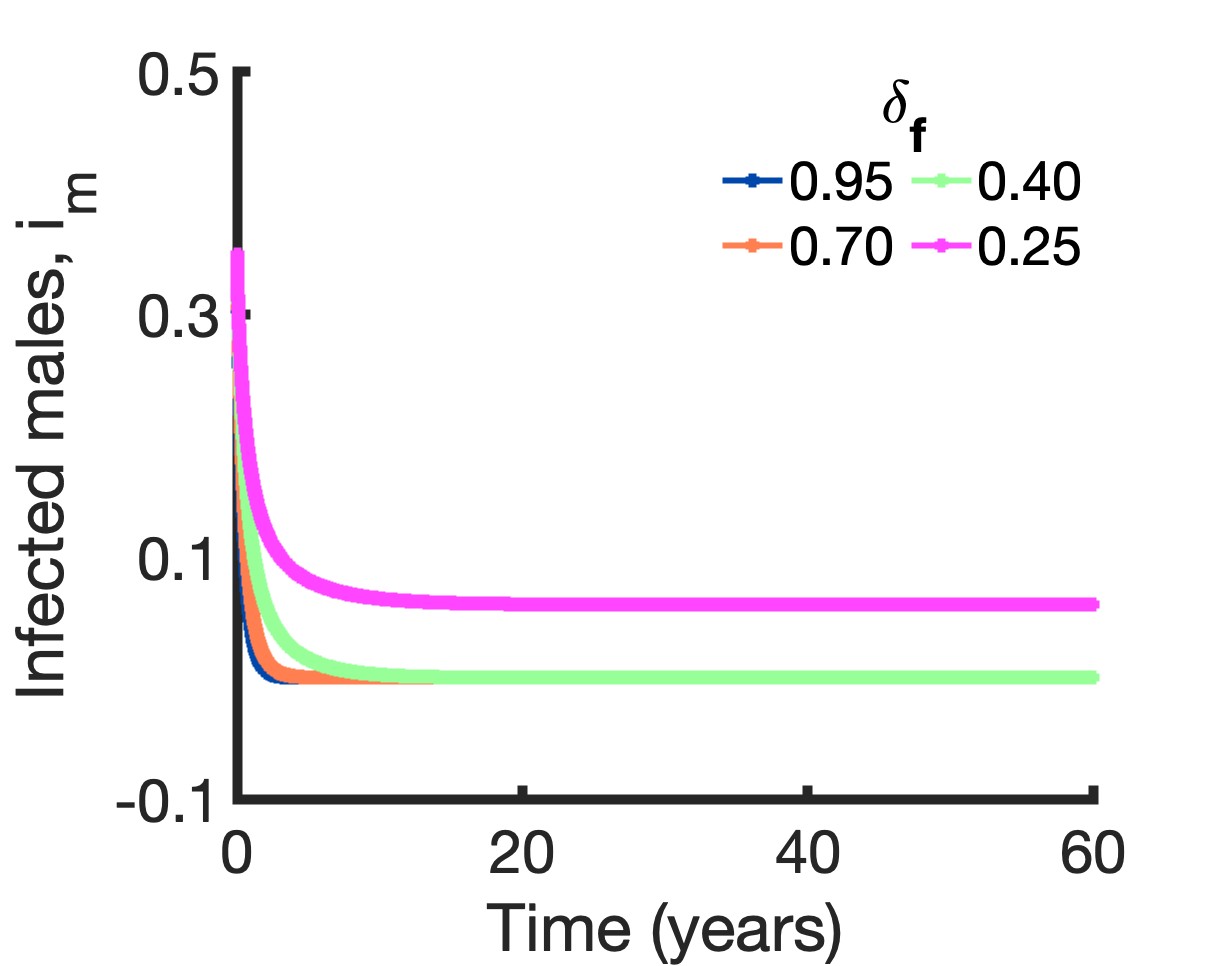}
	\label{fig:sub314}}
	\subfigure[$i_f$ vs $t$.]{\includegraphics[width=0.3\linewidth]{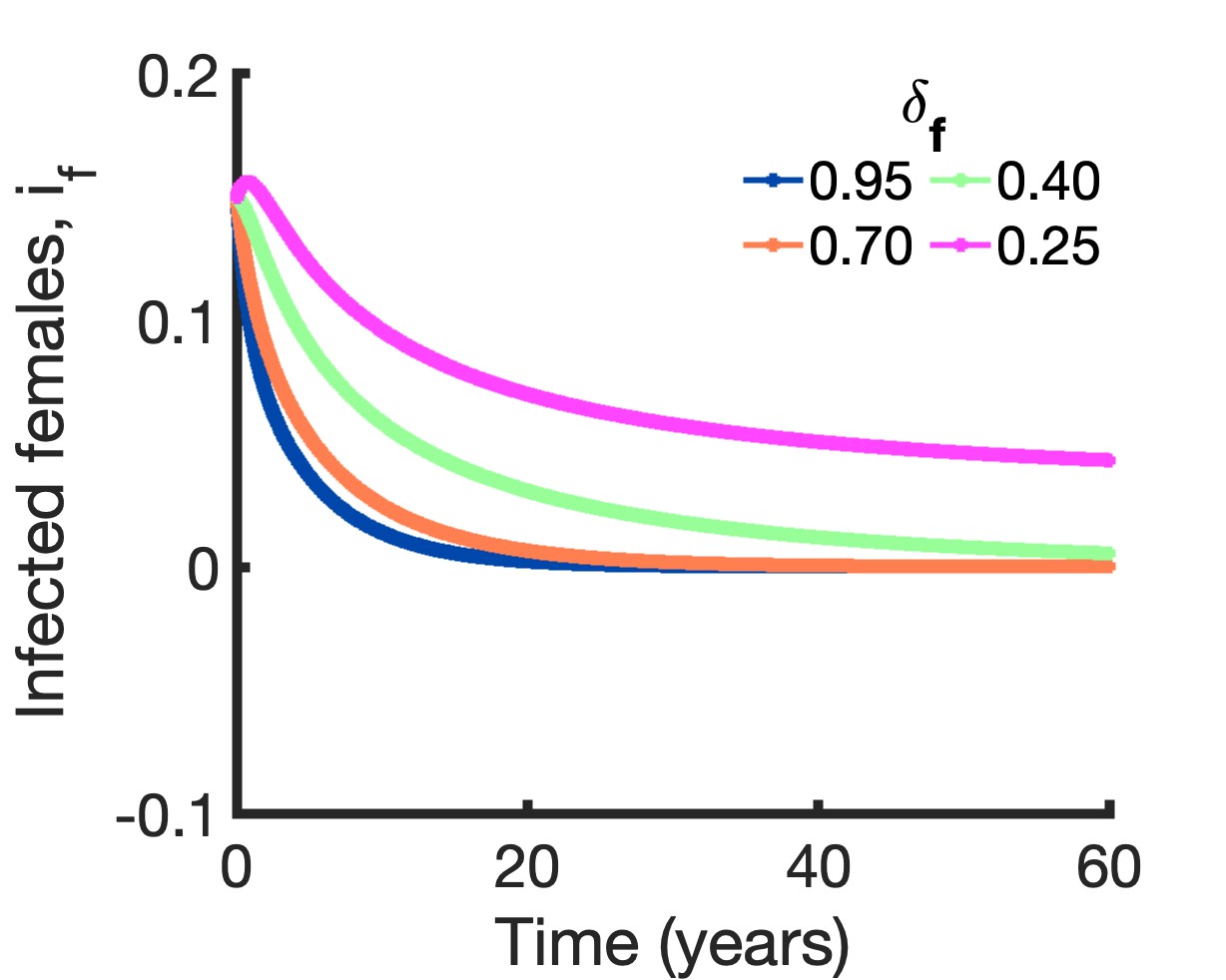}\label{fig:sub324}}
	\subfigure[$t_f$ vs $t$.]{\includegraphics[width=0.3\linewidth]{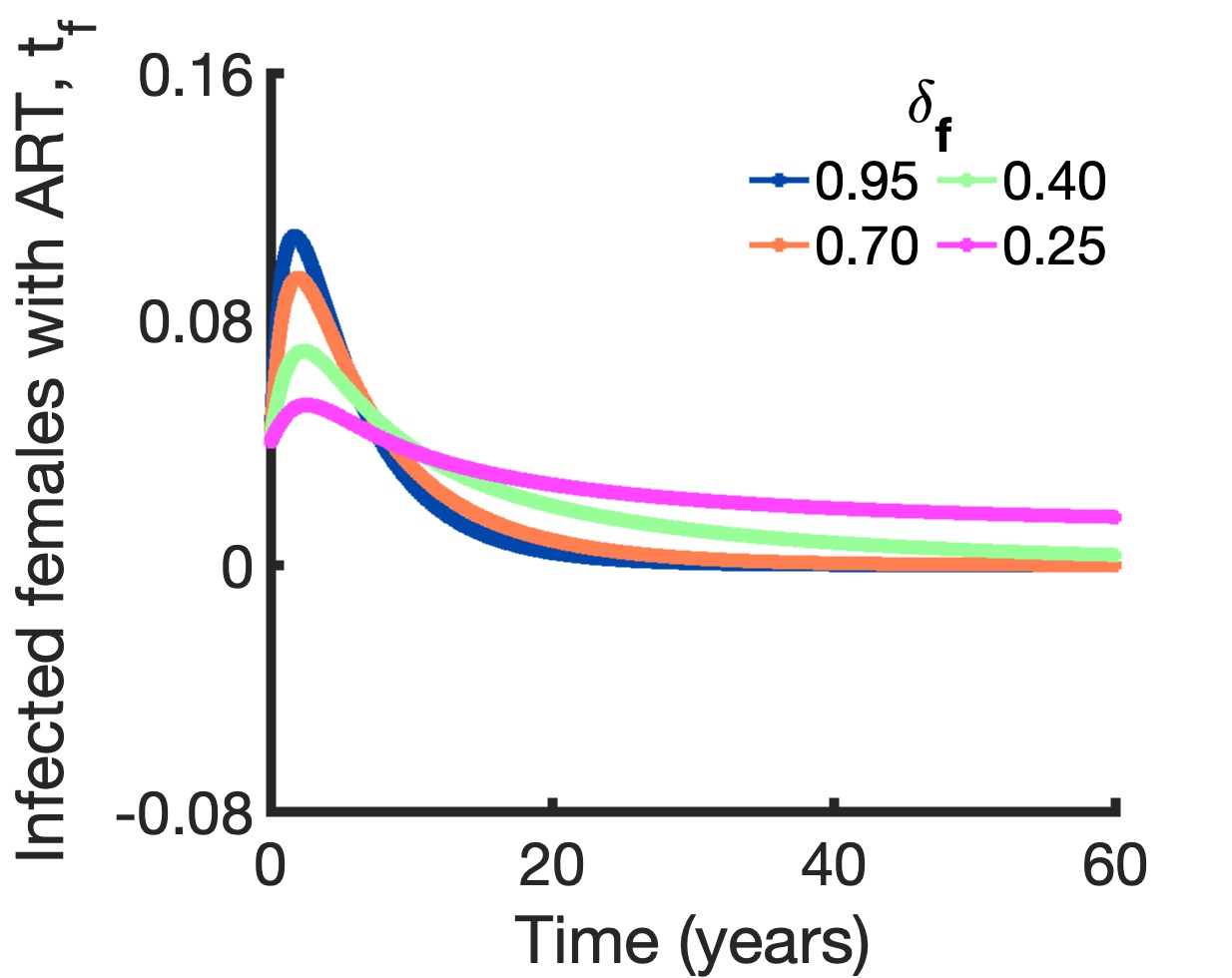}\label{fig:sub334}}
 \caption{Disease-Free Equilibrium (DFE) of the deterministic model for varying values of the portion of the infected females getting ART ($\delta_f$).}
 \label{fig:DFE_Deterministic_delta_f}
\end{figure}

\subsection{Stochastic Euler}
In this subsection, we introduced stochastic elements and employed the Euler method for numerical solutions. Similar to the previous subsection, we analyzed how the parameters $\gamma_f$, $\gamma_m$, and $\delta_f$ influence the three demographic groups ($i_m$, $i_f$, $t_f$). Figures \ref{fig:sub111}, \ref{fig:sub112}, \ref{fig:sub113}, \ref{fig:sub211}, \ref{fig:sub212}, \ref{fig:sub213}, \ref{fig:sub311}, \ref{fig:sub312}, and \ref{fig:sub313} depict similar parameter effects as observed in the deterministic model. The primary distinction here is the variability in the graphs, attributable to the stochastic nature of the Euler method.

\begin{figure}[H]
	\centering
	\subfigure[$i_m$ vs $t$.]{\includegraphics[width=0.3\linewidth]{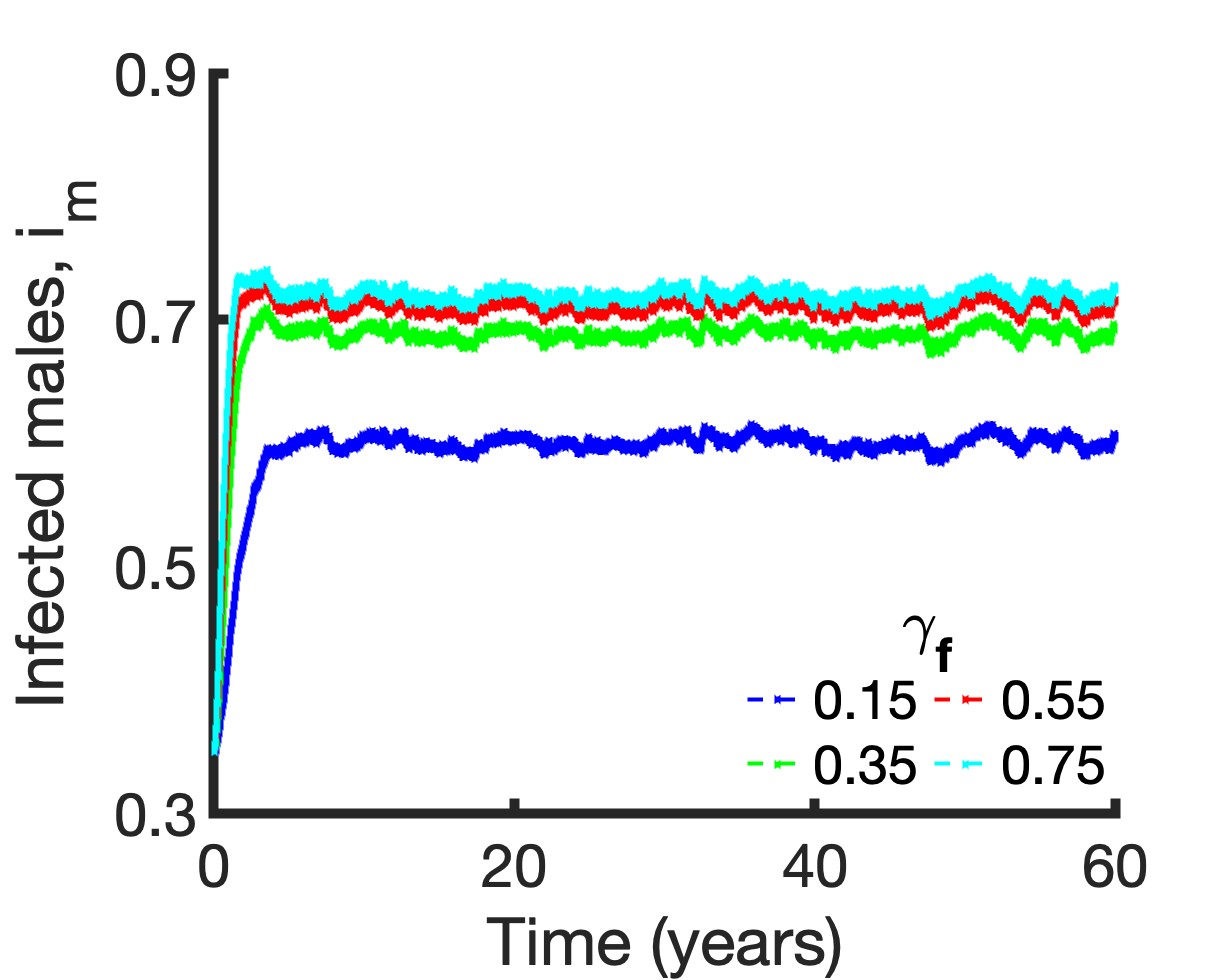}\label{fig:sub111}}
	\subfigure[$i_f$ vs $t$.]{\includegraphics[width=0.3\linewidth]{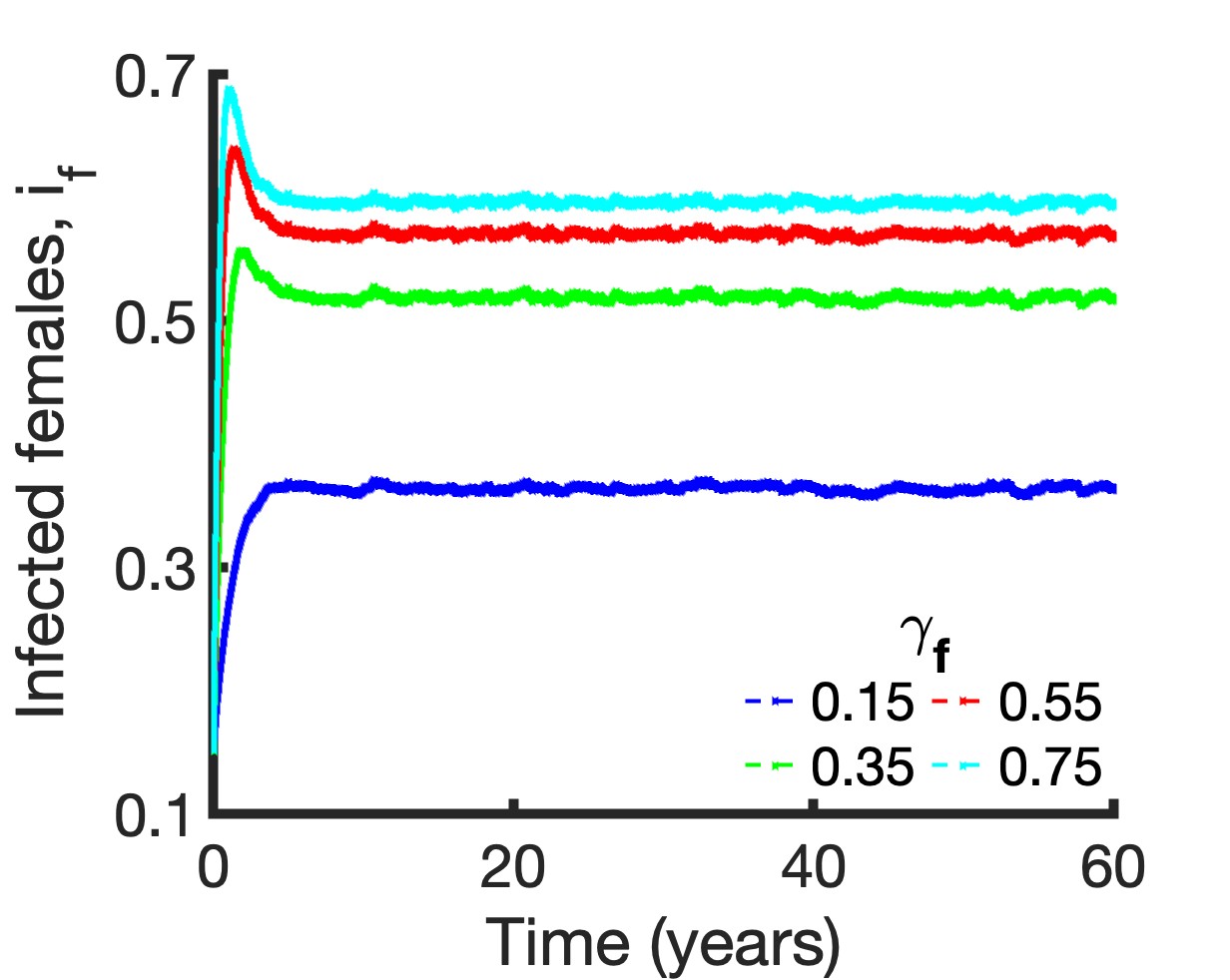}\label{fig:sub112}}
	\subfigure[$t_f$ vs $t$.]{\includegraphics[width=0.3\linewidth]{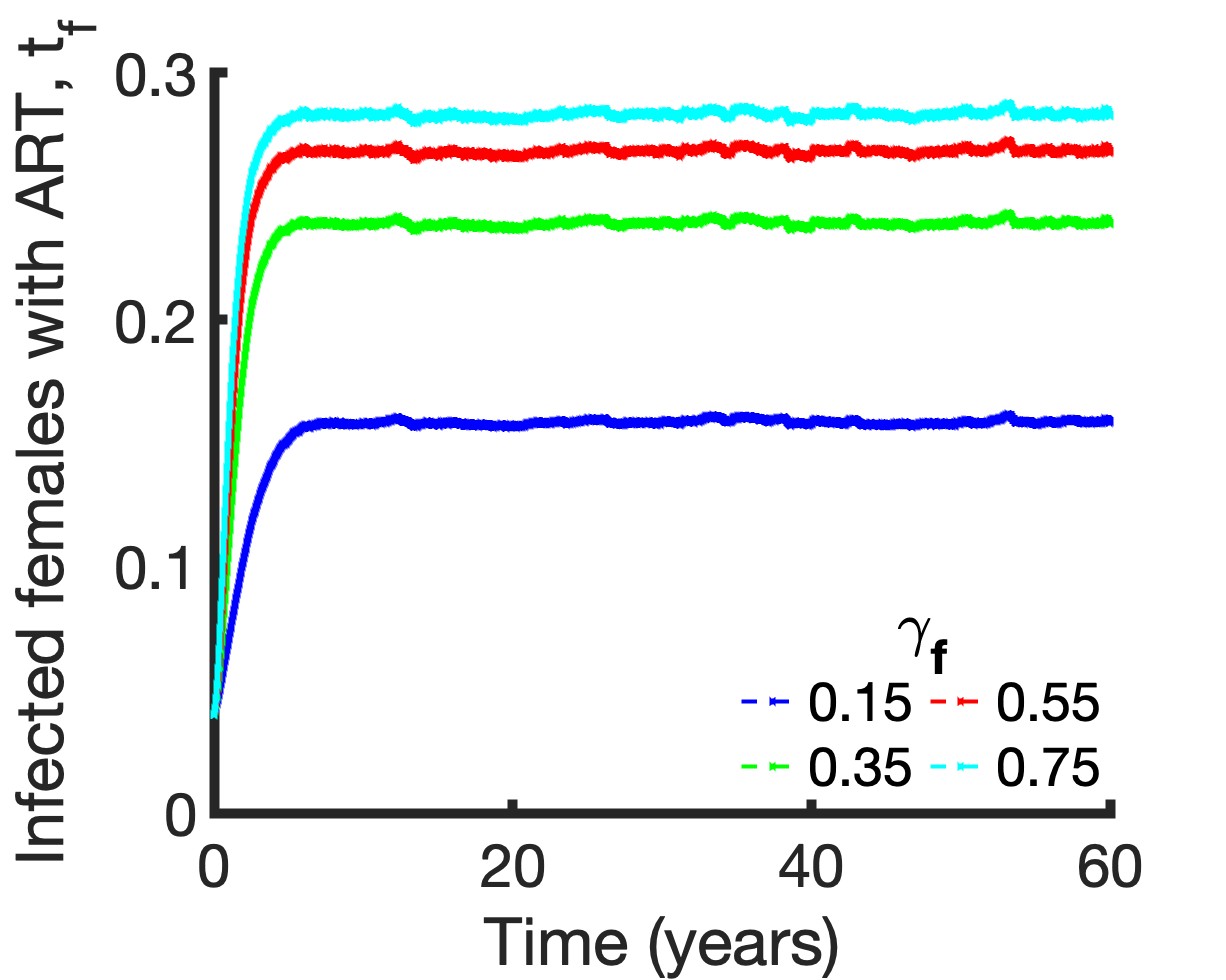}\label{fig:sub113}}
	\caption{Endemic Equilibrium (EE) of the stochastic Euler model for varying values of the probability of transmission by a female who is infected  ($\gamma_f$).}
	\label{fig:EE_Euler_gamma_f}
\end{figure}

\begin{figure}[H]
	\centering
	\subfigure[$i_m$ vs $t$.]{\includegraphics[width=0.3\linewidth]{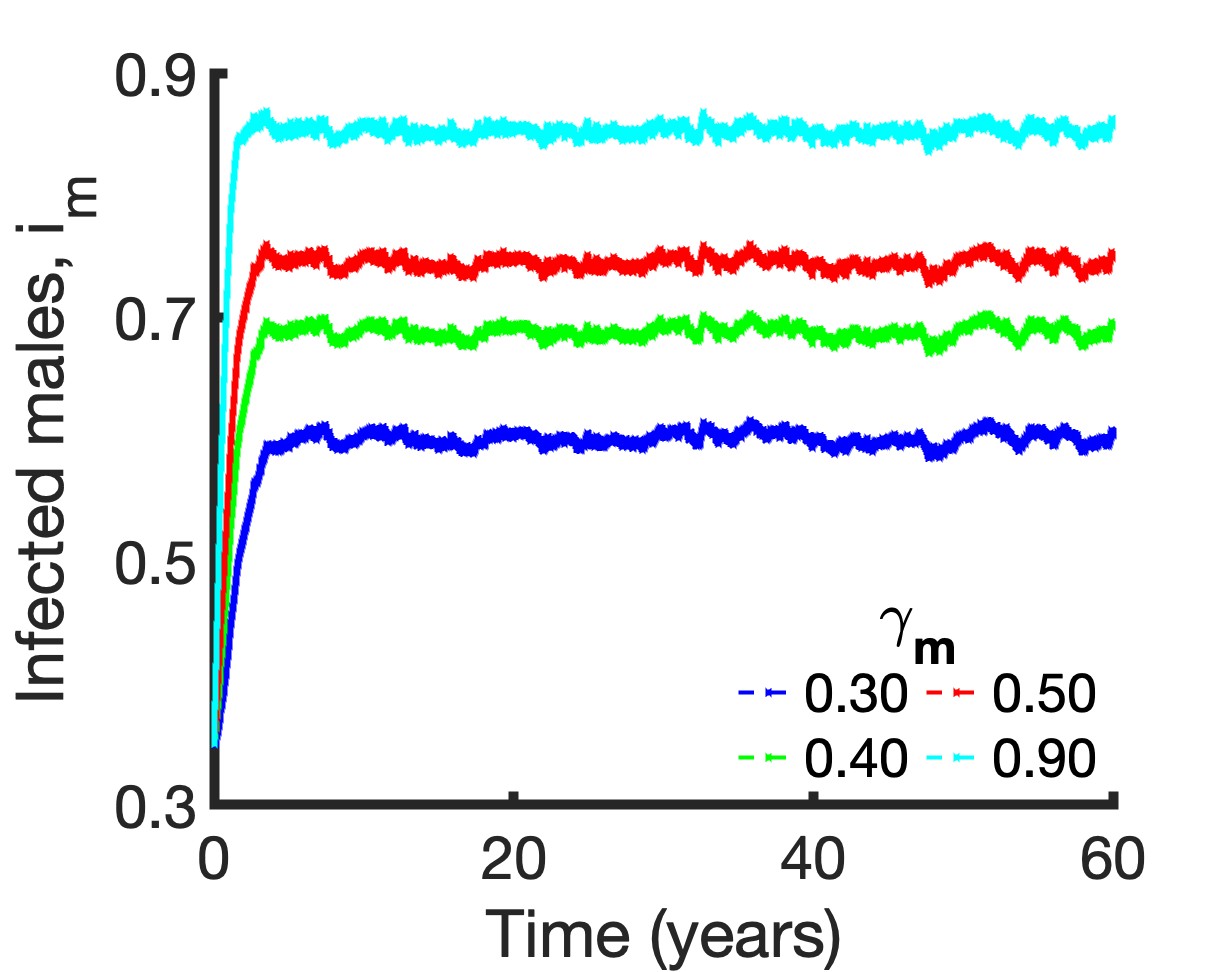}\label{fig:sub211}}
	\subfigure[$i_f$ vs $t$.]{\includegraphics[width=0.3\linewidth]{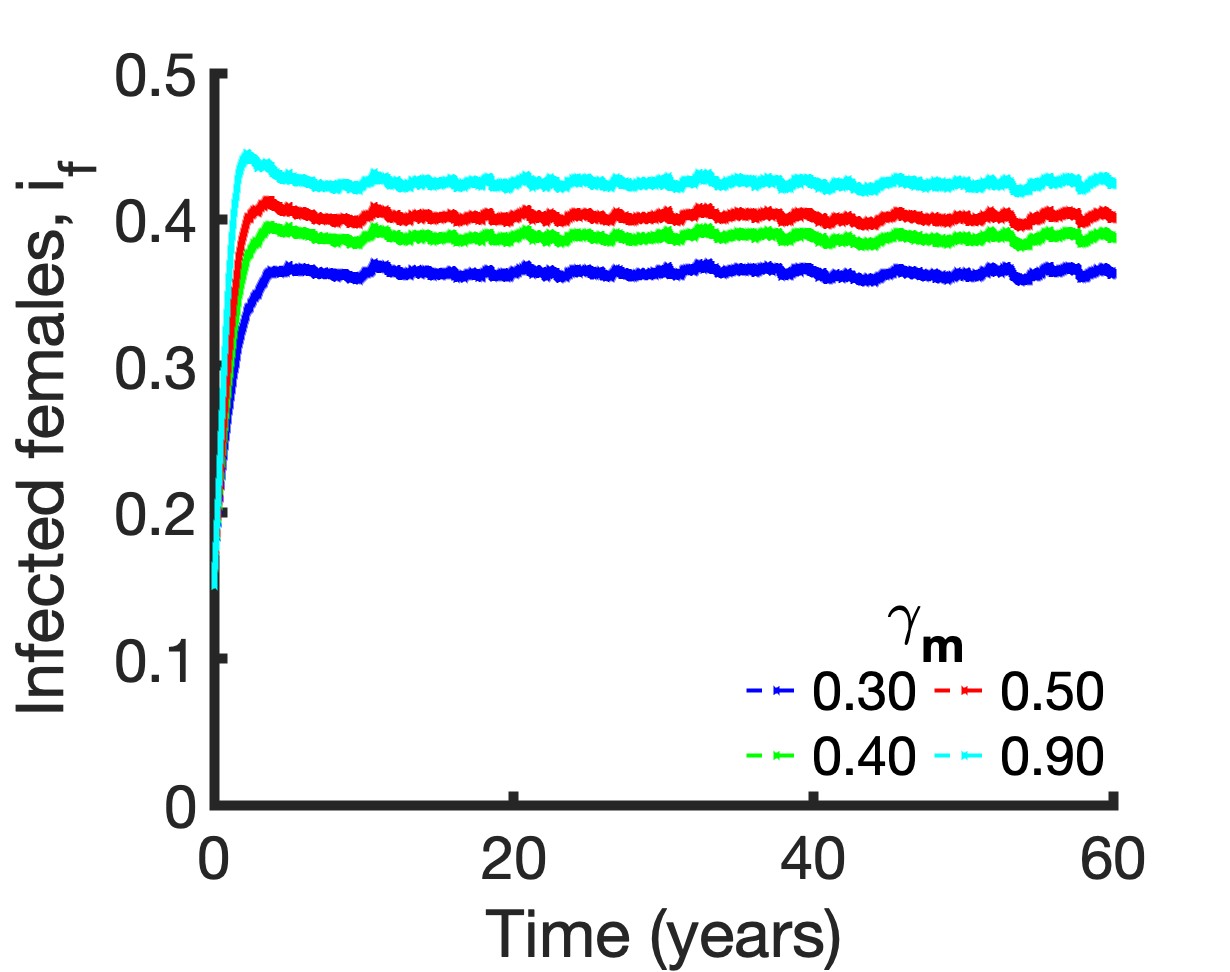}\label{fig:sub212}}
	\subfigure[$t_f$ vs $t$.]{\includegraphics[width=0.3\linewidth]{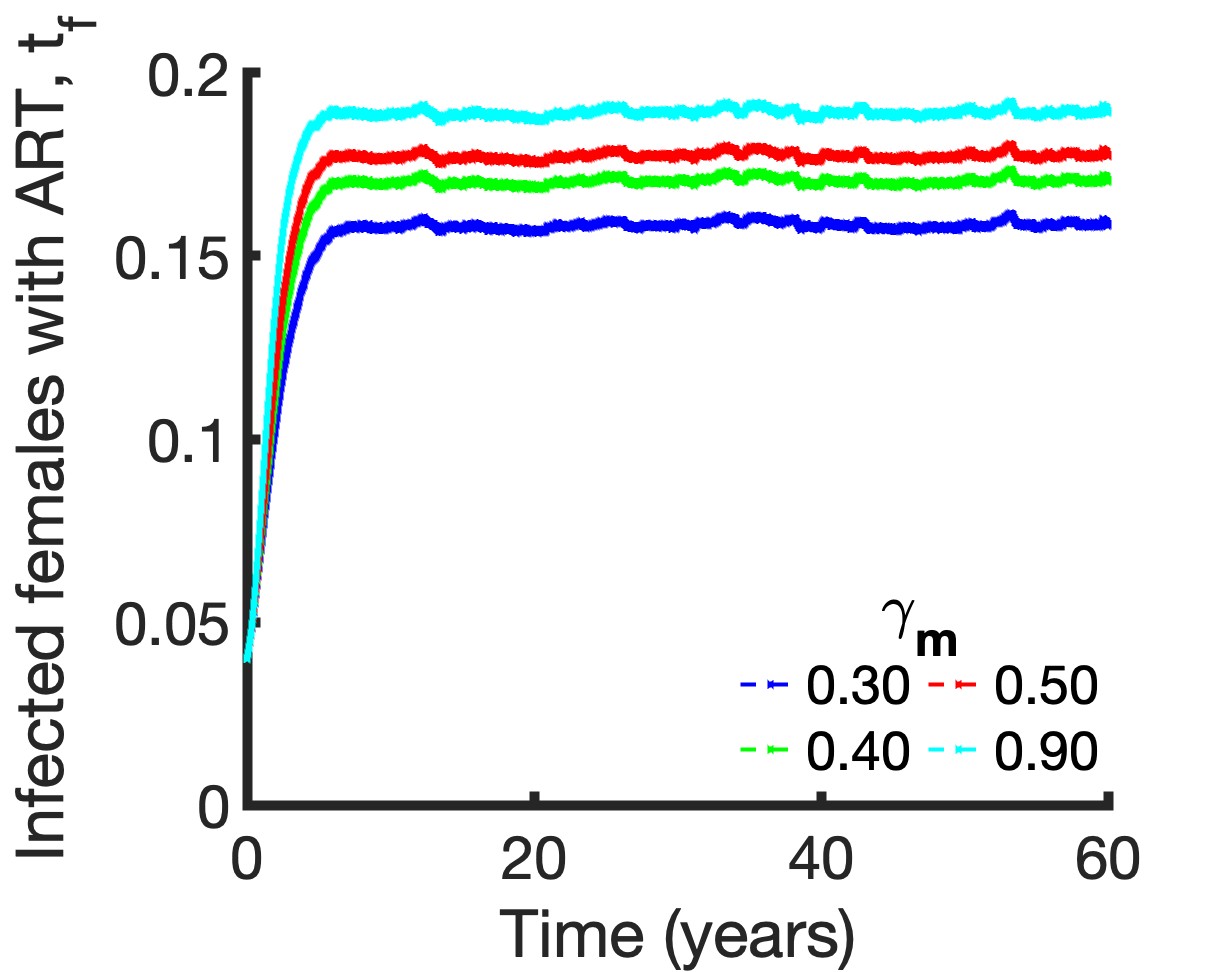}\label{fig:sub213}}
	\caption{Endemic Equilibrium (EE) of the stochastic Euler model for varying values of the probability of transmission by a male who is infected  ($\gamma_m$).}
	\label{fig:EE_Euler_gamma_m}
\end{figure}

\begin{figure}[H]
	\centering
	\subfigure[$i_m$ vs $t$.]{\includegraphics[width=0.3\linewidth]{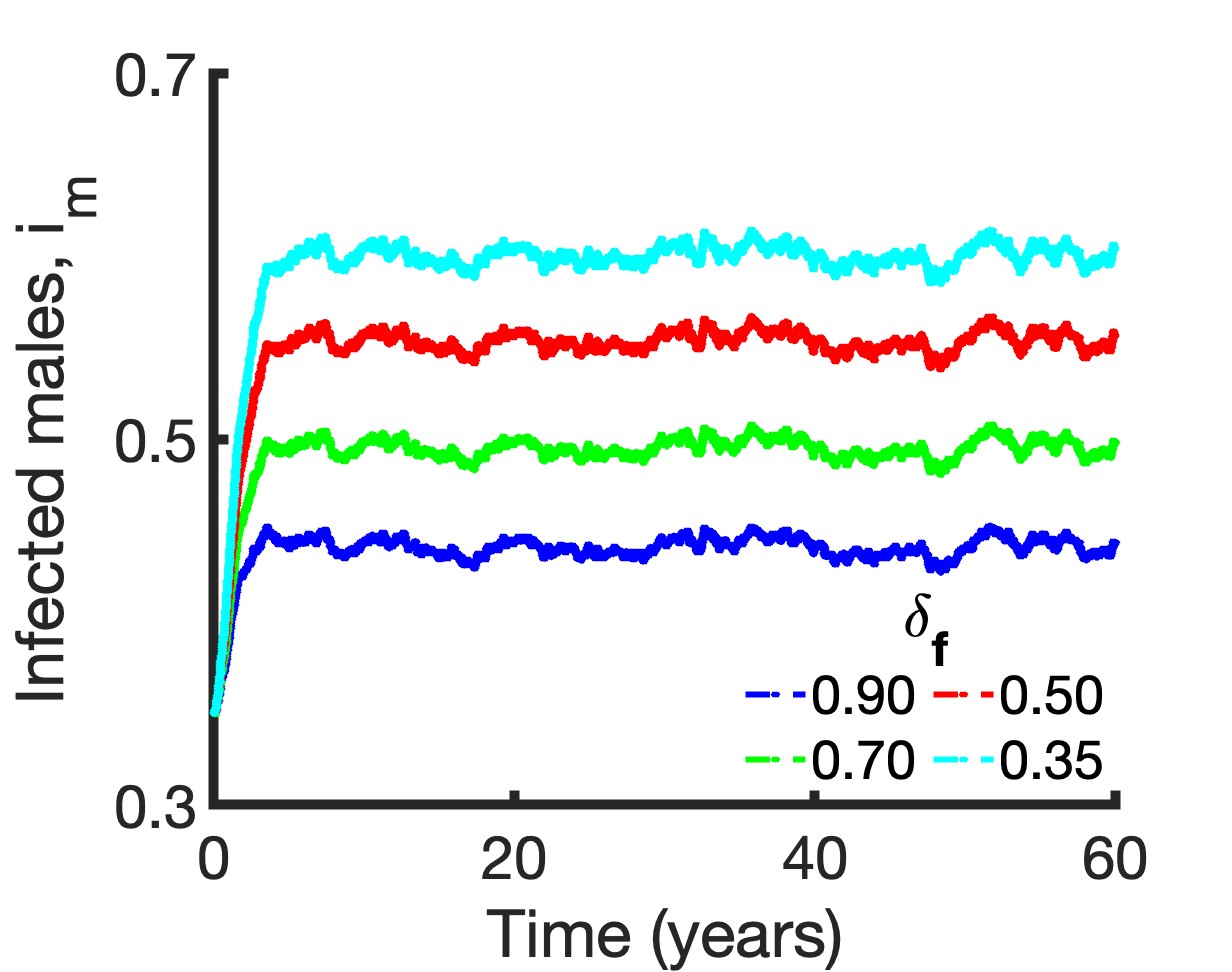}\label{fig:sub311}}
	\subfigure[$i_f$ vs $t$.]{\includegraphics[width=0.3\linewidth]{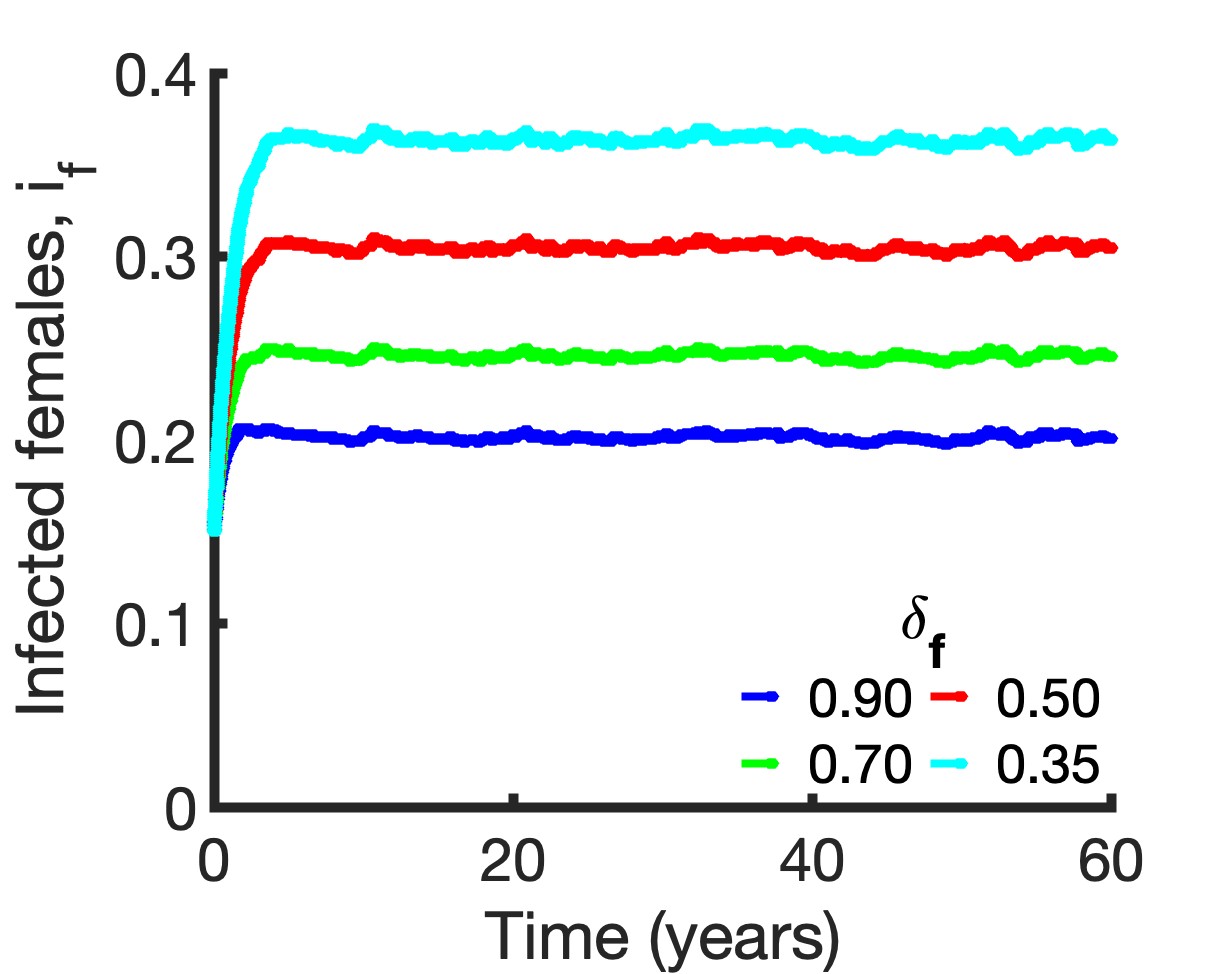}\label{fig:sub312}}
	\subfigure[$t_f$ vs $t$.]{\includegraphics[width=0.3\linewidth]{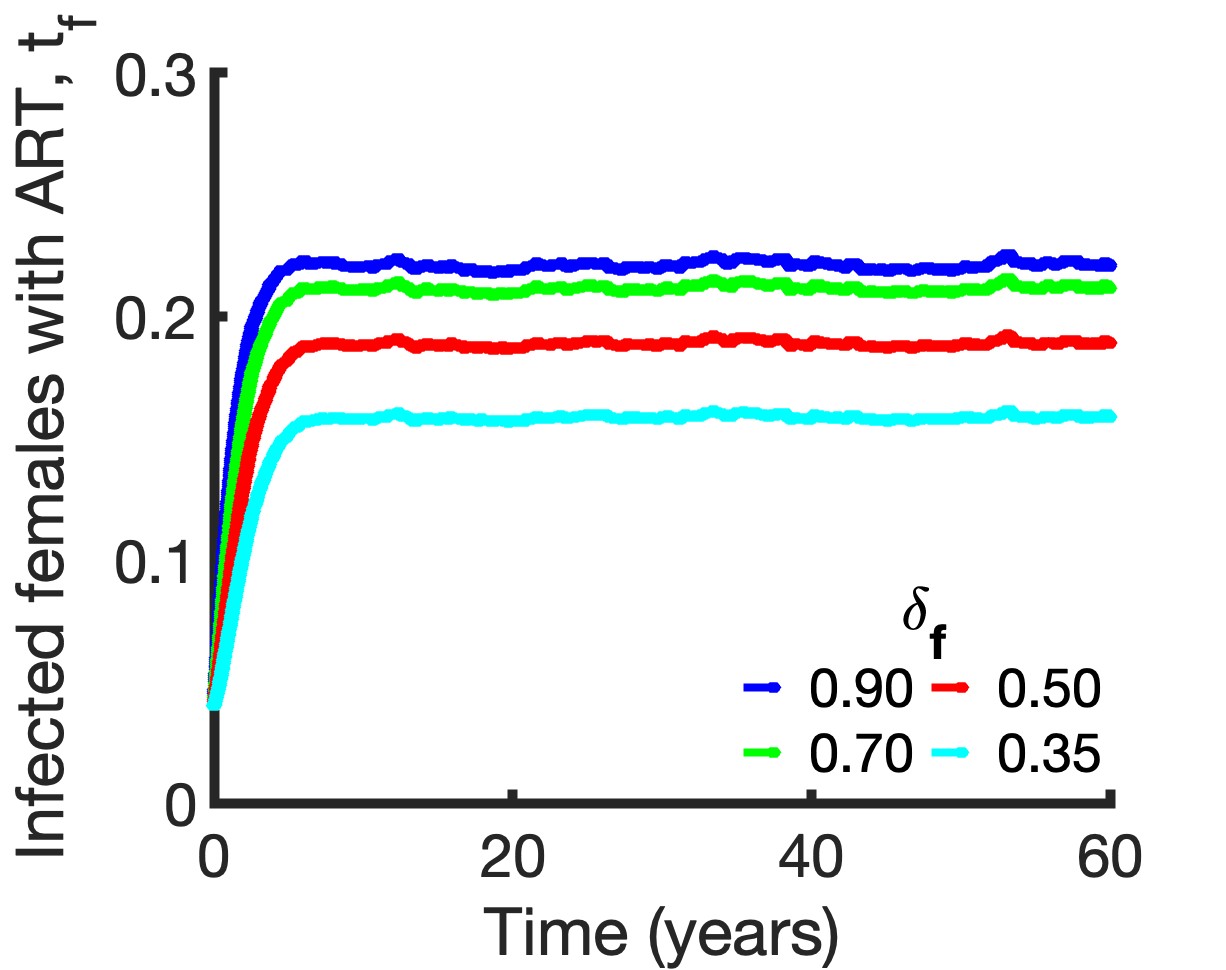}\label{fig:sub313}}
	\caption{Endemic Equilibrium (EE) of the stochastic Euler model for varying values of the portion of the infected females getting ART ($\delta_f$).}
	\label{fig:EE_Euler_delta_f}
\end{figure}

According to figure \ref{fig:sub1111}, the Disease-Free Equilibrium (DFE) of the stochastic Euler model for varying values of the probability of transmission by a female who is infected ($\gamma_f$) shows that as $\gamma_f$ decreases, the number of infected males ($i_m$) decreases. Figure \ref{fig:sub1121} similarly illustrates that a lower $\gamma_f$ leads to a reduction in the number of infected females ($i_f$). Figure \ref{fig:sub1131} demonstrates that the number of infected females receiving antiretroviral therapy (ART) ($t_f$) also decreases with a lower $\gamma_f$, highlighting the interconnected dynamics of gender-specific transmission probabilities.

\begin{figure}[H]
	\centering
	\subfigure[$i_m$ vs $t$.]{\includegraphics[width=0.4\linewidth]{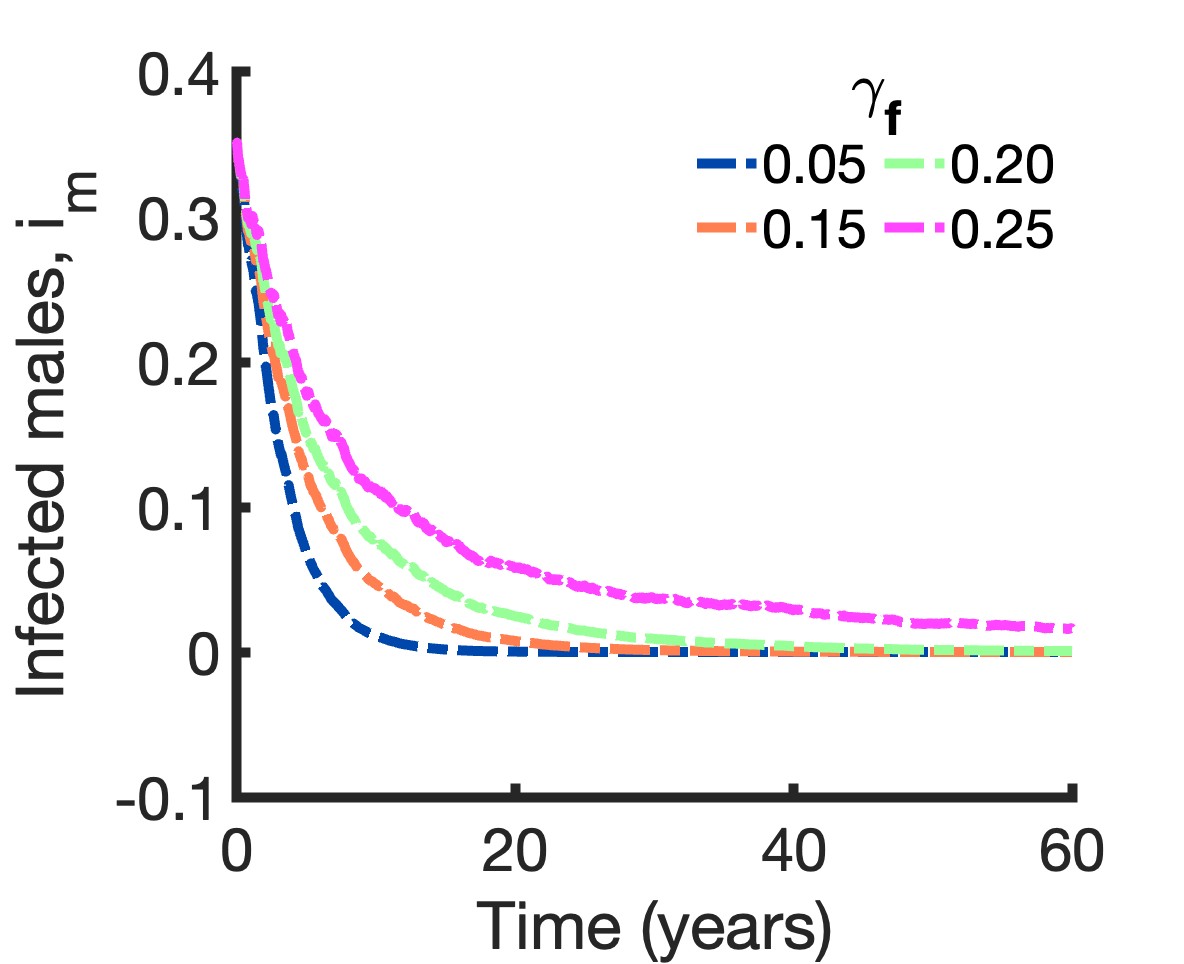}\label{fig:sub1111}}
	\subfigure[$i_f$ vs $t$.]{\includegraphics[width=0.4\linewidth]{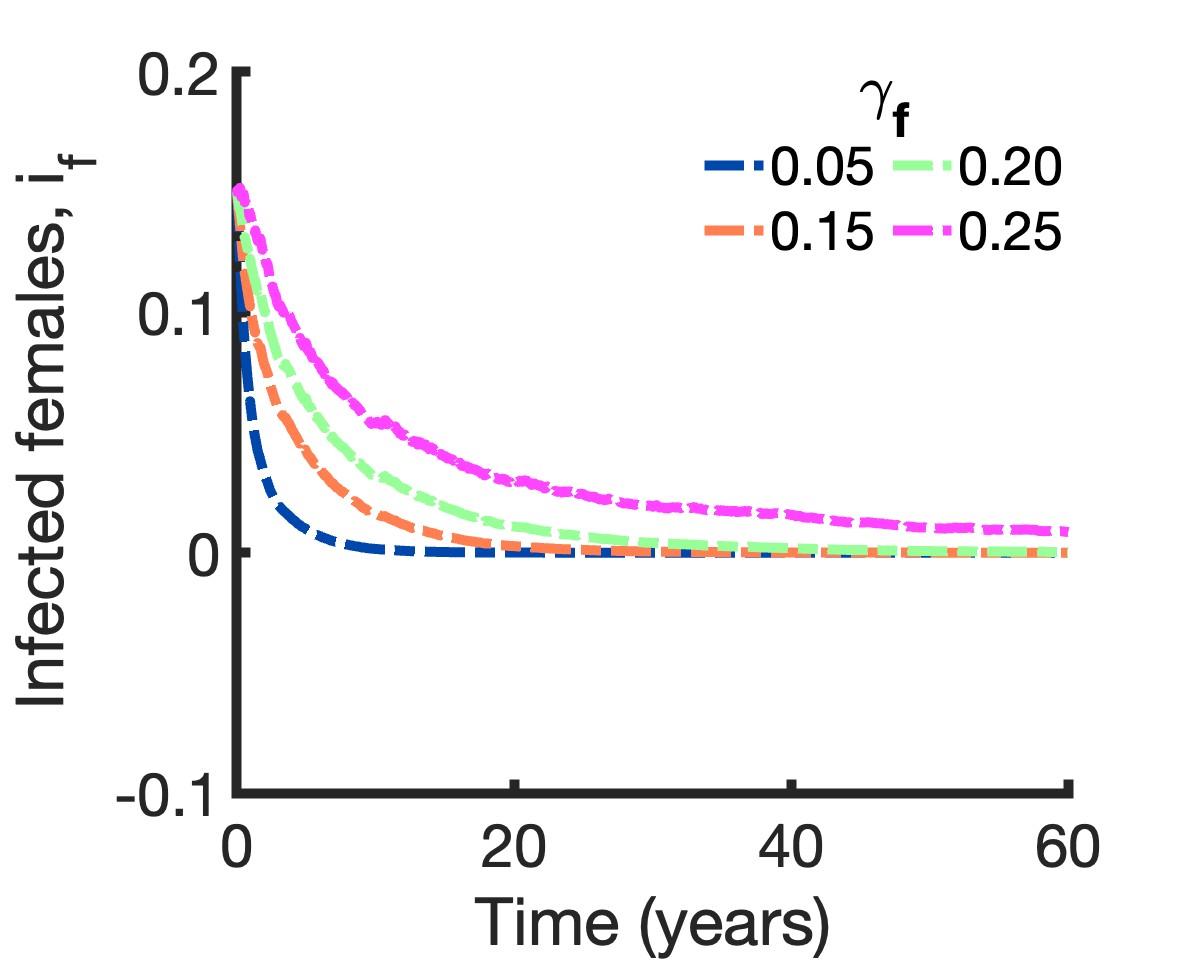}\label{fig:sub1121}}
	\subfigure[$t_f$ vs $t$.]{\includegraphics[width=0.4\linewidth]{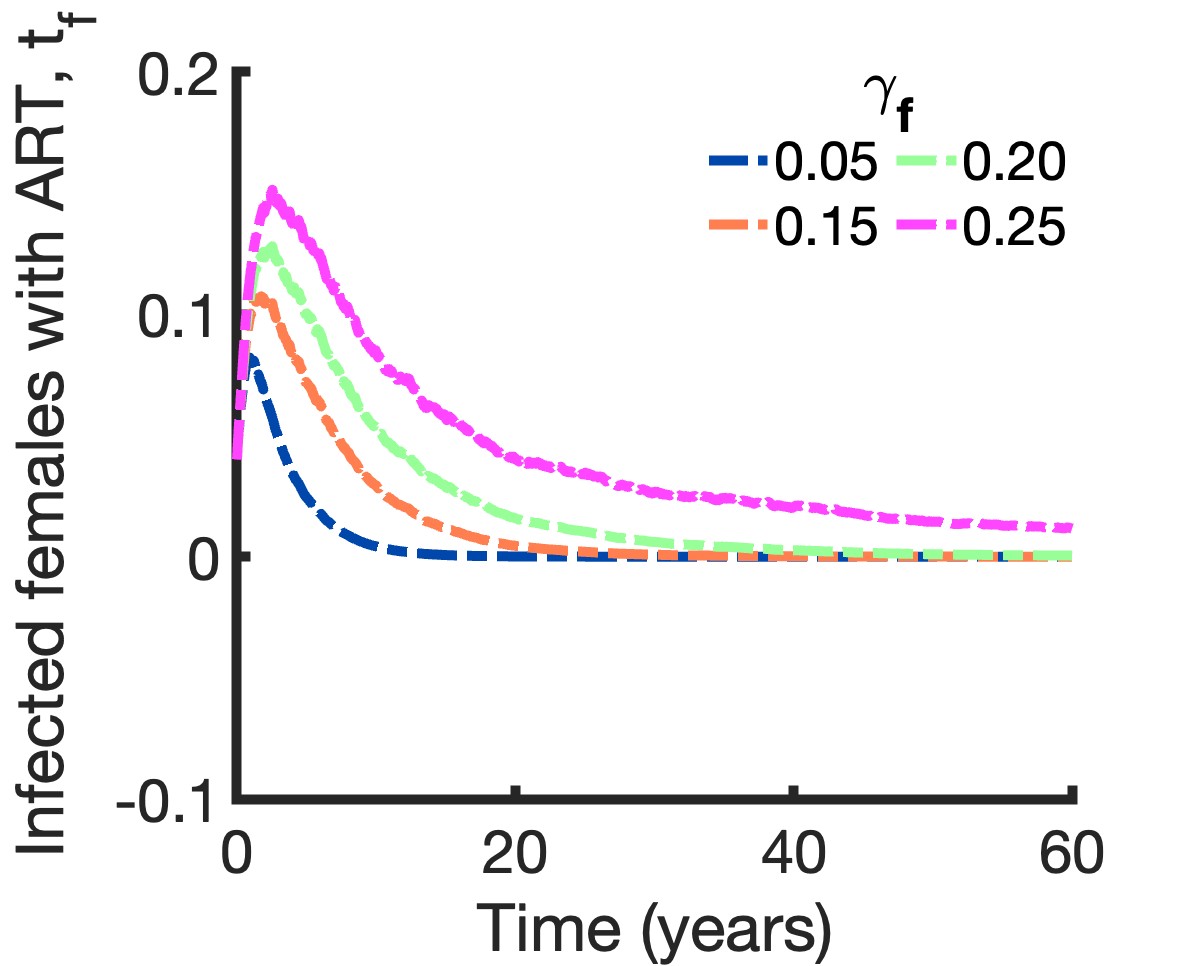}\label{fig:sub1131}}
	\caption{Disease-Free Equilibrium (DFE) of the stochastic Euler model for varying values of the probability of transmission by a female who is infected  ($\gamma_f$).}
	\label{fig:DFE_Euler_gamma_f}
\end{figure}

\noindent Figure \ref{fig:sub2111} illustrates the Disease-Free Equilibrium (DFE) of the stochastic Euler model for varying values of the probability of transmission by a male who is infected ($\gamma_m$). As $\gamma_m$ decreases, the number of infected males ($i_m$) decreases significantly. Figure \ref{fig:sub2121} shows a corresponding decrease in the number of infected females ($i_f$) with lower $\gamma_m$. Figure \ref{fig:sub2131} indicates that the number of infected females receiving ART ($t_f$) also decreases as $\gamma_m$ is reduced, emphasizing the importance of targeting male transmission rates to control the epidemic.

\begin{figure}[H]
	\centering
	\subfigure[$i_m$ vs $t$.]{\includegraphics[width=0.4\linewidth]{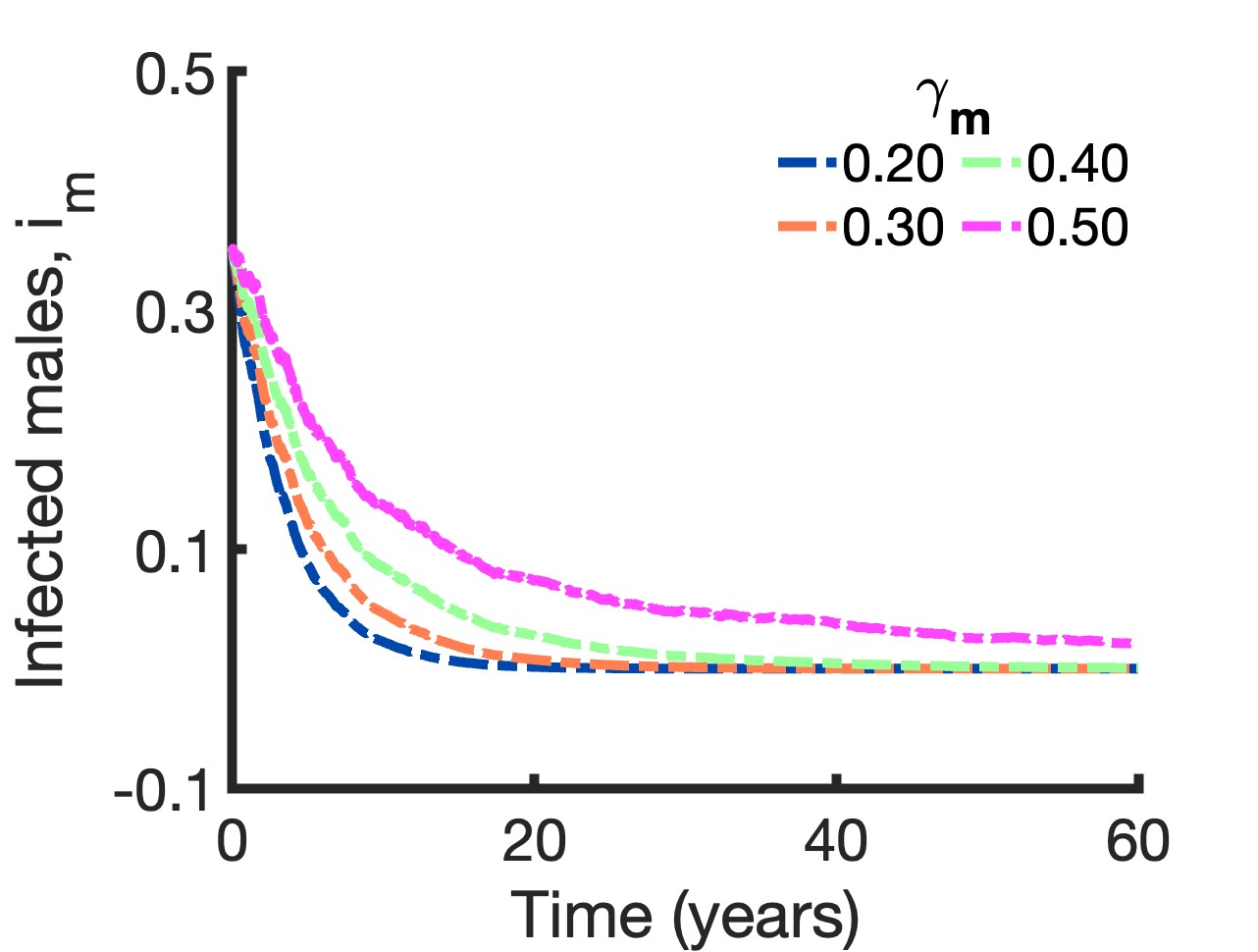}\label{fig:sub2111}}
	\subfigure[$i_f$ vs $t$.]{\includegraphics[width=0.4\linewidth]{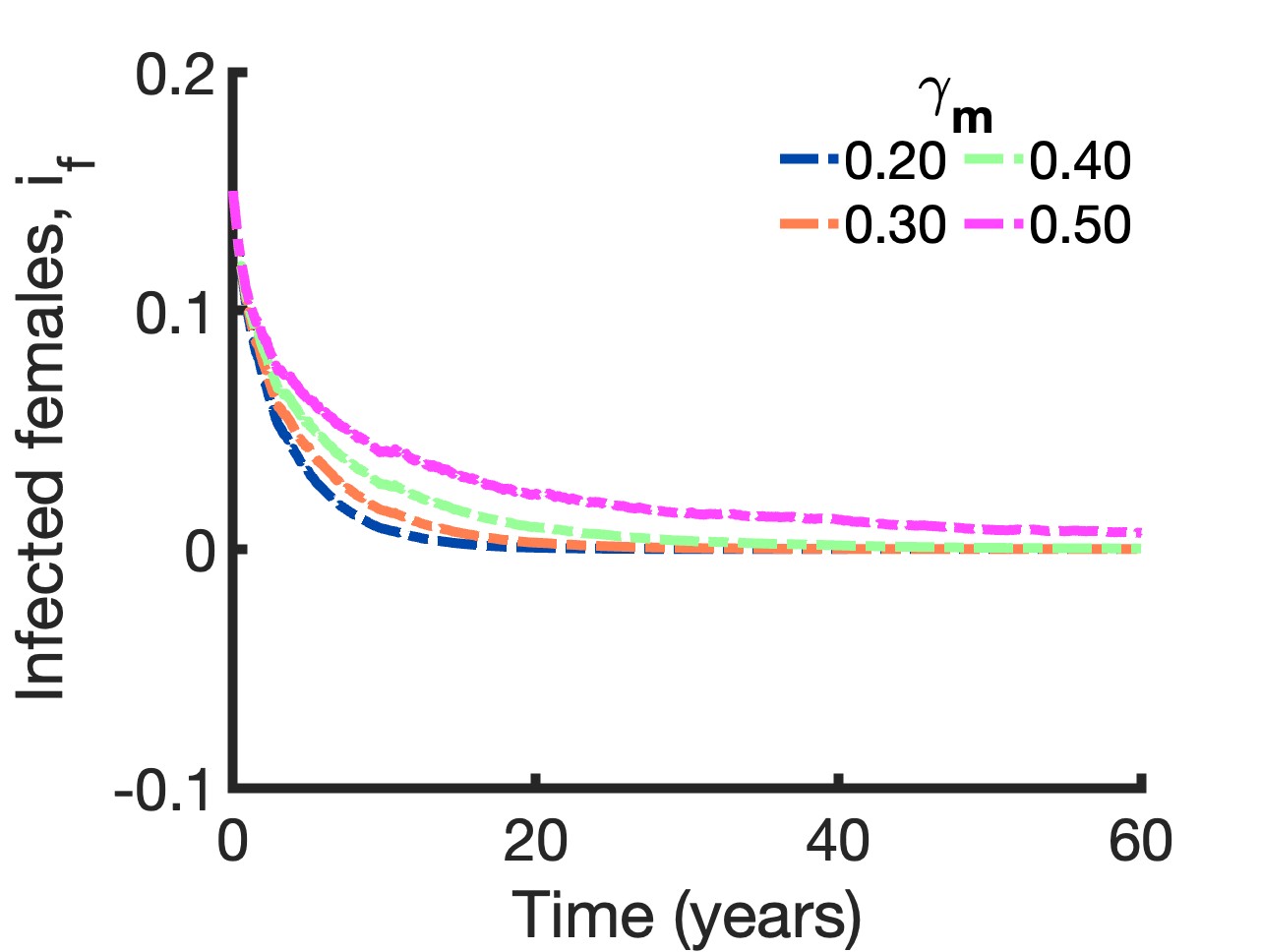}\label{fig:sub2121}}
	\subfigure[$t_f$ vs $t$.]{\includegraphics[width=0.4\linewidth]{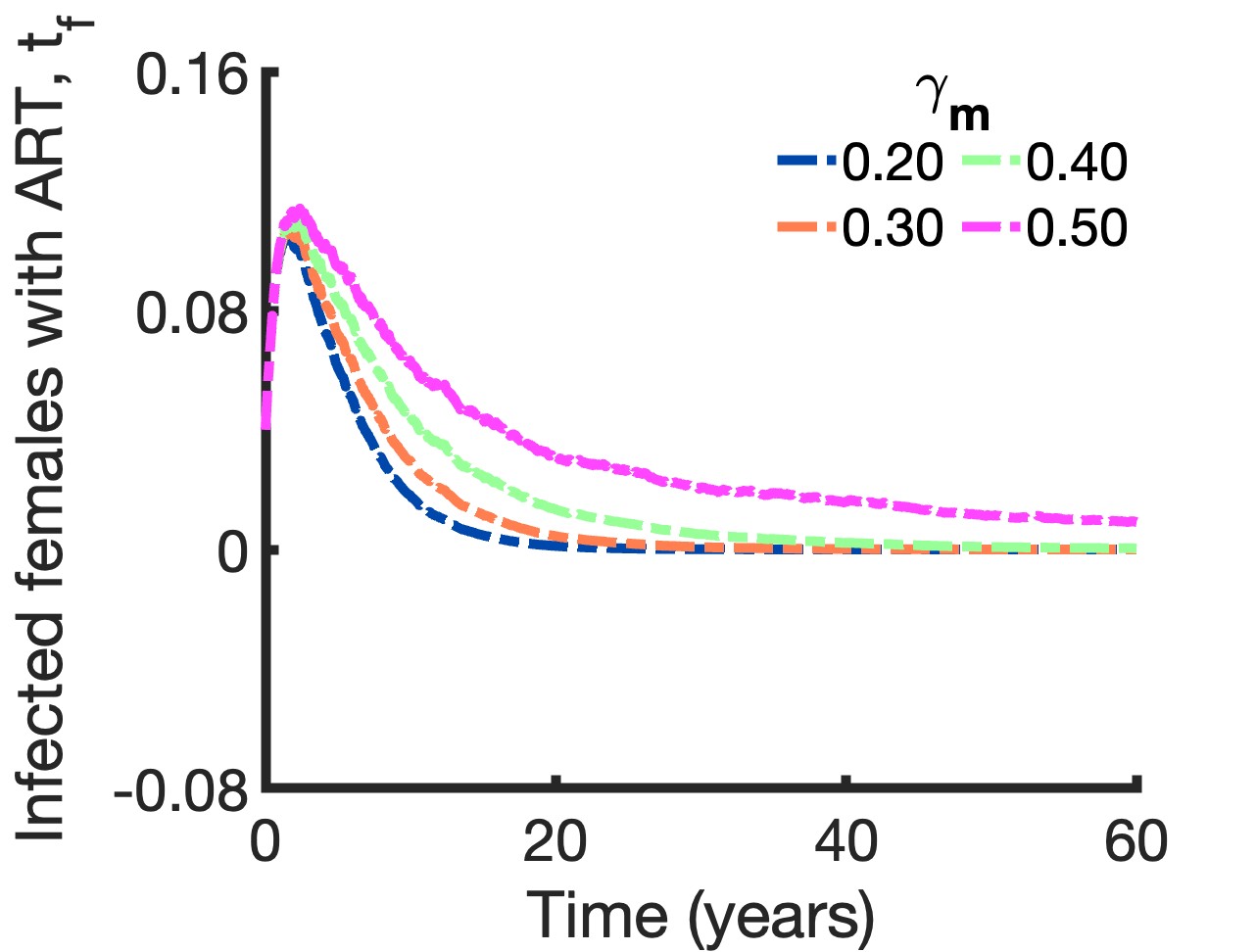}\label{fig:sub2131}}
	\caption{Disease-Free Equilibrium (DFE) of the stochastic Euler model for varying values of the probability of transmission by a male who is infected  ($\gamma_m$).}
	\label{fig:DFE_Euler_gamma_m}
\end{figure}

\noindent In figure \ref{fig:sub3111}, the Disease-Free Equilibrium (DFE) of the stochastic Euler model for varying values of the portion of infected females receiving ART ($\delta_f$) shows that increasing $\delta_f$ leads to a decrease in the number of infected males ($i_m$). Figure \ref{fig:sub3121} demonstrates that higher $\delta_f$ results in fewer infected females ($i_f$). Figure \ref{fig:sub3131} highlights that the number of infected females receiving ART ($t_f$) increases with higher $\delta_f$, underscoring the critical role of ART coverage in reducing overall infection rates and managing the epidemic.

\begin{figure}[H]
	\centering
	\subfigure[$i_m$ vs $t$.]{\includegraphics[width=0.4\linewidth]{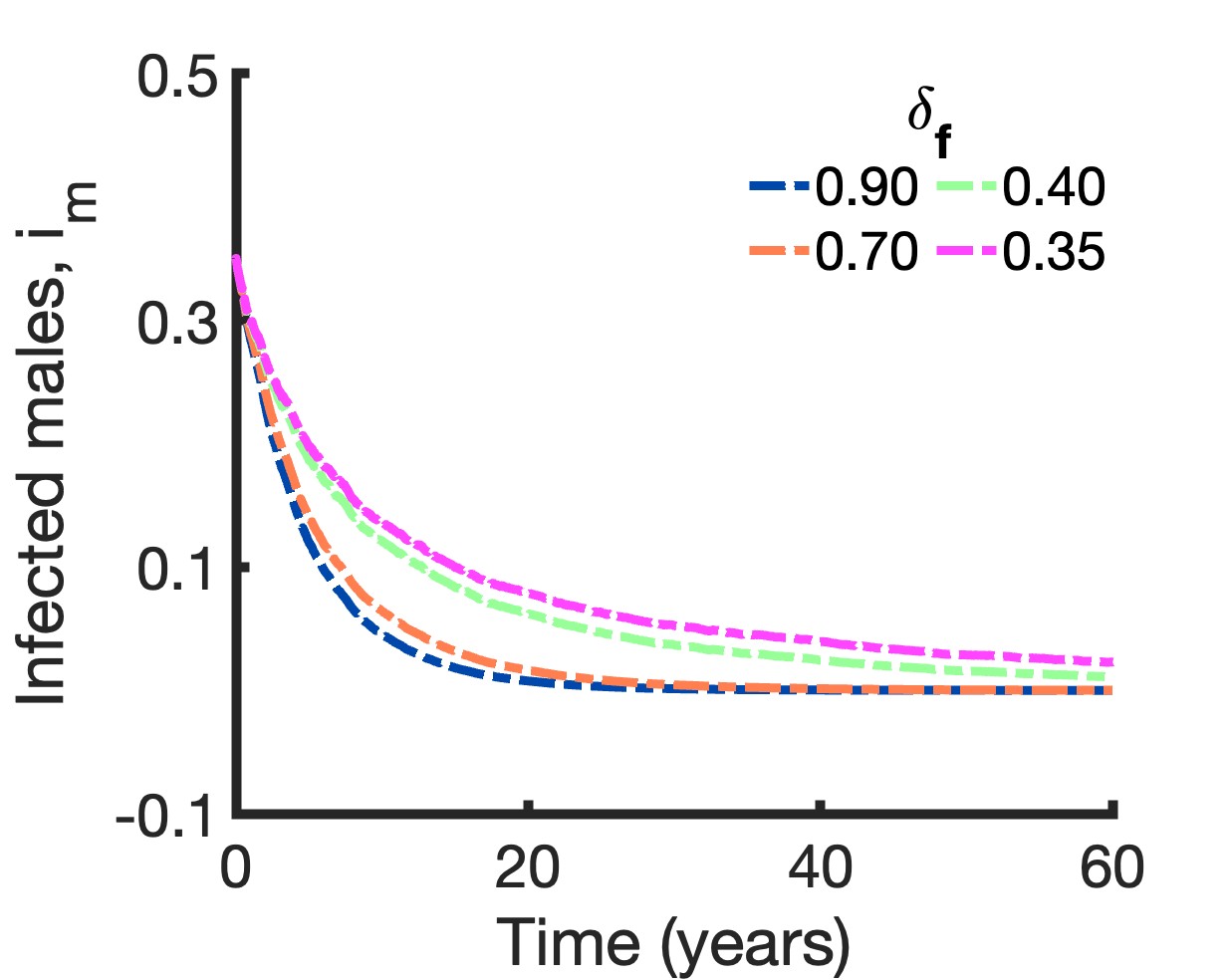}\label{fig:sub3111}}
	\subfigure[$i_f$ vs $t$.]{\includegraphics[width=0.4\linewidth]{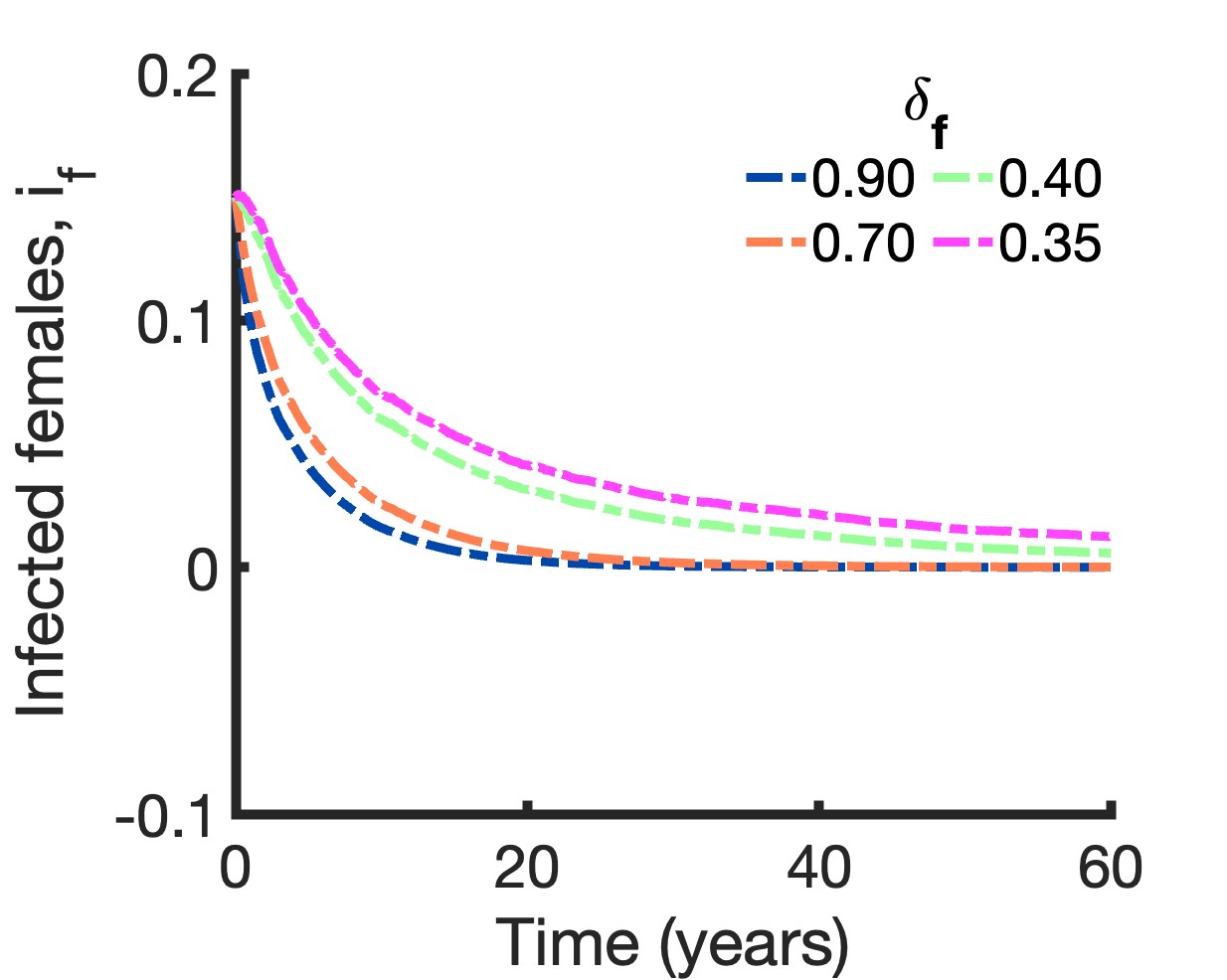}\label{fig:sub3121}}
	\subfigure[$t_f$ vs $t$.]{\includegraphics[width=0.4\linewidth]{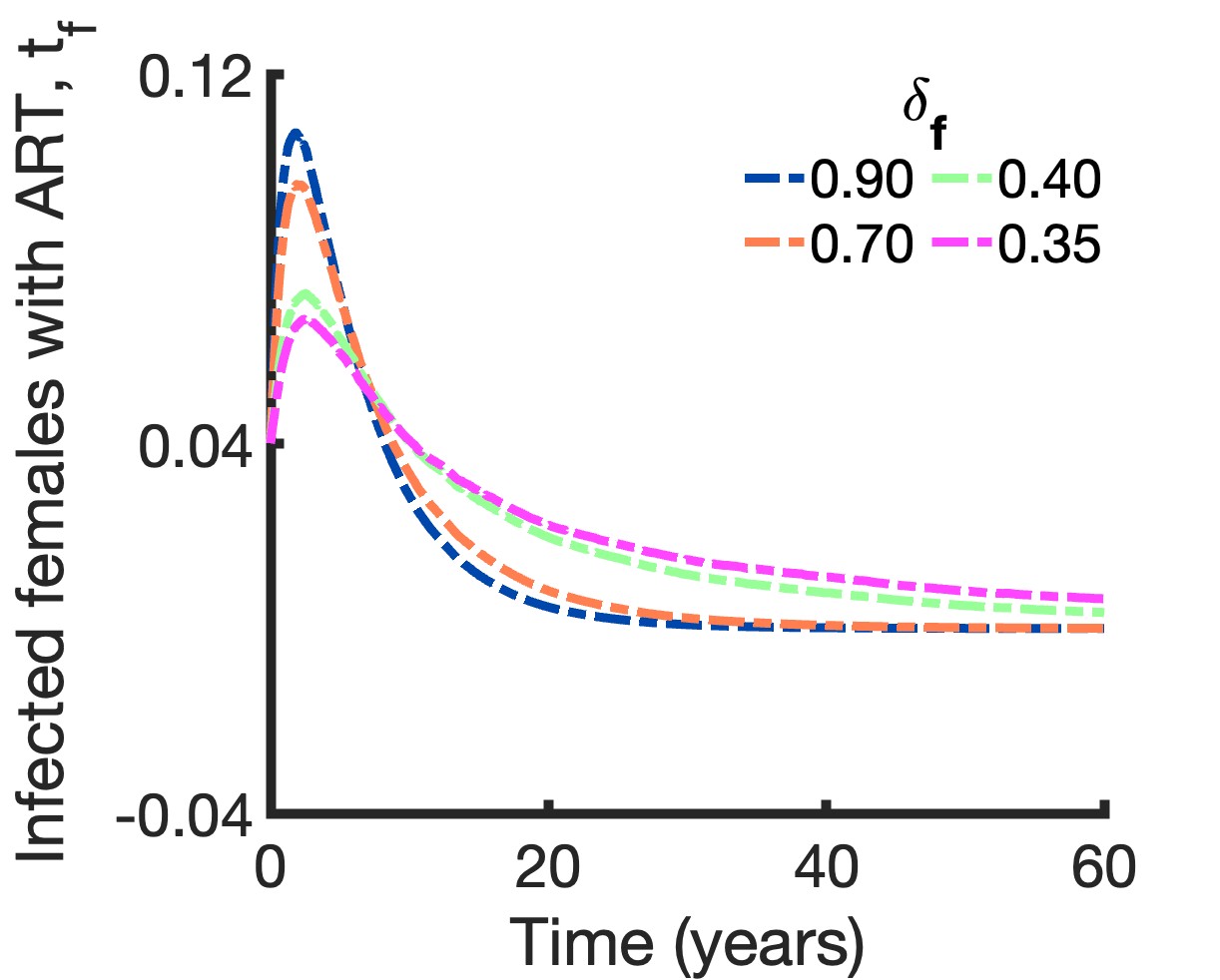}\label{fig:sub3131}}
	\caption{Disease-Free Equilibrium (DFE) of the stochastic Euler model for varying values of the portion of the infected females getting ART ($\delta_f$).}
	\label{fig:DFE_Euler_delta_f}
\end{figure}

\subsection{Stochastic Runge-Kutta}
In this subsection, we implemented the stochastic Runge-Kutta method to add complexity to our model. As in previous analyses, we examined the impact of the parameters $\gamma_f$, $\gamma_m$, and $\delta_f$ on the demographic groups ($i_m$, $i_f$, $t_f$). According to figure \ref{fig:sub421}, the Endemic Equilibrium (EE) of the stochastic Runge-Kutta model for varying values of the probability of transmission by a female who is infected ($\gamma_f$) shows that as $\gamma_f$ decreases, the number of infected males ($i_m$) decreases. Figure \ref{fig:sub422} similarly illustrates that a lower $\gamma_f$ leads to a reduction in the number of infected females ($i_f$). Figure \ref{fig:sub423} demonstrates that the number of infected females receiving antiretroviral therapy (ART) ($t_f$) also decreases with a lower $\gamma_f$, highlighting the interconnected dynamics of gender-specific transmission probabilities.
\begin{figure}[H]
	\centering
	\subfigure[$i_m$ vs $t$.]{\includegraphics[width=0.4\linewidth]{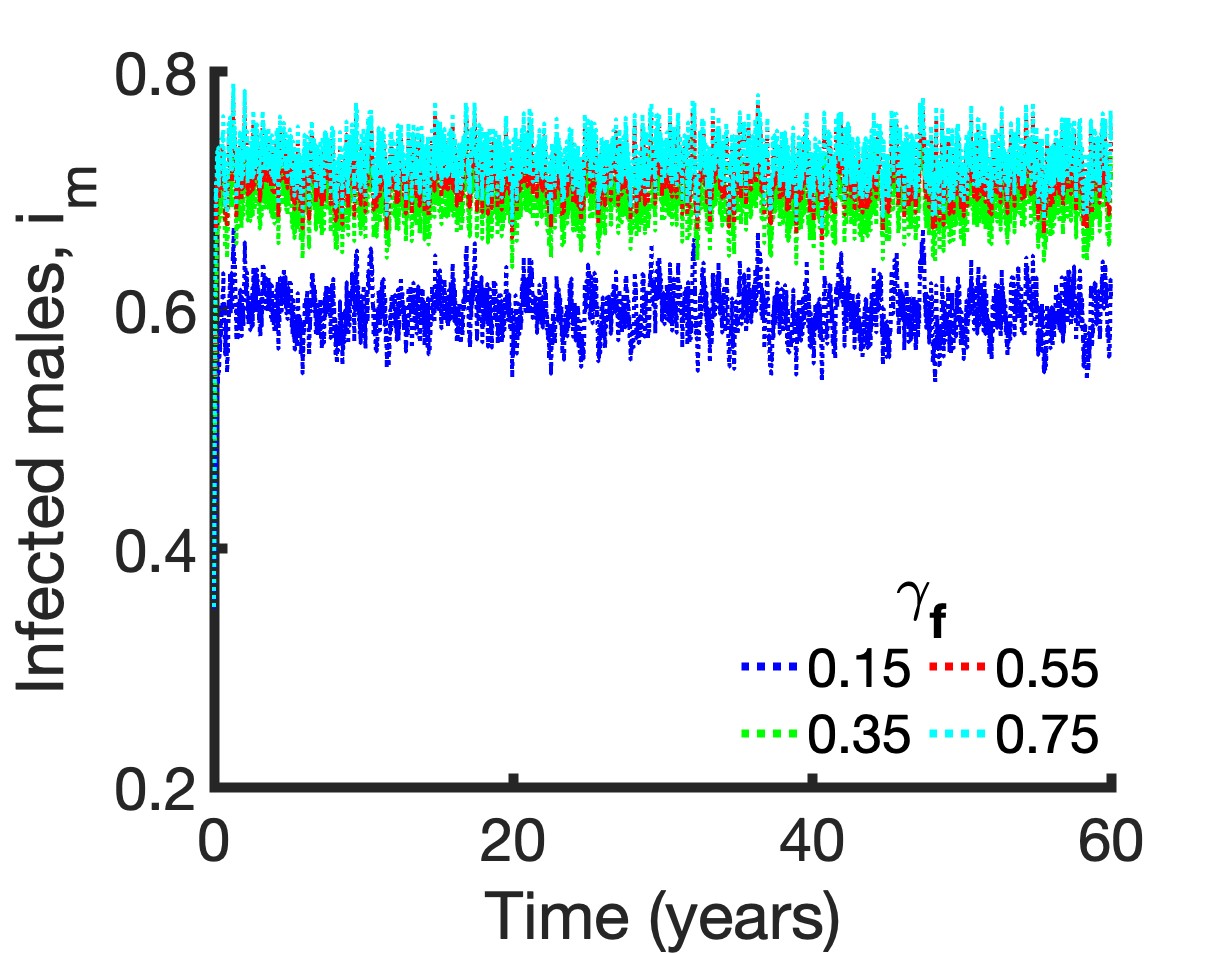}\label{fig:sub421}}
	\subfigure[$i_f$ vs $t$.]{\includegraphics[width=0.4\linewidth]{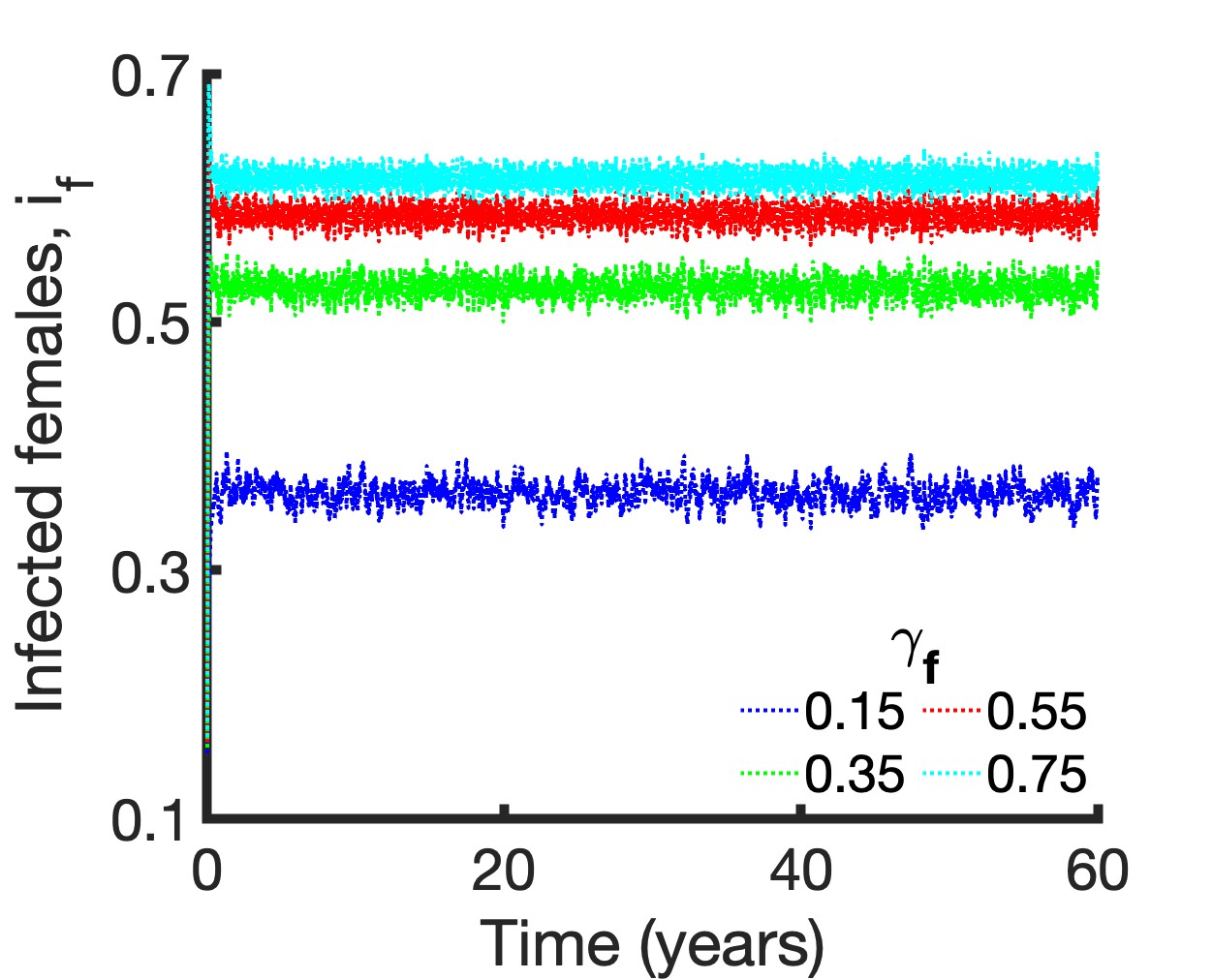}\label{fig:sub422}}
	\subfigure[$t_f$ vs $t$.]{\includegraphics[width=0.4\linewidth]{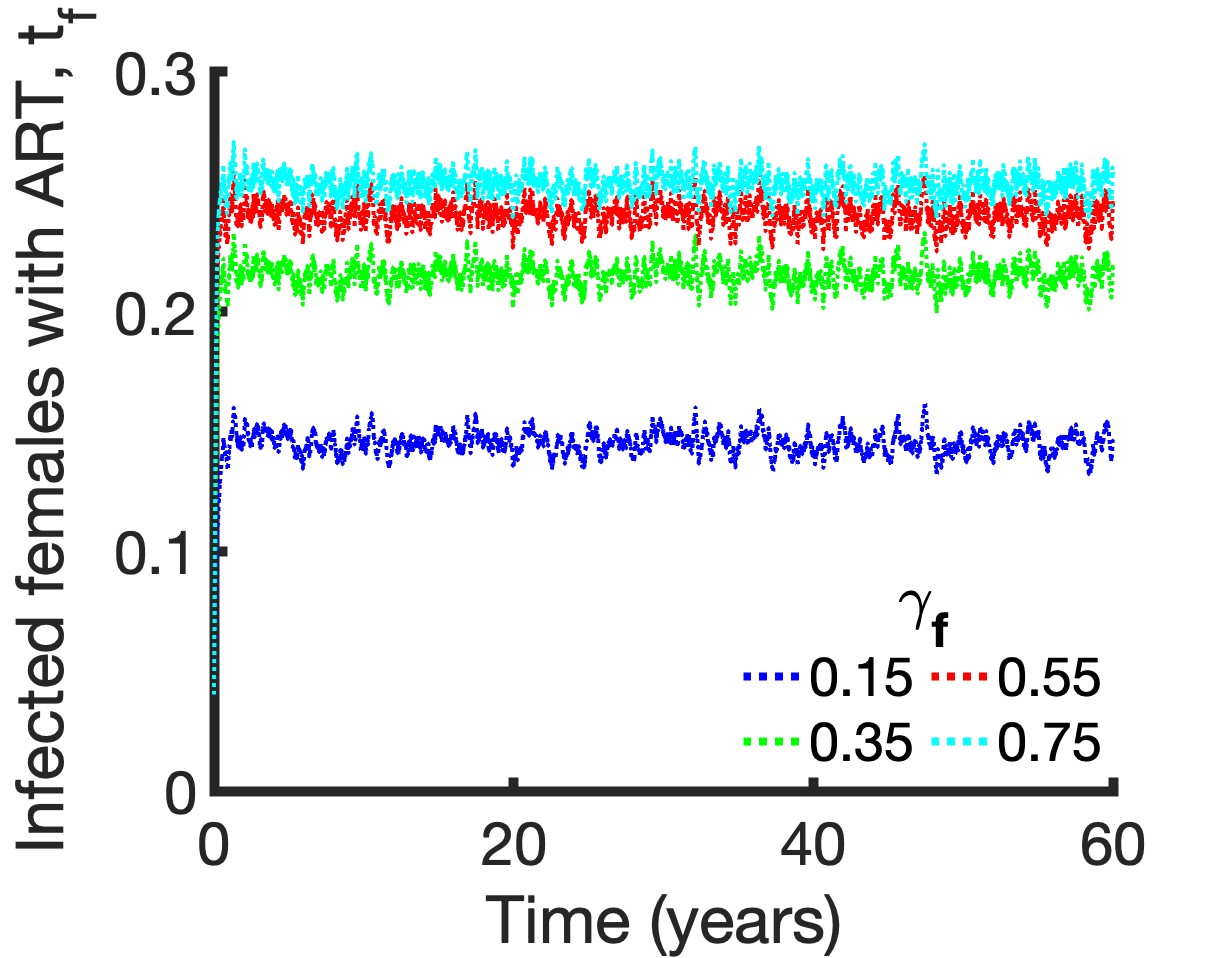}\label{fig:sub423}}
	\caption{Endemic Equilibrium (EE) of the stochastic Runge-Kutta model for varying values of the probability of transmission by a female who is infected  ($\gamma_f$).}
  \label{fig:EE_StoRK_gamma_f}
\end{figure}

\noindent According to figure \ref{fig:sub221}, the Endemic Equilibrium (EE) of the stochastic Runge-Kutta model for varying values of the probability of transmission by a male who is infected ($\gamma_m$) demonstrates that as $\gamma_m$ decreases, the number of infected males ($i_m$) also decreases. Figure \ref{fig:sub222} shows a similar trend for infected females ($i_f$), where a lower $\gamma_m$ leads to a reduced number of infected females. Figure \ref{fig:sub223} reveals that the number of infected females receiving antiretroviral therapy (ART) ($t_f$) also decreases as $\gamma_m$ is reduced, highlighting the interconnected dynamics of gender-specific transmission probabilities.

\begin{figure}[H]
	\centering
	\subfigure[$i_m$ vs $t$.]{\includegraphics[width=0.4\linewidth]{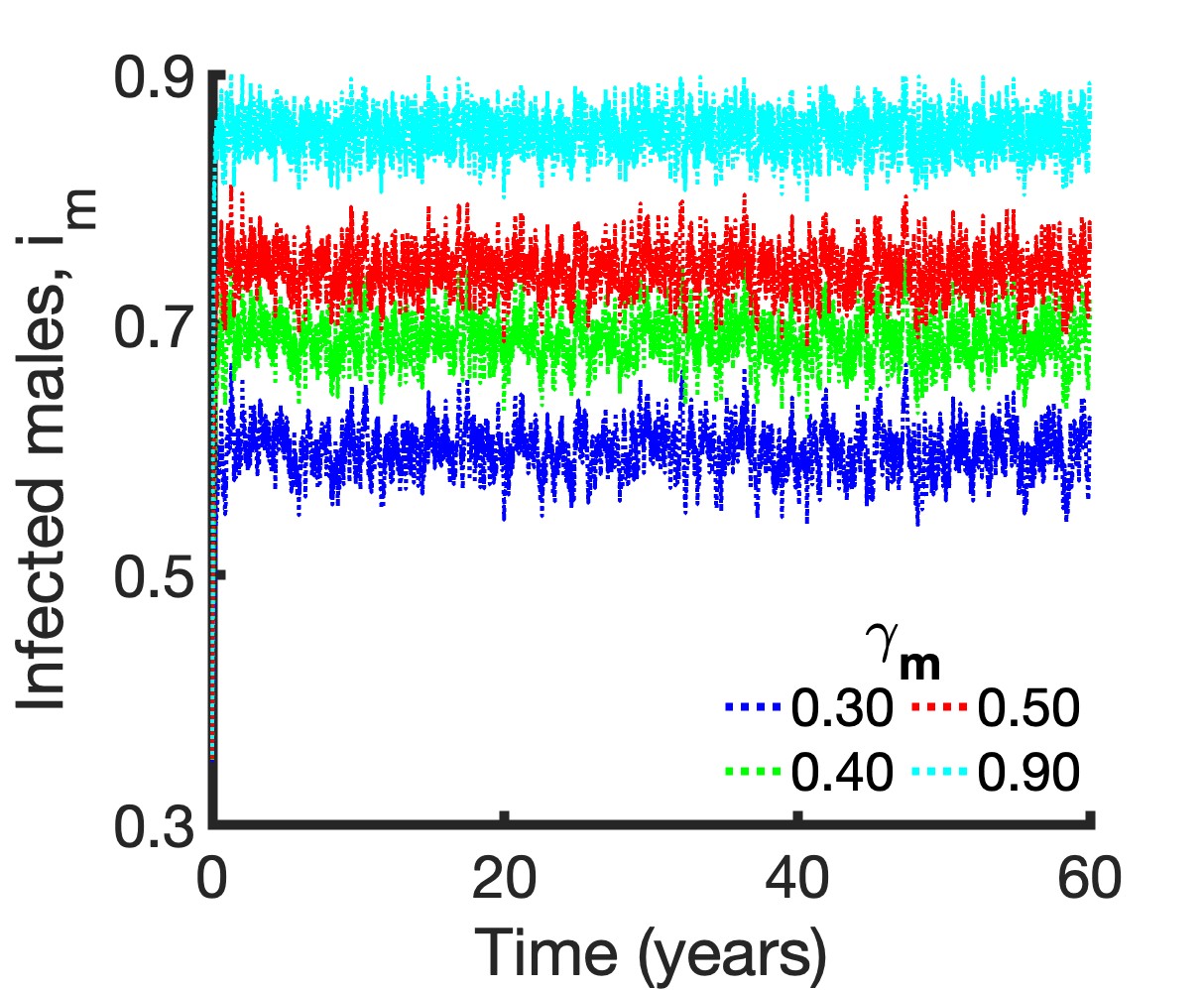}\label{fig:sub221}}
	\subfigure[$i_f$ vs $t$.]{\includegraphics[width=0.4\linewidth]{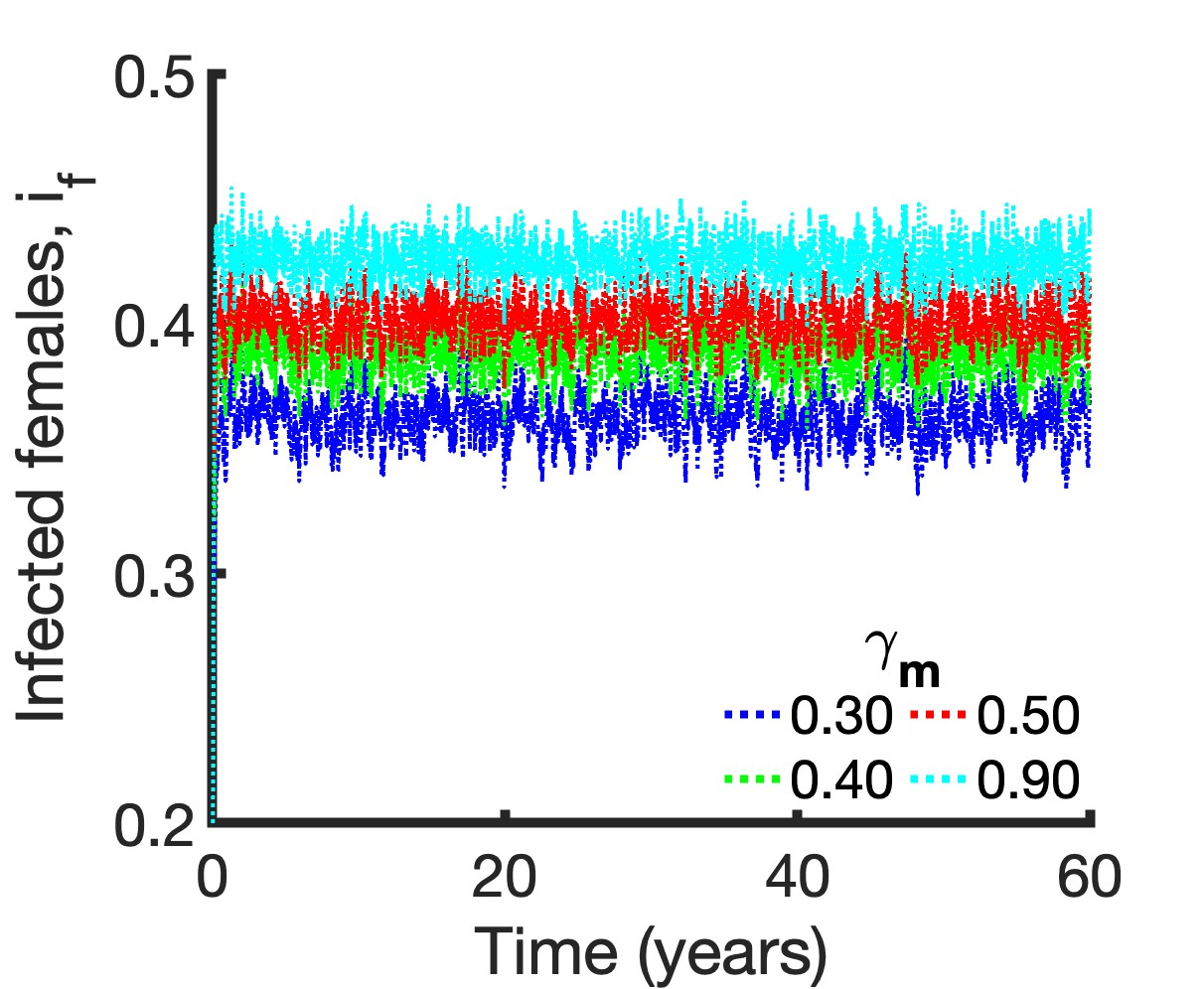}\label{fig:sub222}}
	\subfigure[$t_f$ vs $t$.]{\includegraphics[width=0.4\linewidth]{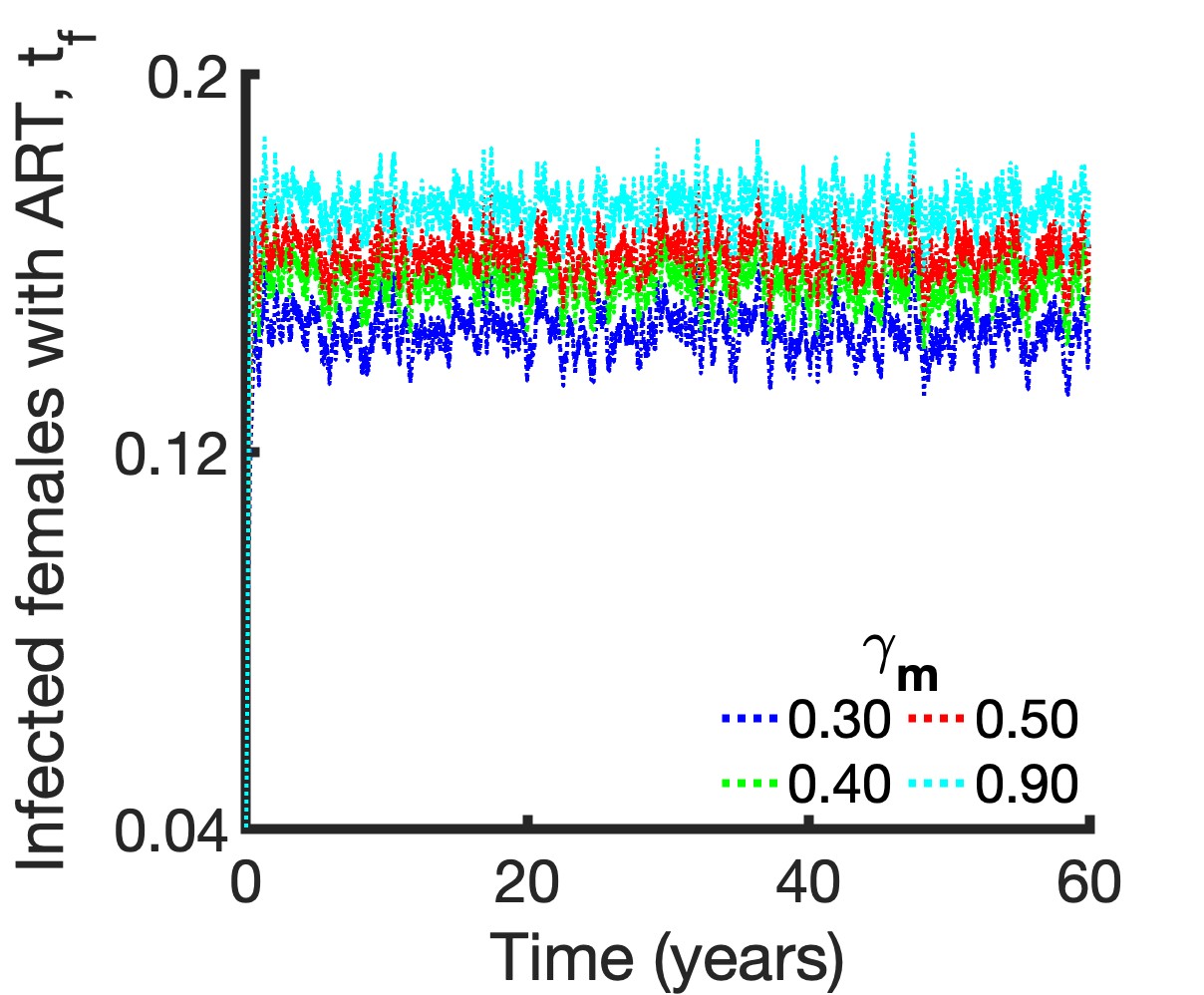}\label{fig:sub223}}
	\caption{Endemic Equilibrium (EE) of the stochastic Runge-Kutta model for varying values of the probability of transmission by a male who is infected  ($\gamma_m$).}
	\label{fig:EE_StoRK_gamma_m}
\end{figure}

\noindent As shown in figure \ref{fig:sub321}, the Endemic Equilibrium (EE) of the stochastic Runge-Kutta model for varying values of the portion of the infected females getting ART ($\delta_f$) indicates that increasing $\delta_f$ leads to a decrease in the number of infected males ($i_m$). Figure \ref{fig:sub322} illustrates that a higher $\delta_f$ results in fewer infected females ($i_f$). Figure \ref{fig:sub323} highlights that the number of infected females receiving ART ($t_f$) increases with higher $\delta_f$, underscoring the critical role of ART coverage in reducing overall infection rates and managing the epidemic. The key difference with the Runge-Kutta method is the increased fluctuation observed in the graphs, which highlights the method's sensitivity to capturing more nuanced variations in the data.

\begin{figure}[H]
	\centering
	\subfigure[$i_m$ vs $t$.]{\includegraphics[width=0.4\linewidth]{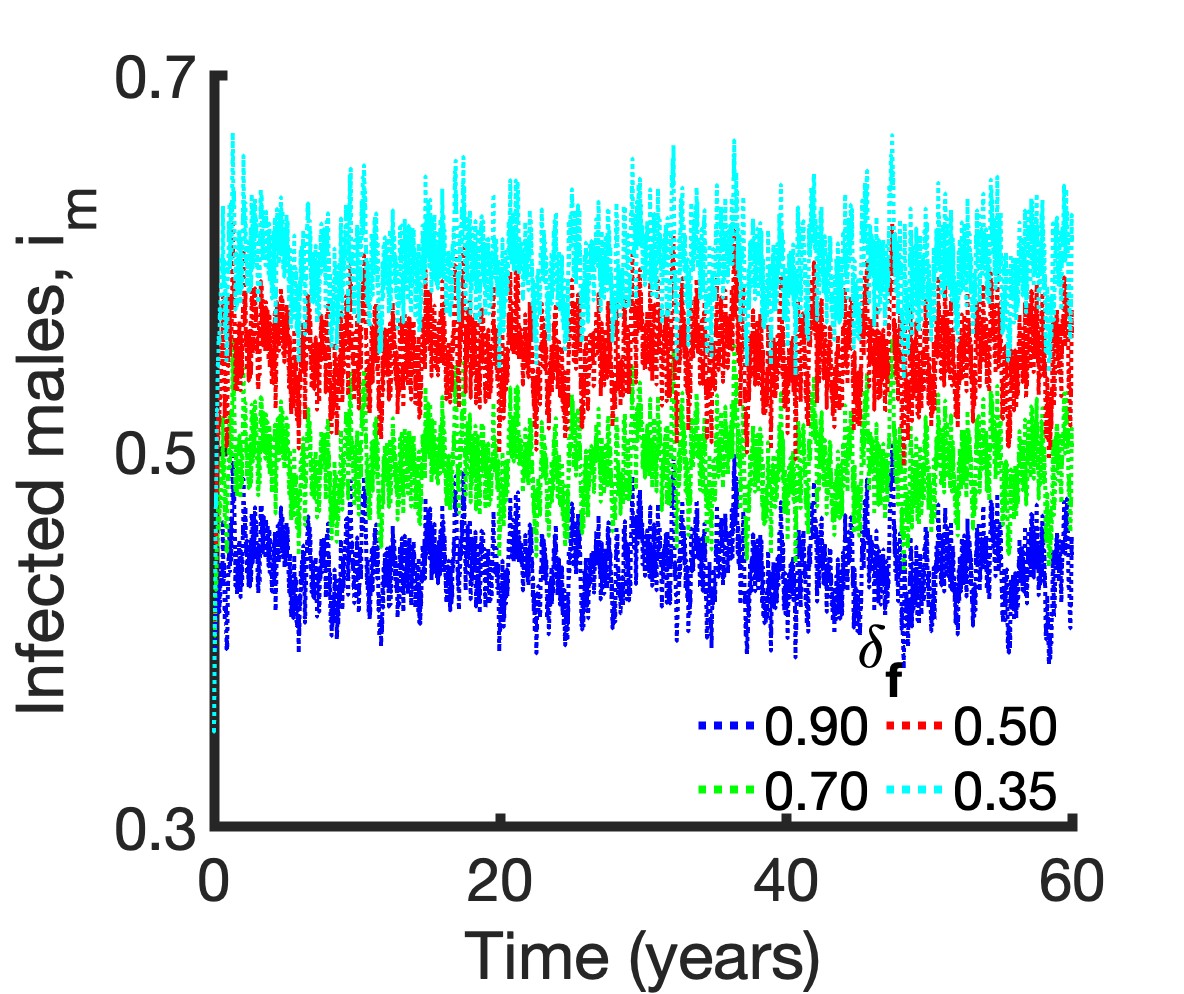}\label{fig:sub321}}
	\subfigure[$i_f$ vs $t$.]{\includegraphics[width=0.4\linewidth]{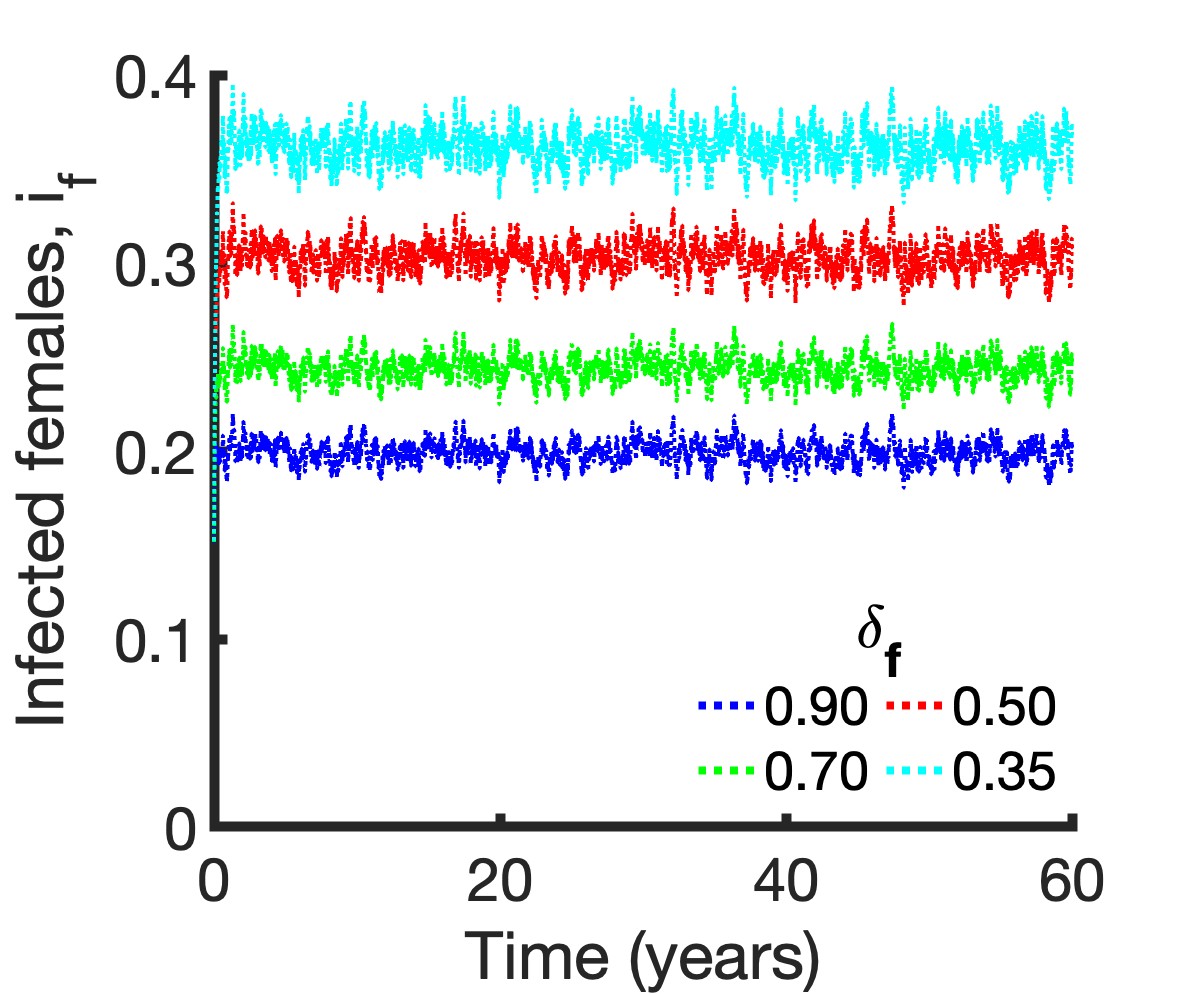}\label{fig:sub322}}
	\subfigure[$t_f$ vs $t$.]{\includegraphics[width=0.4\linewidth]{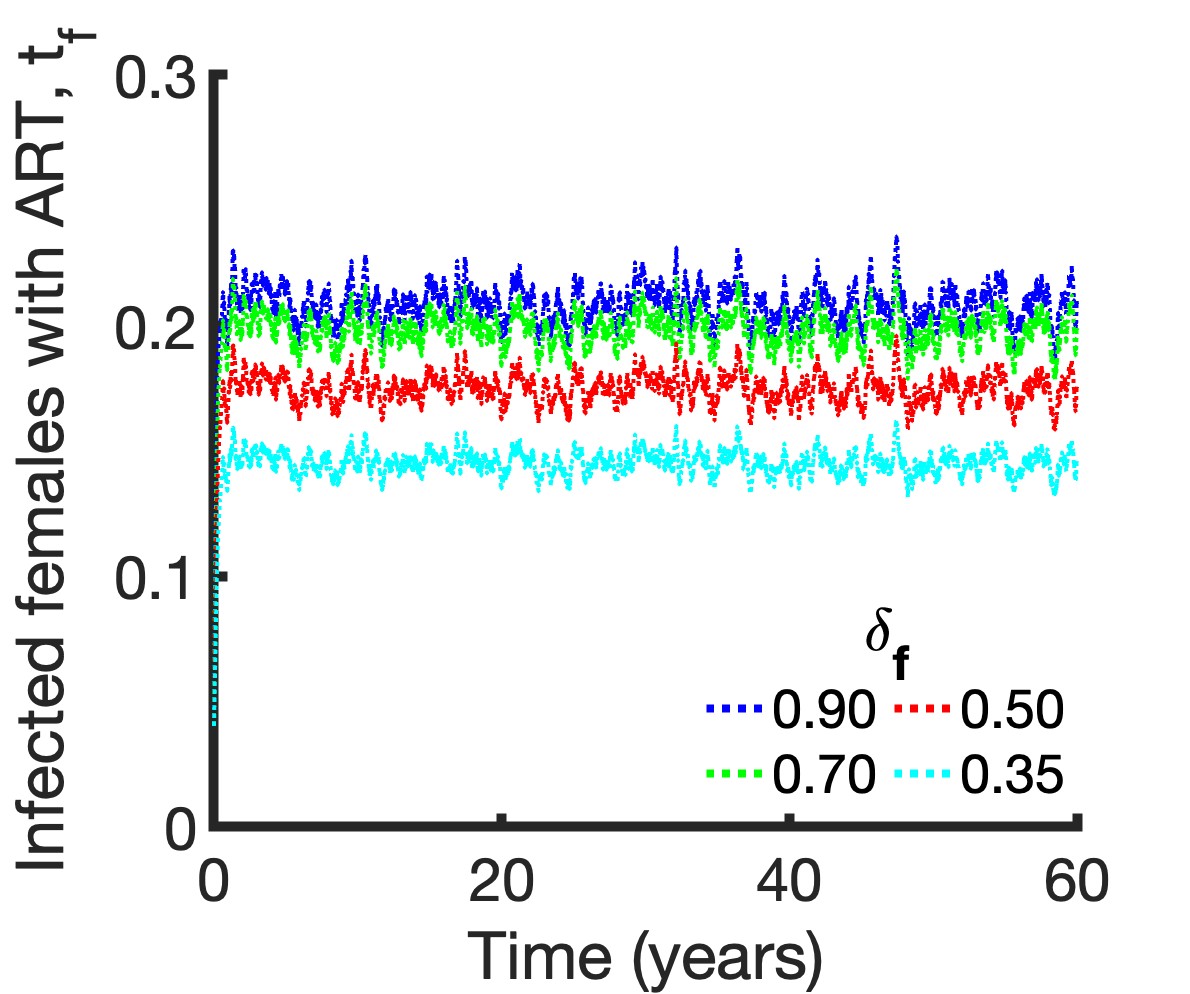}\label{fig:sub323}}
	  \caption{Endemic Equilibrium (EE) of the stochastic Runge-Kutta model for varying values of the portion of the infected females getting ART ($\delta_f$).}
	 \label{fig:EE_StoRK_delta_f}
\end{figure}

\noindent In our analysis using the Runge-Kutta method, the effect of step size on the accuracy of results for each population group ($i_m$, $i_f$, $t_f$) is evident from the graphs. Smaller step sizes lead to more precise outcomes, as they allow for a finer resolution of the dynamic changes within the model (see figures \ref{fig:sub411},  \ref{fig:sub412} and  \ref{fig:sub413}). This increased granularity helps capture subtle fluctuations and trends more accurately. Consequently, reducing the step size in the Runge-Kutta calculations enhances the fidelity of the model's predictions, aligning them more closely with expected theoretical behaviors and empirical data across all demographic groups. 

\begin{figure}[H]
	\centering
	\subfigure[$i_m$ vs $t$.]{\includegraphics[width=0.4\linewidth]{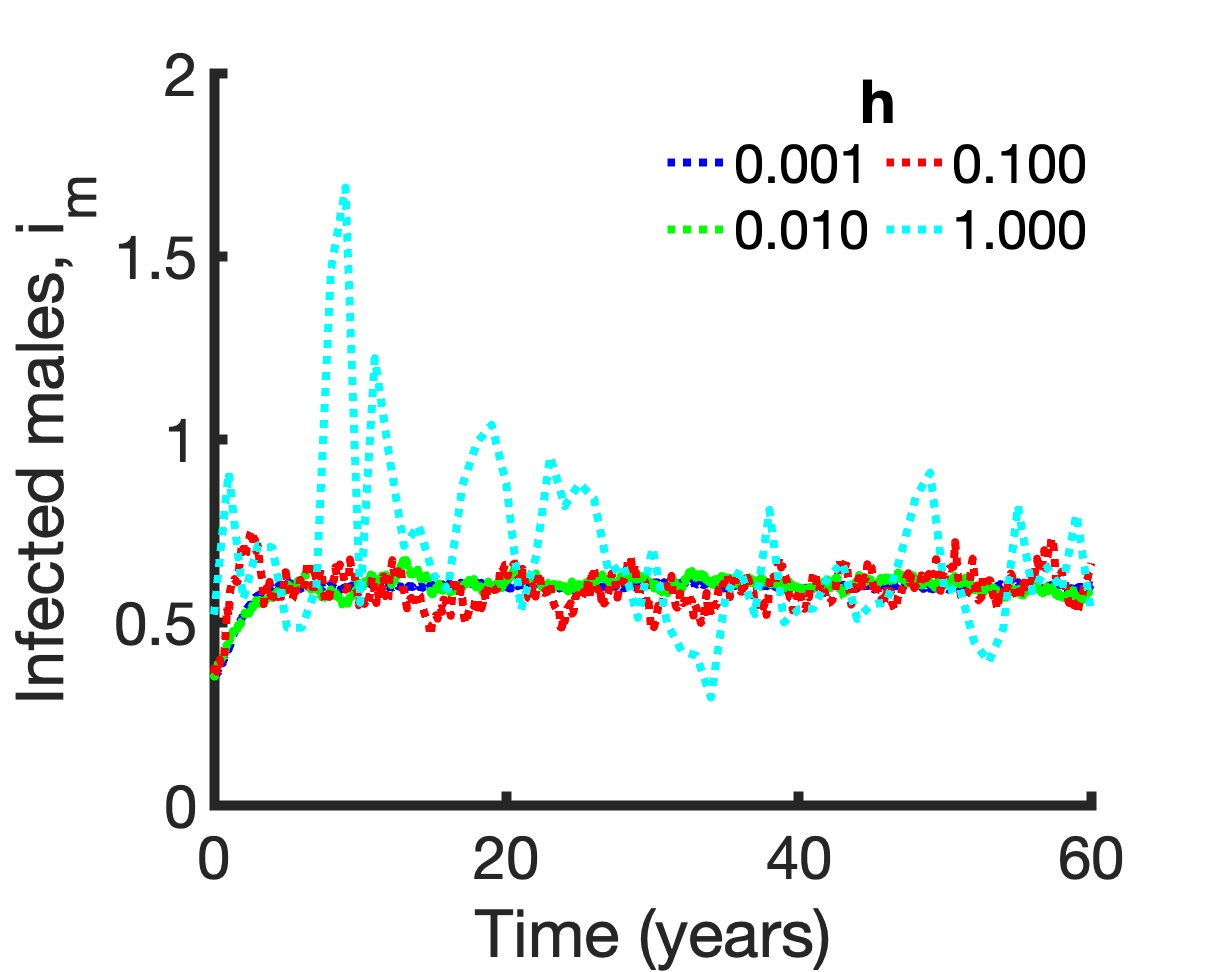}\label{fig:sub411}}
	\subfigure[$i_f$ vs $t$.]{\includegraphics[width=0.4\linewidth]{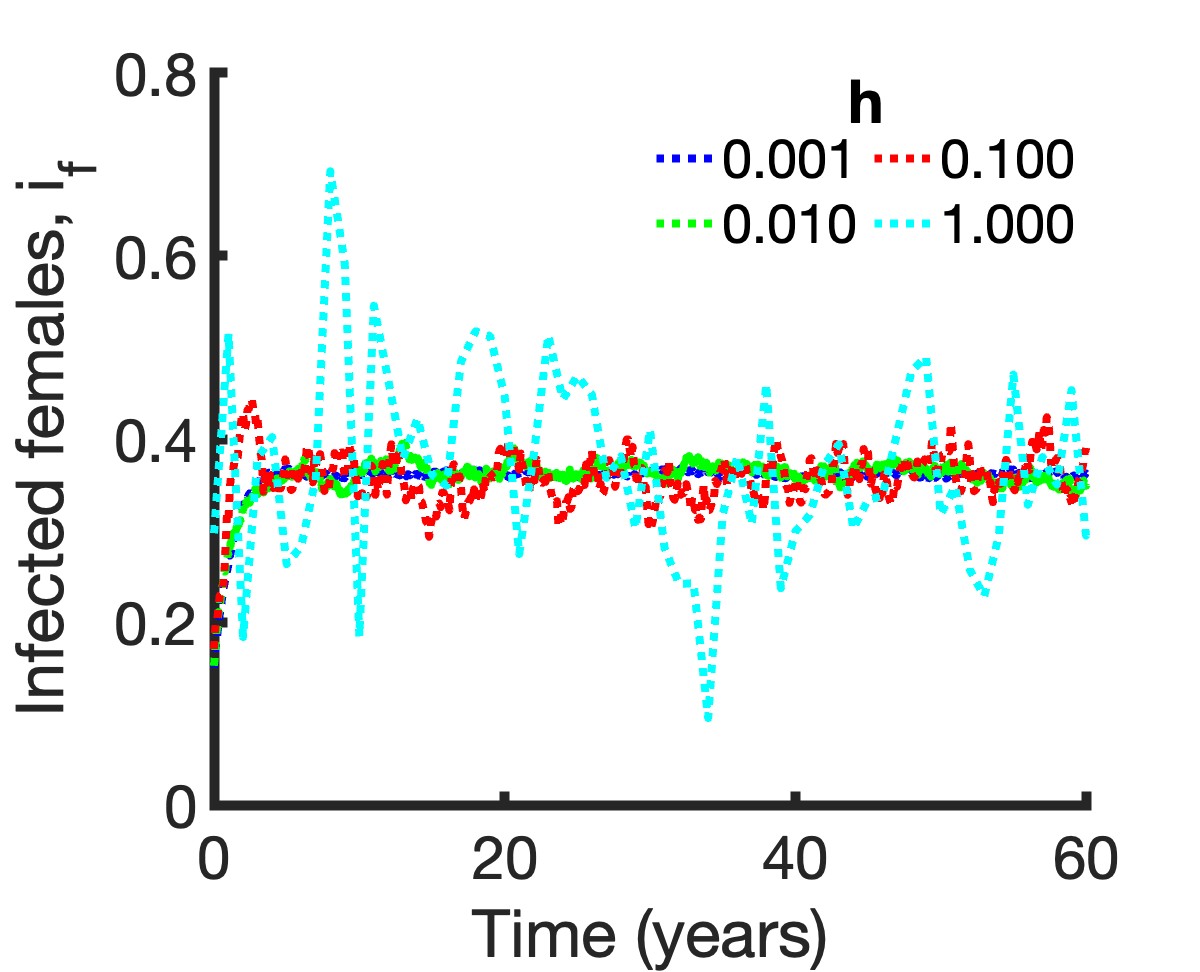}\label{fig:sub412}}
	\subfigure[$t_f$ vs $t$.]{\includegraphics[width=0.4\linewidth]{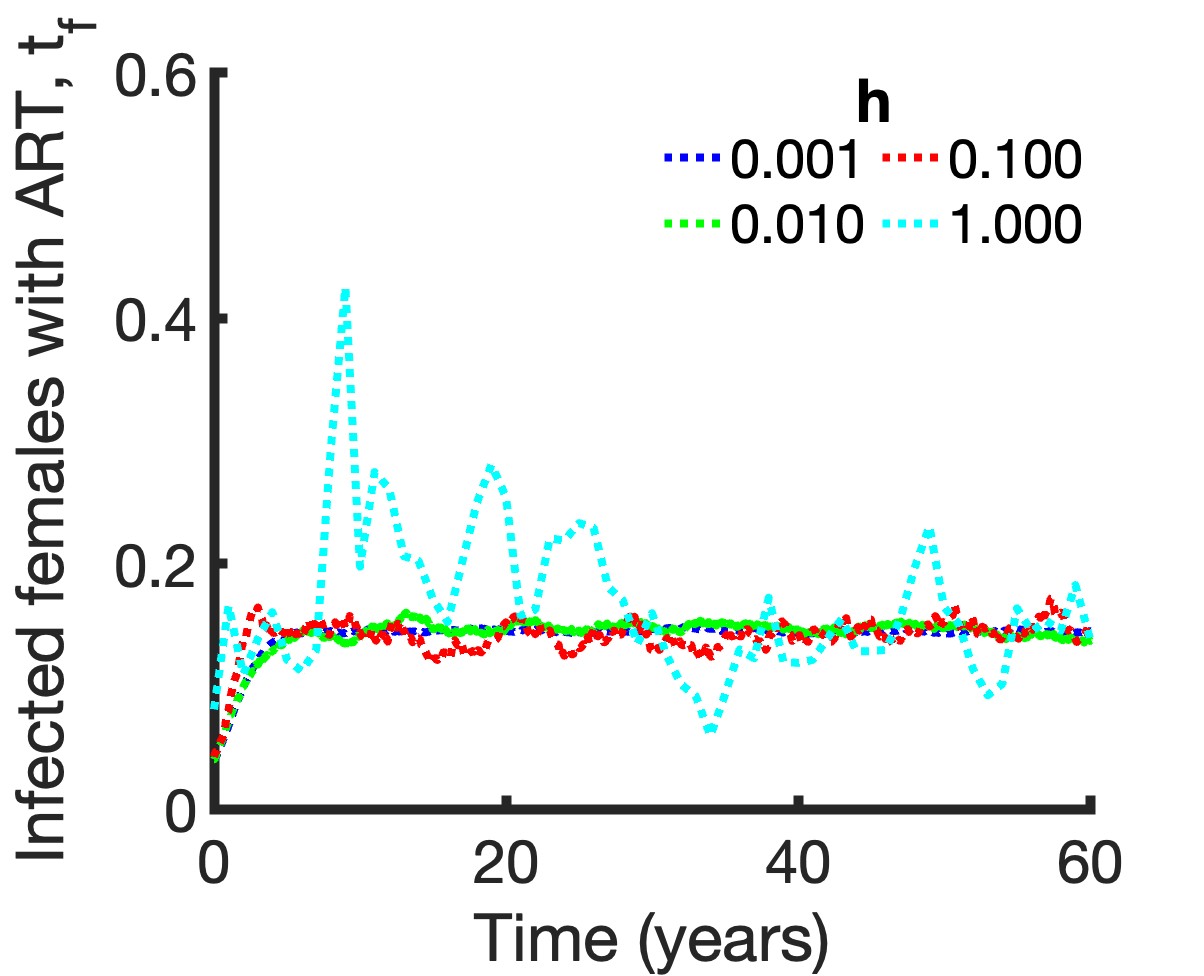}\label{fig:sub413}}
	\caption{Endemic Equilibrium (EE) of the stochastic Runge-Kutta model for varying values of step size $h$.}
	\label{fig:EE_StoRK_h}
\end{figure}

In this part, we utilised the stochastic Runge-Kutta approach to augment the intricacy of our model. In line with our previous studies, we examined the impact of the parameters $\gamma_f$, $\gamma_m$, and $\delta_f$ on the demographic groupings ($i_m$, $i_f$, $t_f$). The results, depicted in new figures that are similar to the ones stated before, confirm the comparable effects of these parameters as found in both the deterministic and Euler method-based models (refer to figures \ref{fig:sub511}, \ref{fig:sub512}, \ref{fig:sub513}, \ref{fig:sub521}, \ref{fig:sub522}, \ref{fig:sub523}, \ref{fig:sub531}, \ref{fig:sub532} and \ref{fig:sub533}). The main characteristic of the Runge-Kutta technique is its increased fluctuations in the graphs, which indicate the system's capacity to capture more intricate changes in the data.

\begin{figure}[H]
	\centering
	\subfigure[$i_m$ vs $t$.]{\includegraphics[width=0.4\linewidth]{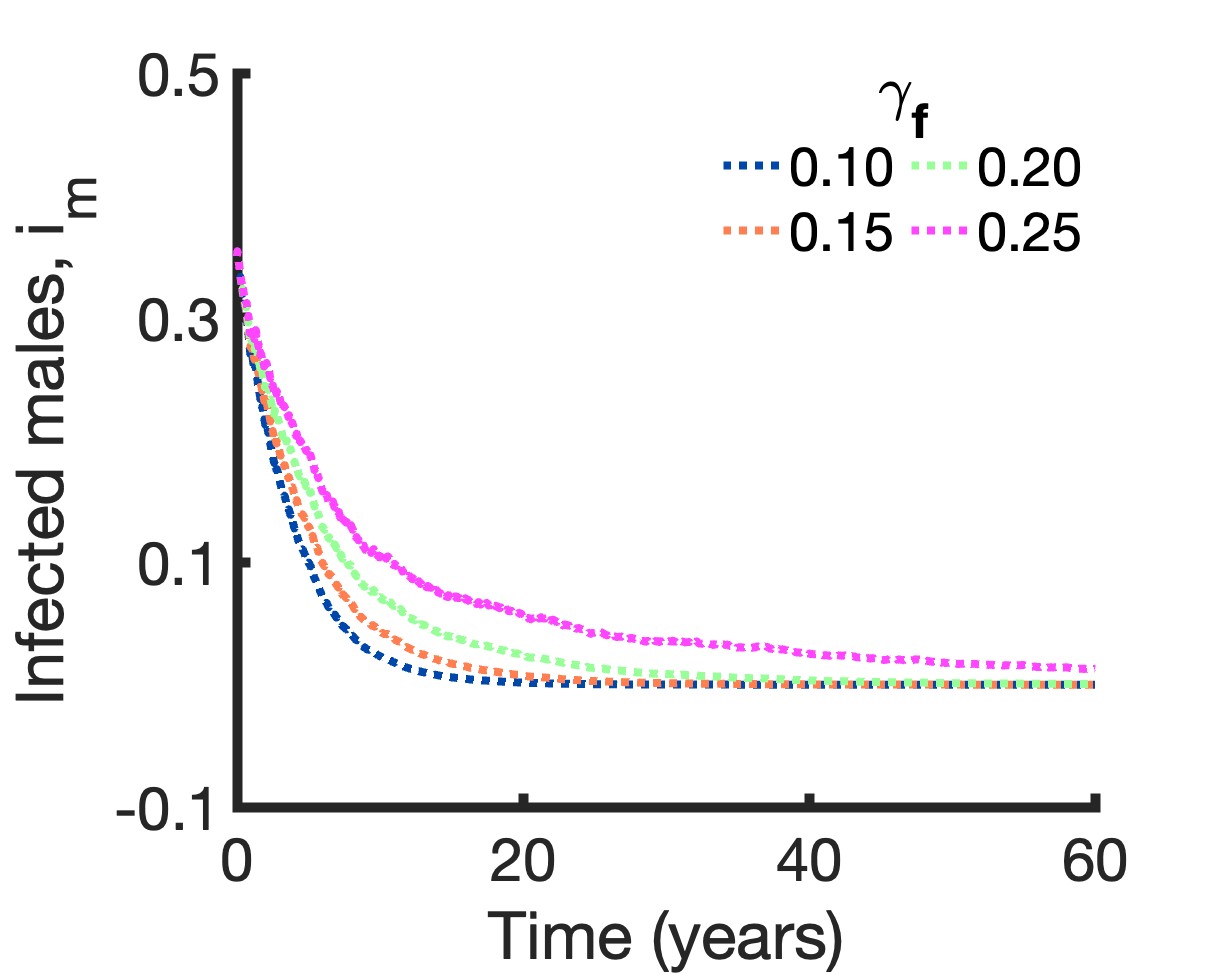}\label{fig:sub511}}
	\subfigure[$i_f$ vs $t$.]{\includegraphics[width=0.4\linewidth]{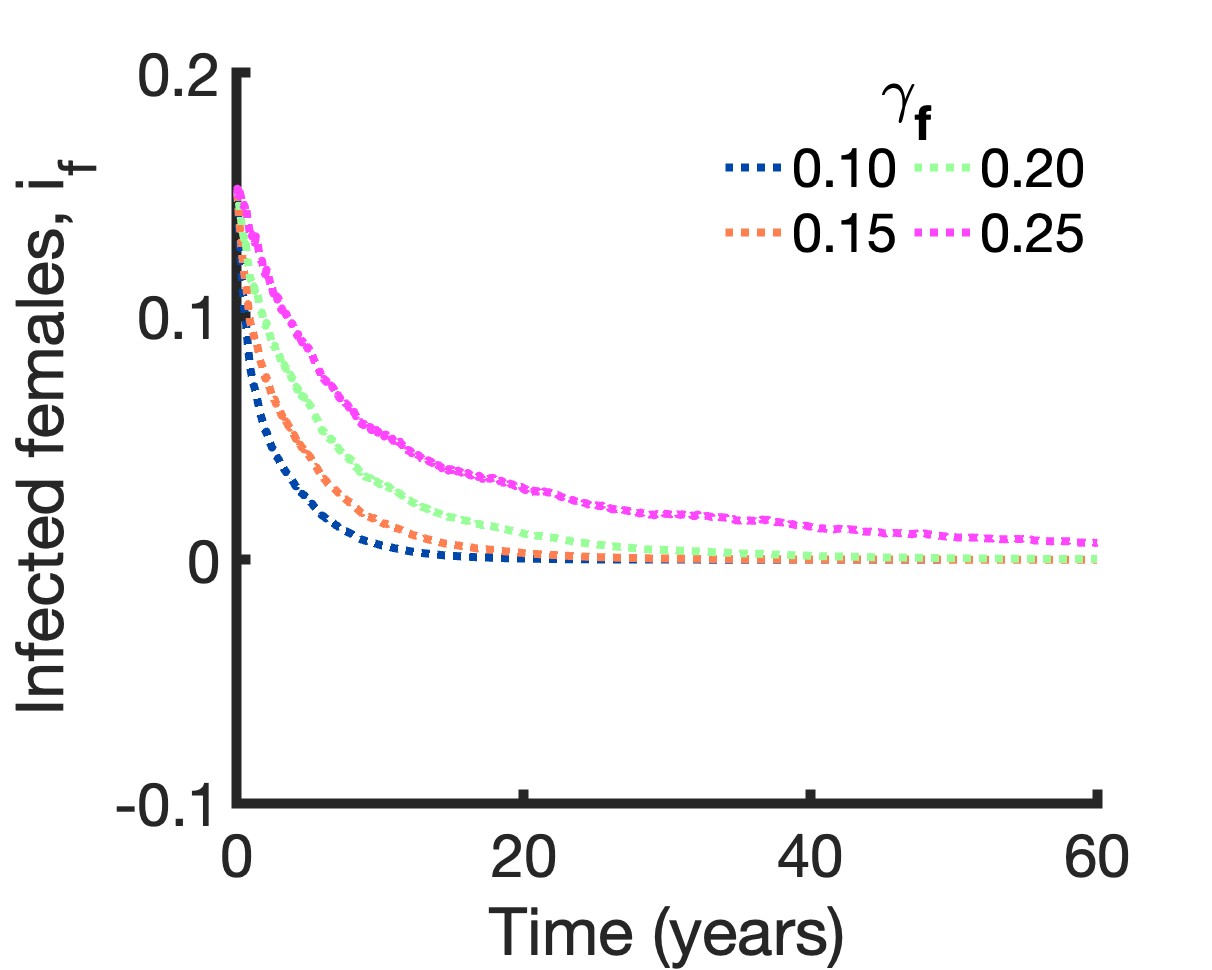}\label{fig:sub512}}
	\subfigure[$t_f$ vs $t$.]{\includegraphics[width=0.4\linewidth]{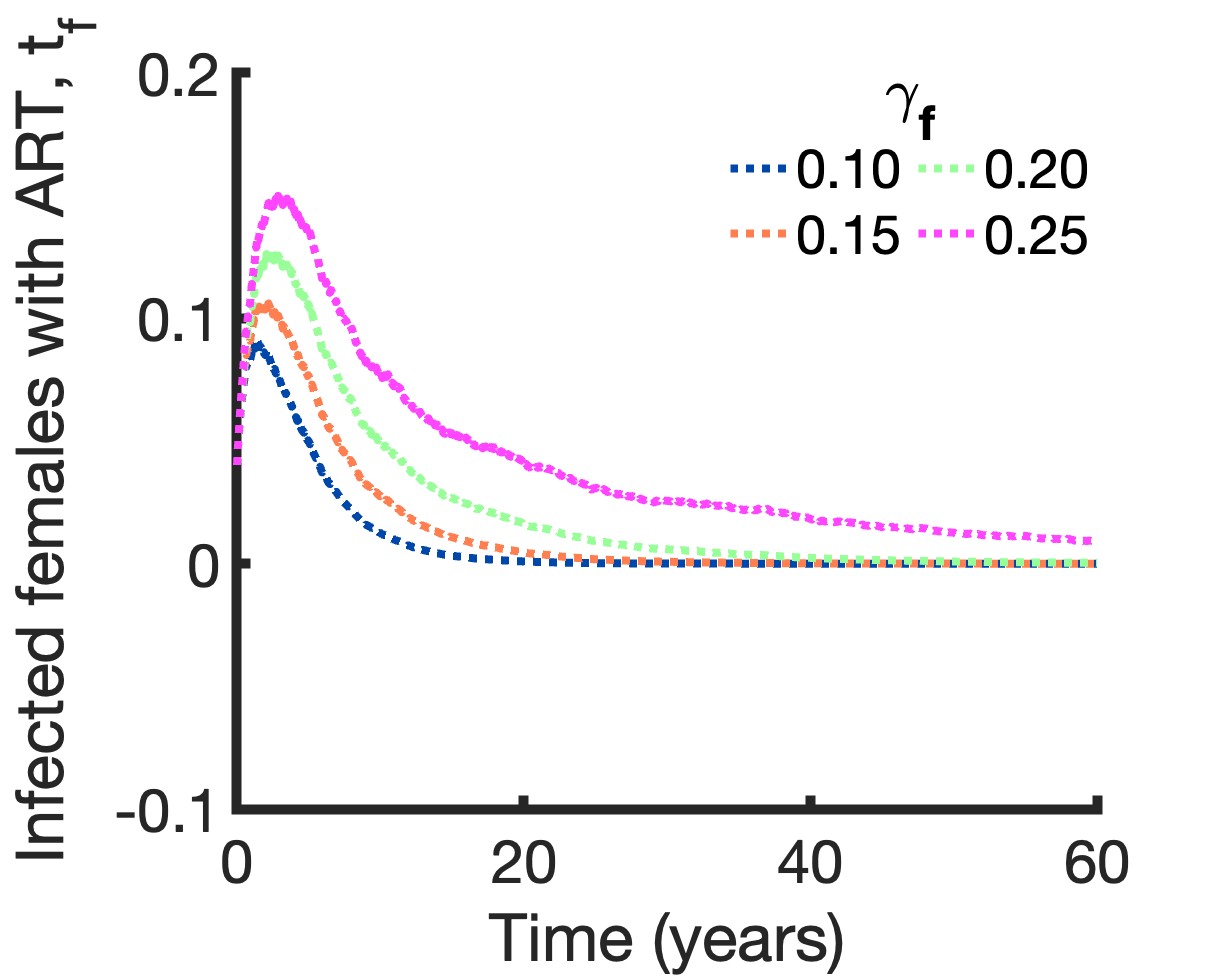}\label{fig:sub513}}
	\caption{Disease-Free Equilibrium (DFE) of the stochastic Runge-Kutta model for varying values of the probability of transmission by a female who is infected  ($\gamma_f$).}
	 \label{fig:DFE_StoRK_gamma_f}
\end{figure}

\begin{figure}[H]
	\centering
	\subfigure[$i_m$ vs $t$.]{\includegraphics[width=0.4\linewidth]{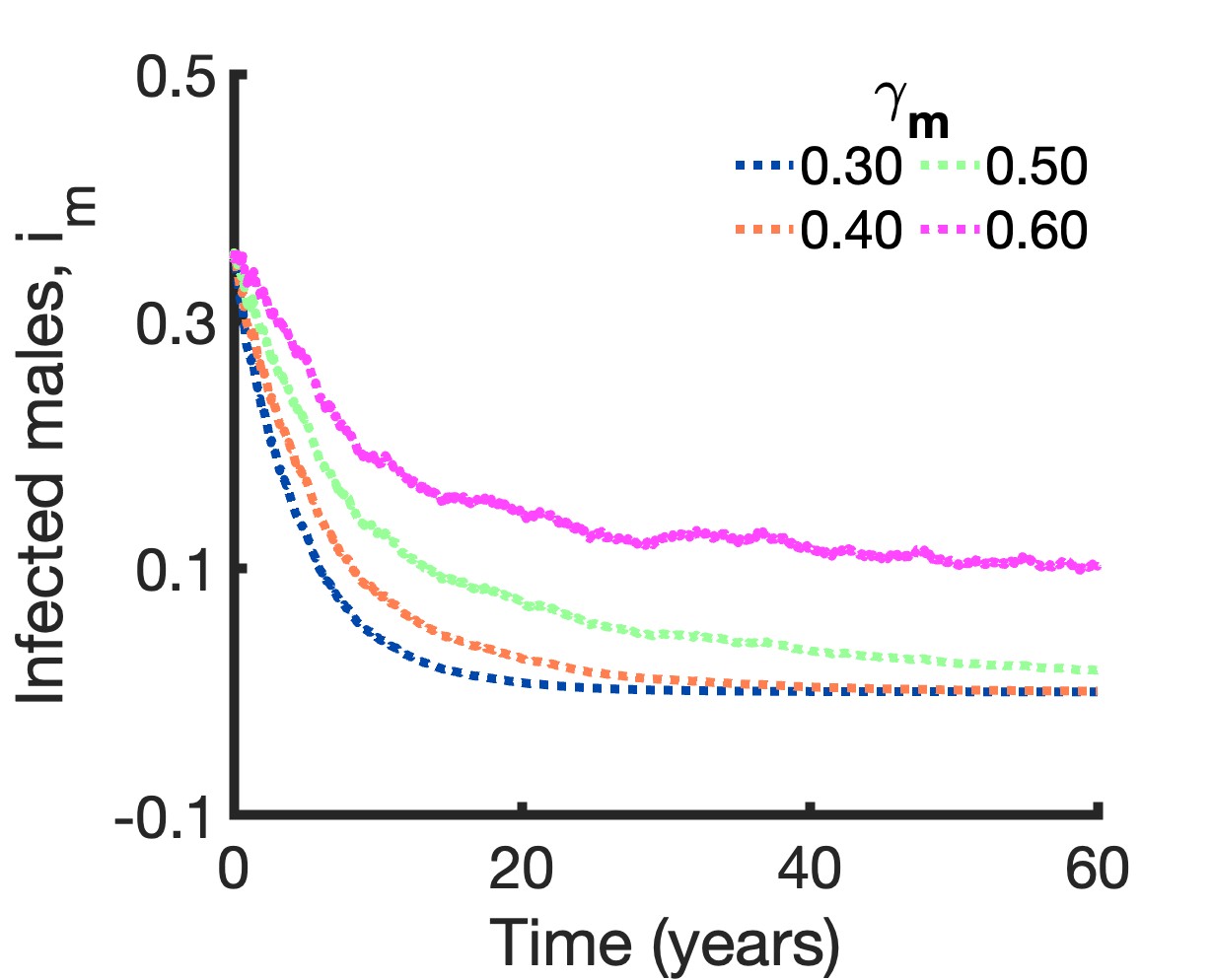}\label{fig:sub521}}
	\subfigure[$i_f$ vs $t$.]{\includegraphics[width=0.4\linewidth]{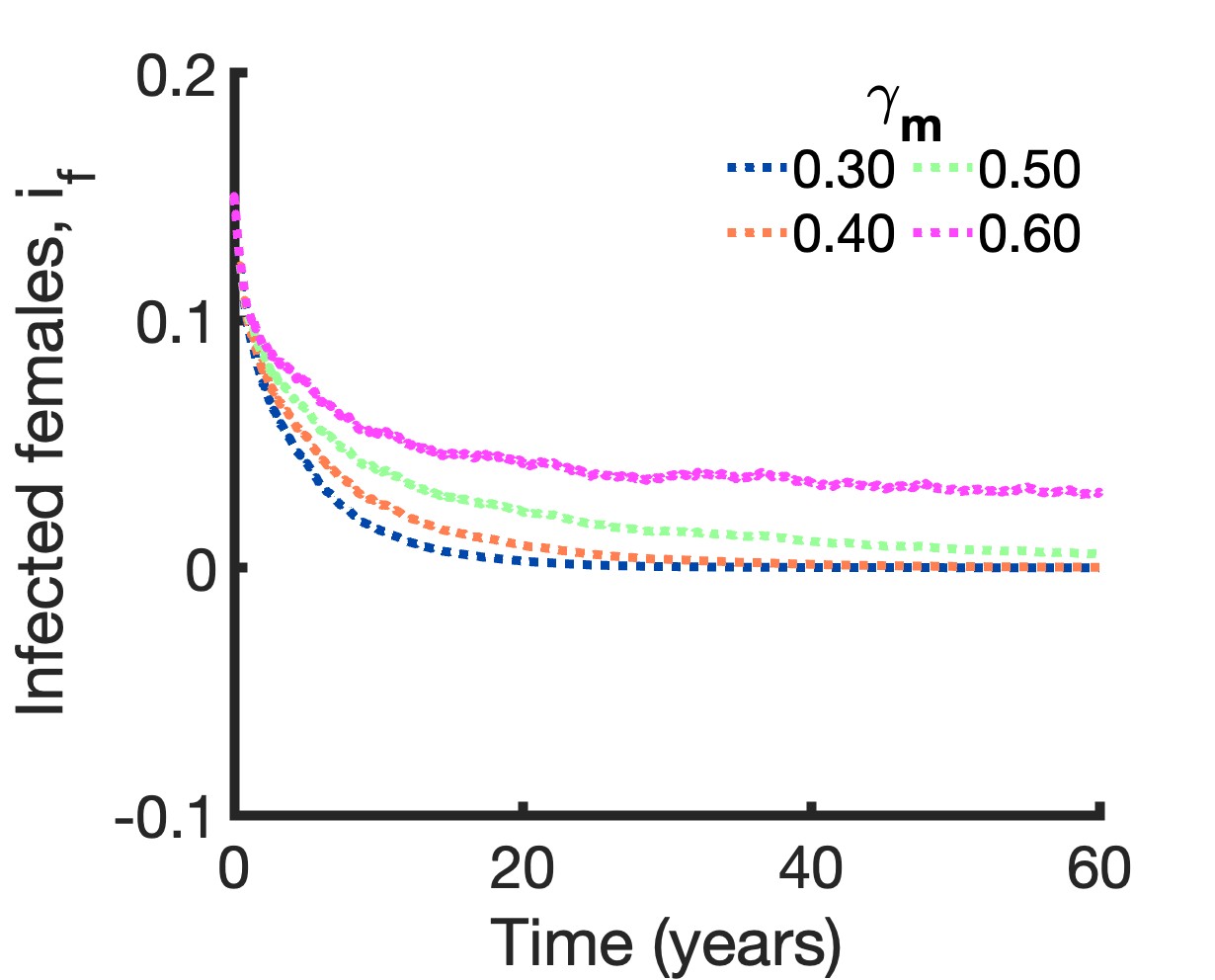}\label{fig:sub522}}
	\subfigure[$t_f$ vs $t$.]{\includegraphics[width=0.4\linewidth]{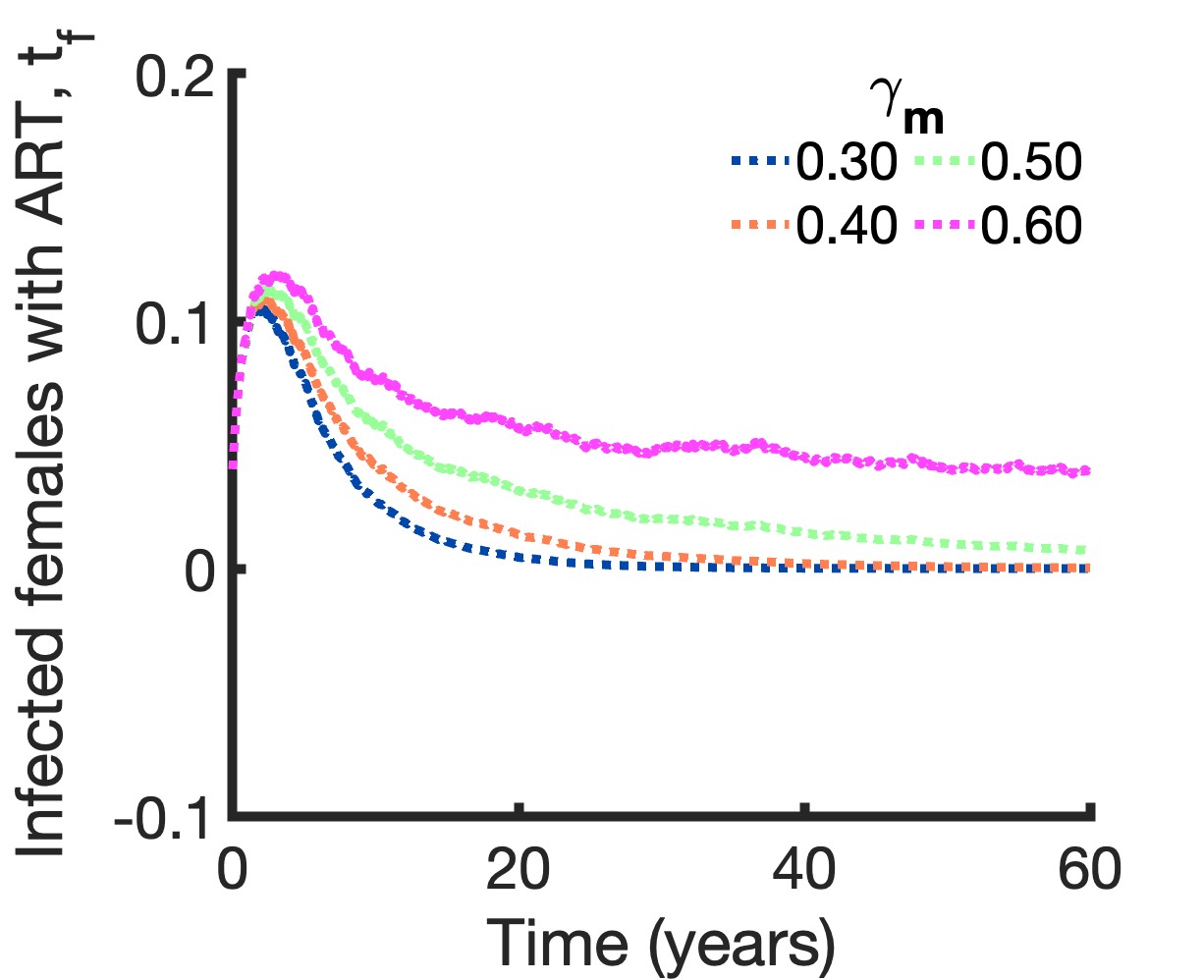}\label{fig:sub523}}
	\caption{Disease-Free Equilibrium (DFE) of the stochastic Runge-Kutta model for varying values of the probability of transmission by a male who is infected  ($\gamma_m$).}
	\label{fig:DFE_StoRK_gamma_m}
\end{figure}

\begin{figure}[H]
	\centering
	\subfigure[$i_m$ vs $t$.]{\includegraphics[width=0.4\linewidth]{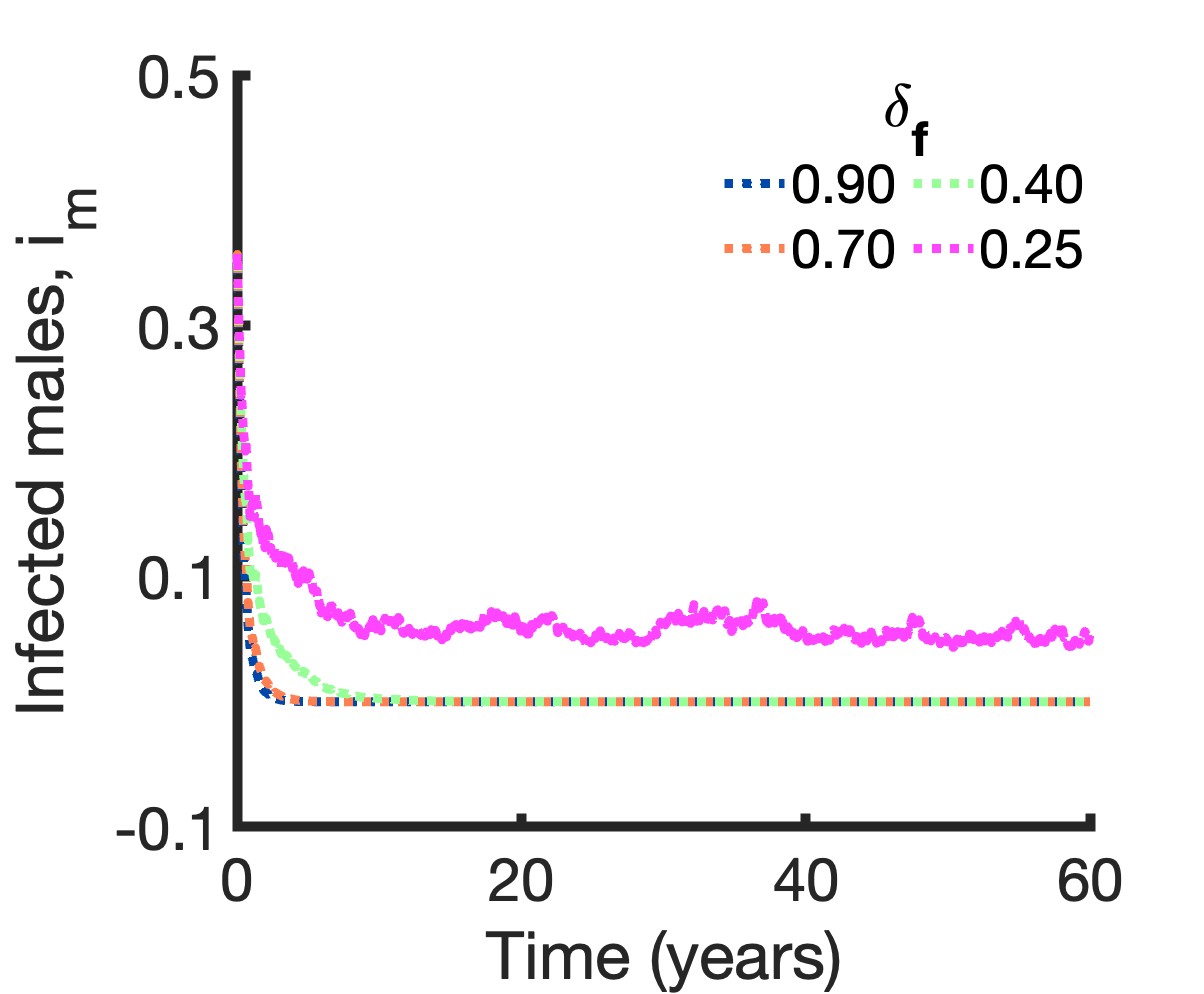}\label{fig:sub531}}
	\subfigure[$i_f$ vs $t$.]{\includegraphics[width=0.4\linewidth]{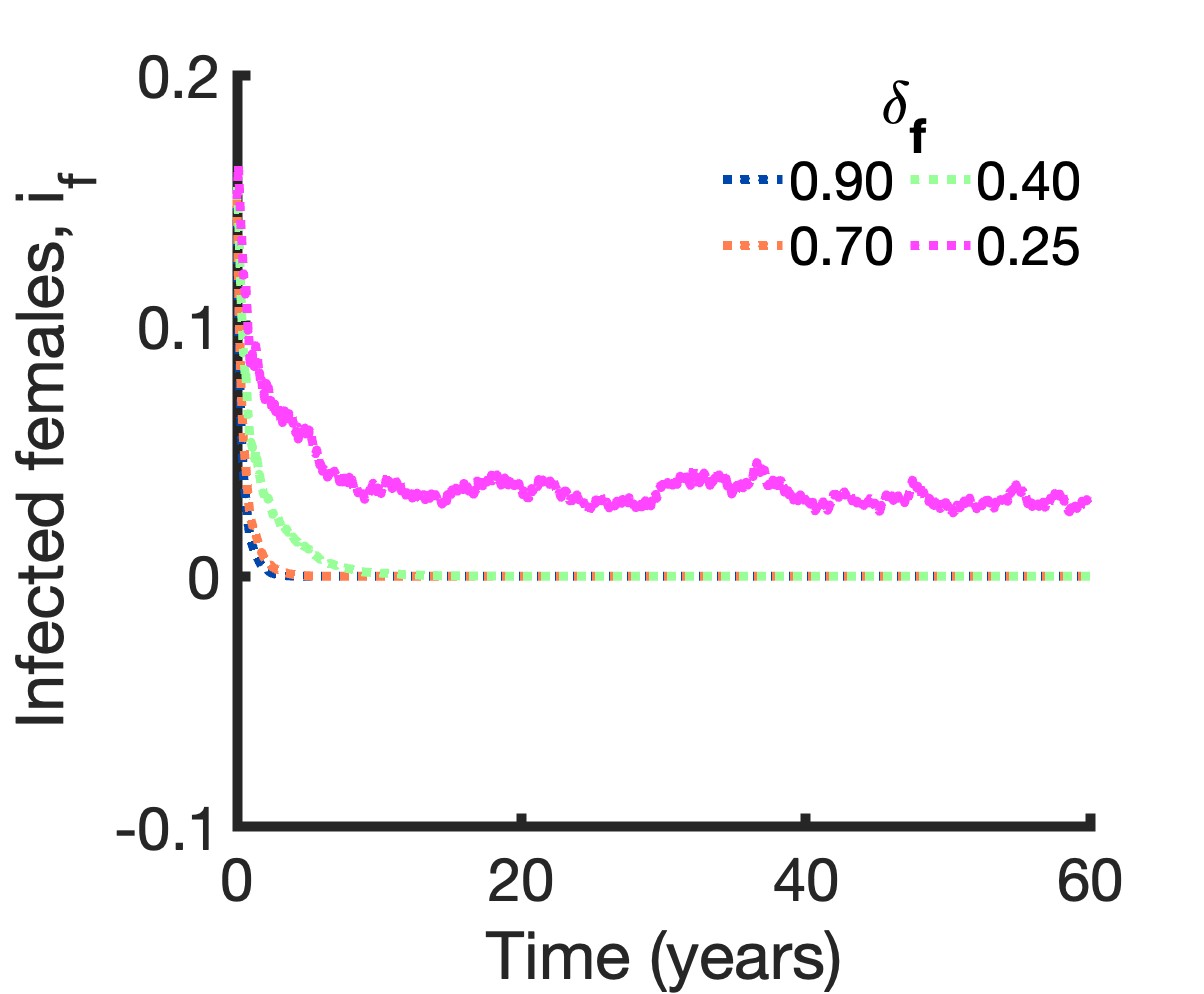}\label{fig:sub532}}
	\subfigure[$t_f$ vs $t$.]{\includegraphics[width=0.4\linewidth]{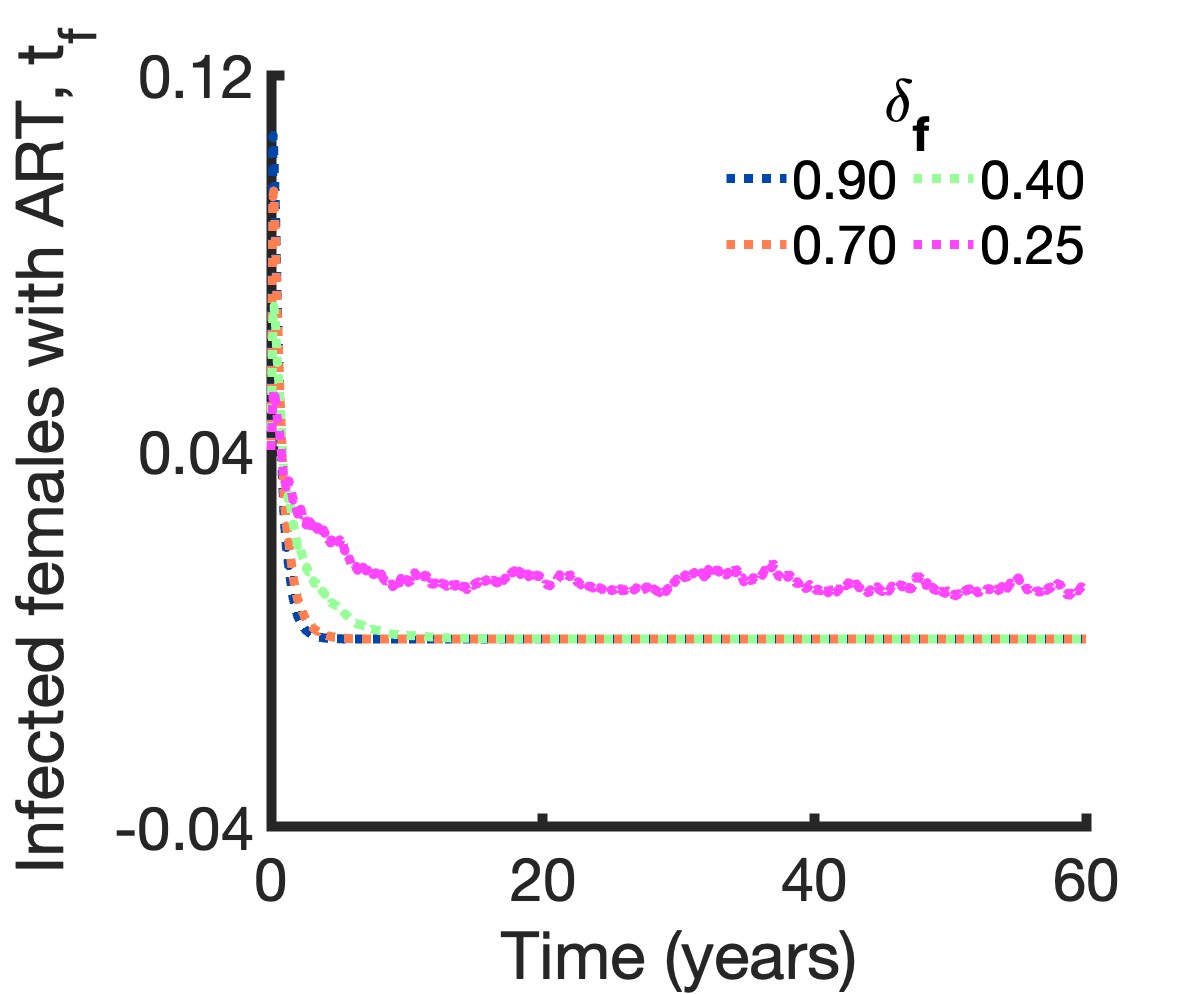}\label{fig:sub533}}
 \caption{Disease-Free Equilibrium (DFE) of the stochastic Runge-Kutta model for varying values of the portion of the infected females getting ART ($\delta_f$).}
  \label{fig:DFE_StoRK_delta_f}
\end{figure}

\noindent The plots (see figures \ref{fig:sub4112}, \ref{fig:sub4122} and \ref{fig:sub4132}) show that step size affects the accuracy of results for each population group ($i_m, i_f, t_f$) in our Runge-Kutta study. Smaller step sizes result in more accurate results by allowing for finer resolution of dynamic changes in the model. This granularity enhances the accuracy of capturing small variations and patterns. Reducing the step size in Runge-Kutta computations improves model fidelity, matching predictions with theoretical and empirical evidence across all demographics.

\begin{figure}[H]
	\centering
	\subfigure[$i_m$ vs $t$.]{\includegraphics[width=0.4\linewidth]{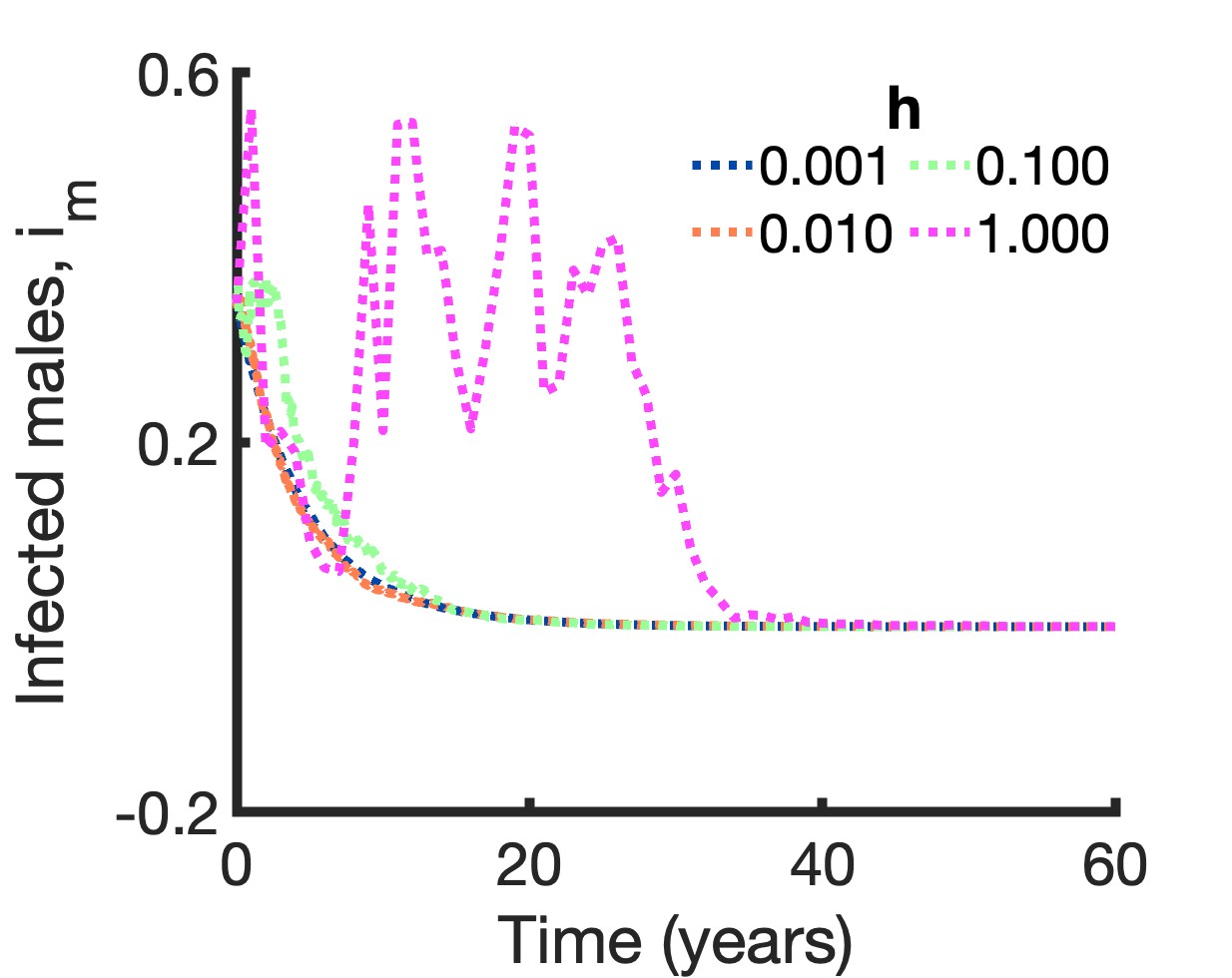}\label{fig:sub4112}}
	\subfigure[$i_f$ vs $t$.]{\includegraphics[width=0.4\linewidth]{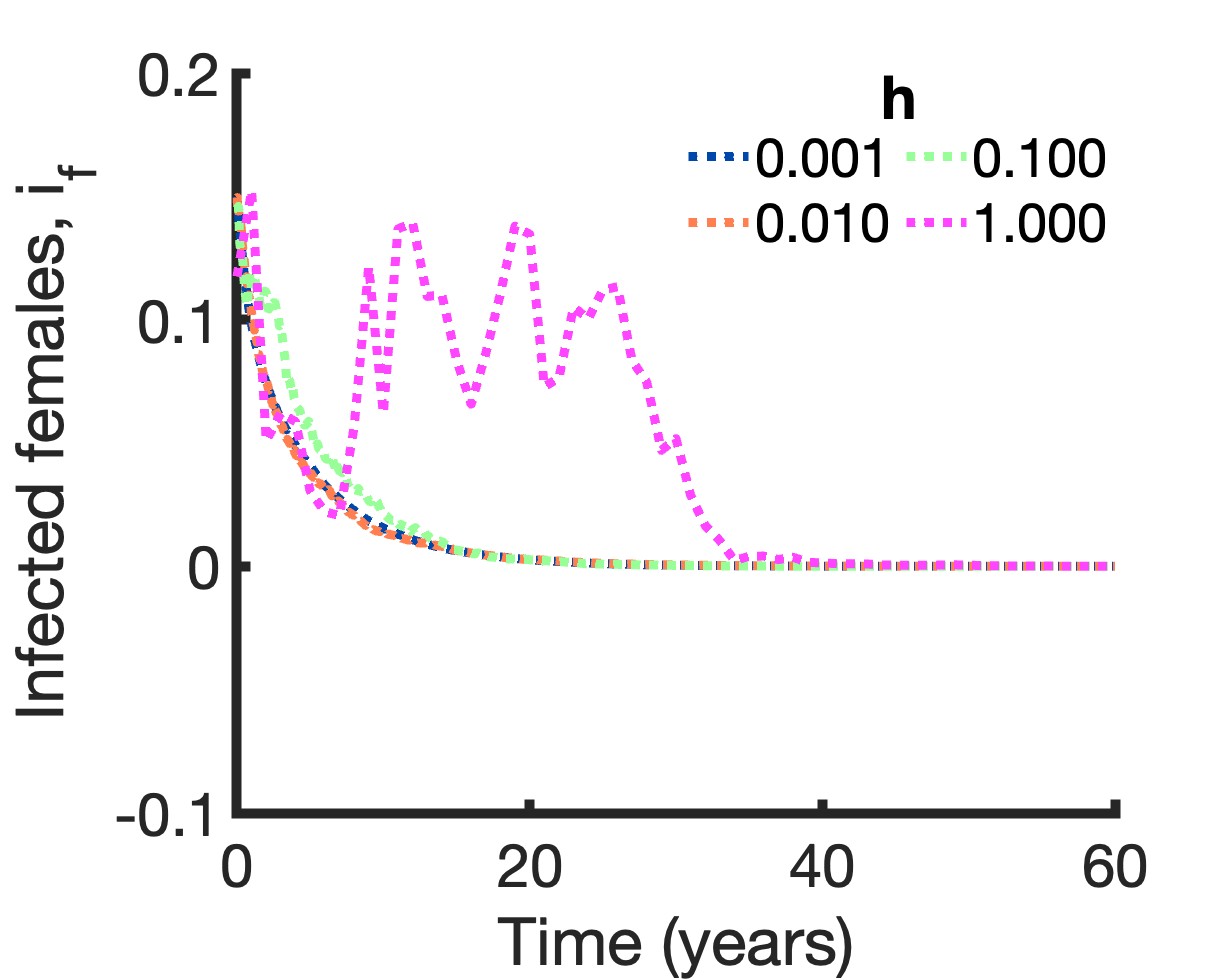}\label{fig:sub4122}}
	\subfigure[$t_f$ vs $t$.]{\includegraphics[width=0.4\linewidth]{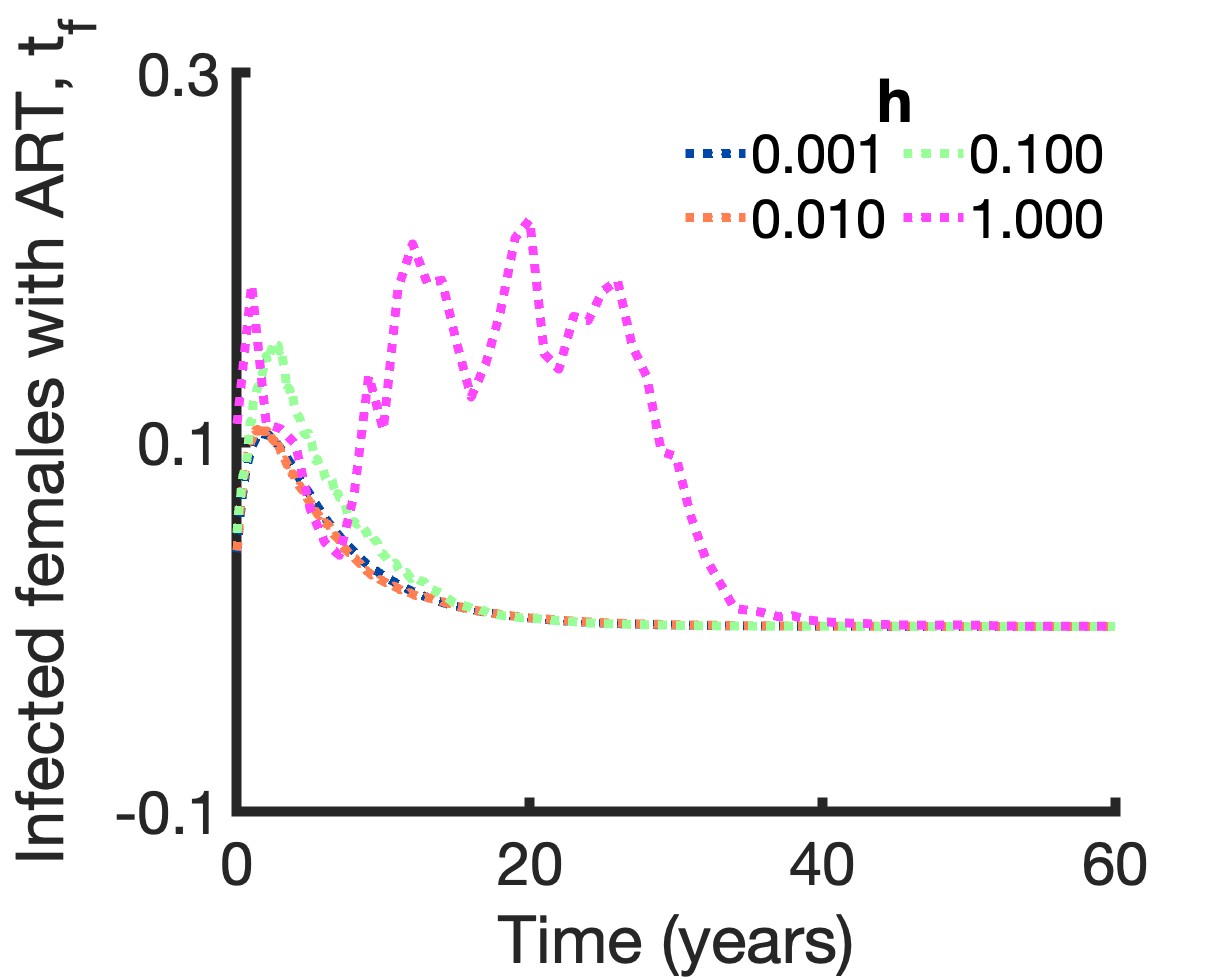}\label{fig:sub4132}}
	 \caption{Disease-Free Equilibrium (DFE) of the stochastic Runge-Kutta model for varying values of step size $h$.}
 \label{fig:DFE_StoRK_h}
\end{figure}

\subsection{Comparing Three Methods}
In comparing the deterministic, stochastic Euler, and stochastic Runge-Kutta methods applied to our model, several key differences emerge in their outcomes and computational characteristics. The deterministic method provides a clear, smooth prediction of population trends ($i_m$, $i_f$, $t_f$) without considering random fluctuations, making it straightforward but potentially oversimplified. The stochastic Euler method introduces randomness, adding realistic variability to the results, which depict how populations might actually behave under uncertain conditions. However, this method can sometimes oversimplify complex dynamics due to its basic numerical approach. The stochastic Runge-Kutta method, on the other hand, offers a more sophisticated numerical technique that captures a wider range of fluctuations and provides a finer resolution of the model’s dynamics (see figures \ref{fig:comparing} and \ref{fig:comparing_DFE}). \\

\noindent Specifically, figure \ref{fig:sub611} shows the number of infected males ($i_m$) over time, demonstrating how each method captures the fluctuations and trends in male infection rates. Figure  \ref{fig:sub612} illustrates the number of infected females ($i_f$) over time, highlighting the differences in infection dynamics across the methods. Figure  \ref{fig:sub613} focuses on the number of infected females receiving antiretroviral therapy (ART) ($t_f$), revealing the impact of ART coverage on the infection trends among females.

\begin{figure}[H]
	\centering
	\subfigure[$i_m$ vs $t$.]{\includegraphics[width=0.4\linewidth]{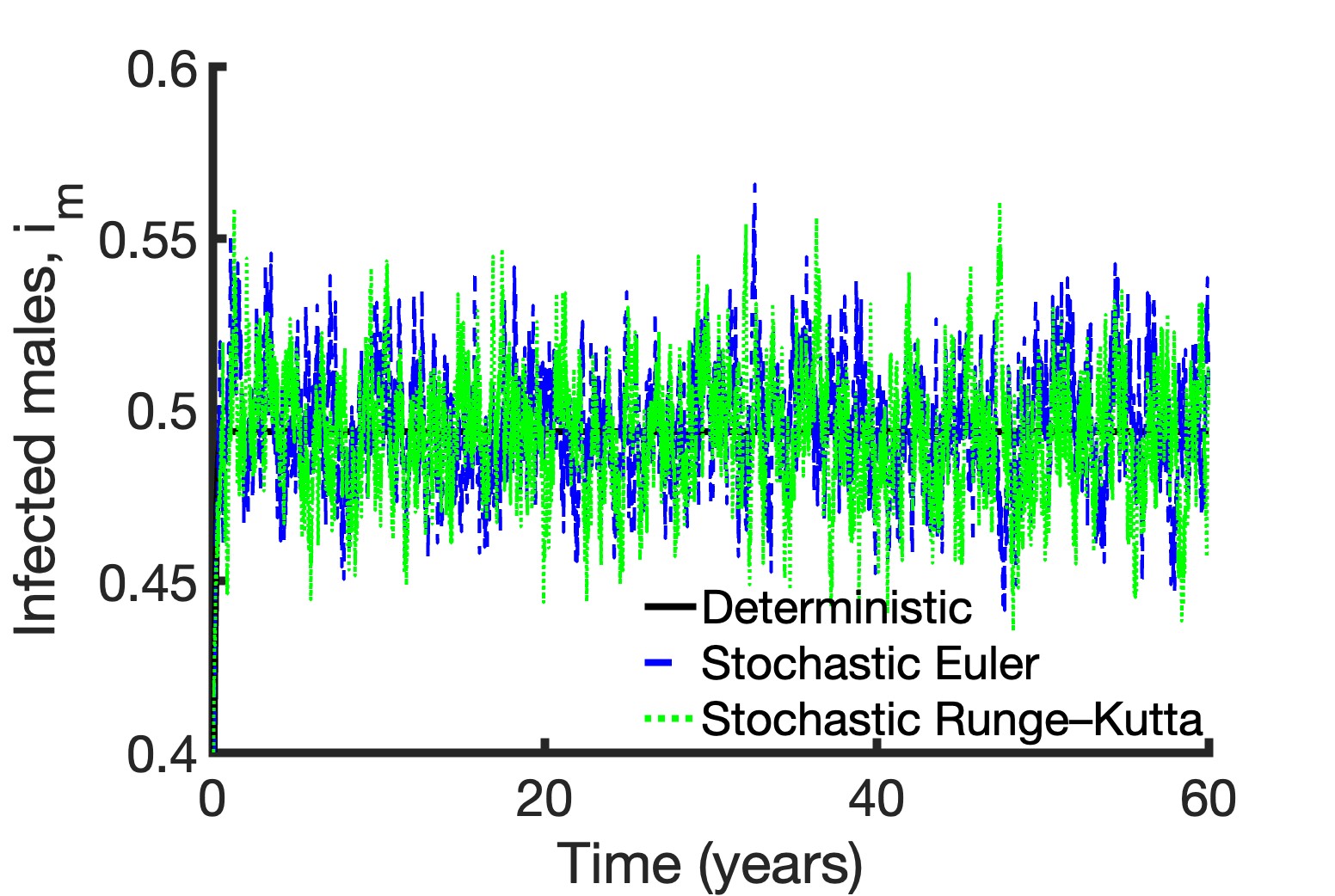}\label{fig:sub611}}
	\subfigure[$i_f$ vs $t$.]{\includegraphics[width=0.4\linewidth]{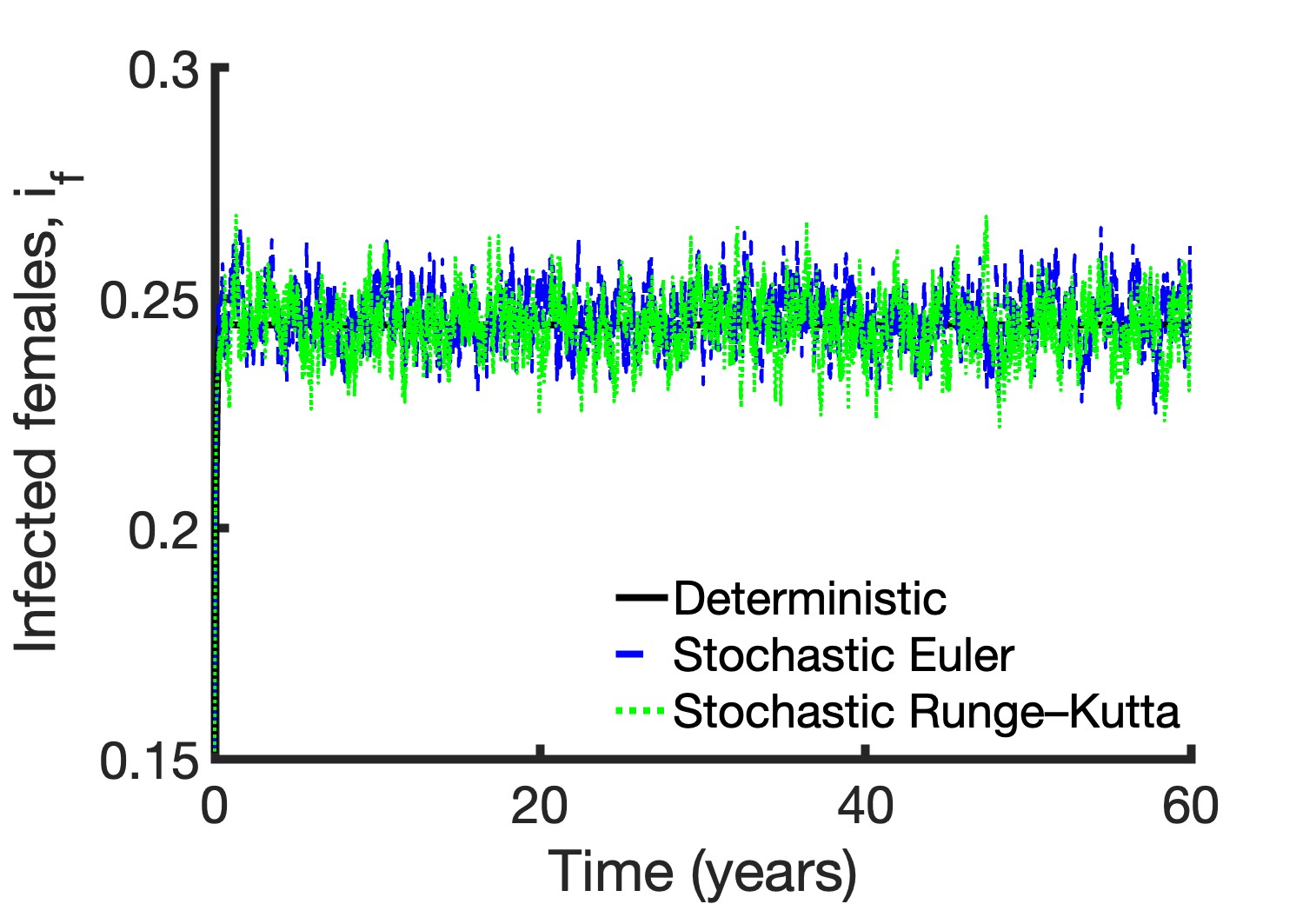}\label{fig:sub612}}
	\subfigure[$t_f$ vs $t$.]{\includegraphics[width=0.4\linewidth]{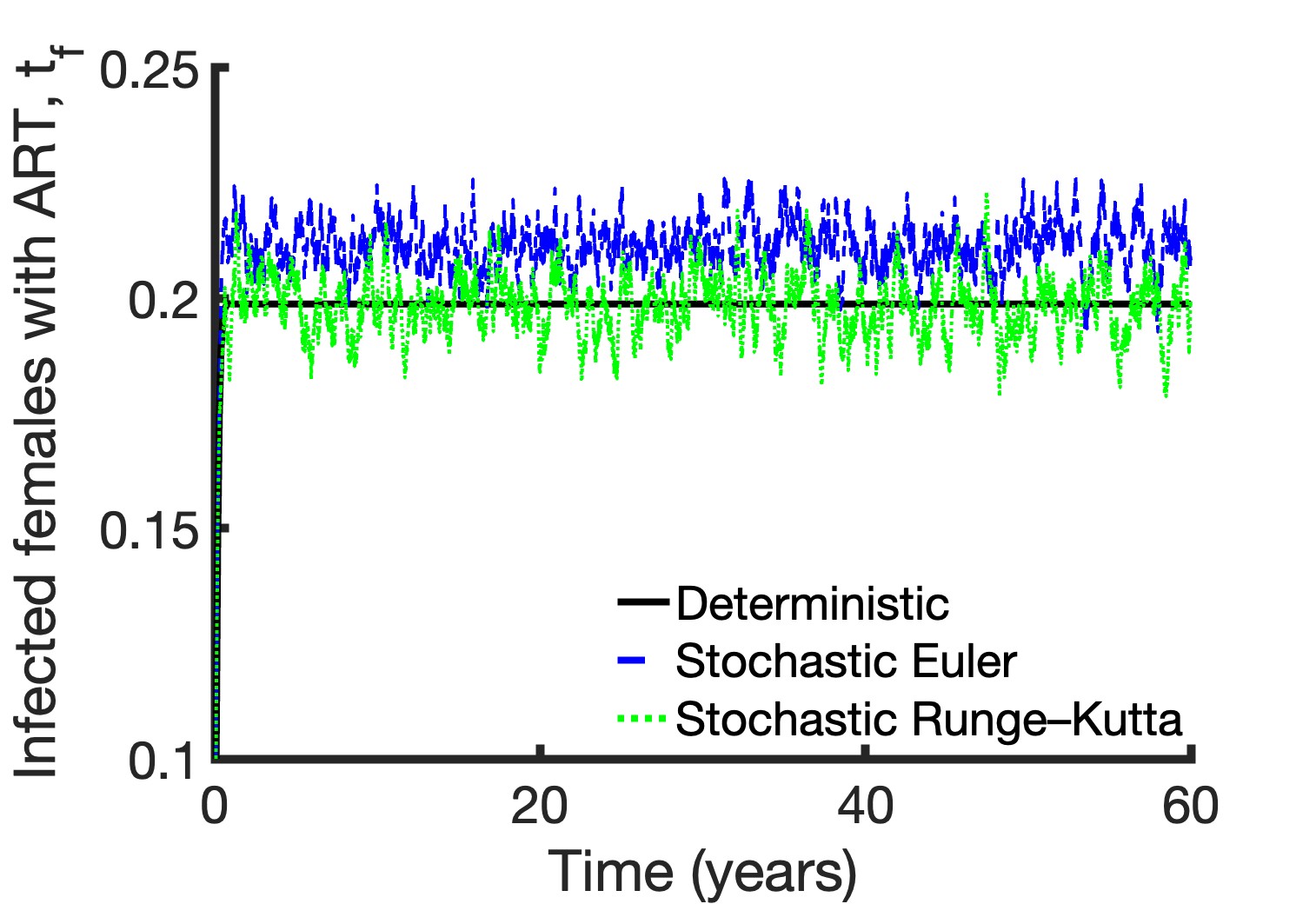}\label{fig:sub613}}
 \caption{Comparing three methods (Endemic Equilibrium).}
 \label{fig:comparing}
\end{figure}

\noindent  In the Disease-Free Equilibrium, figure \ref{fig:sub621} depicts the number of infected males ($i_m$), showing how each method approaches the disease-free state for the male population. Figure \ref{fig:sub622} presents the number of infected females ($i_f$), indicating the effectiveness of each method in reducing female infection rates. Lastly, figure \ref{fig:sub623} illustrates the number of infected females receiving ART ($t_f$), emphasizing the role of ART in achieving a disease-free state across the different modeling approaches.

\begin{figure}[H]
	\centering
	\subfigure[$i_m$ vs $t$.]{\includegraphics[width=0.4\linewidth]{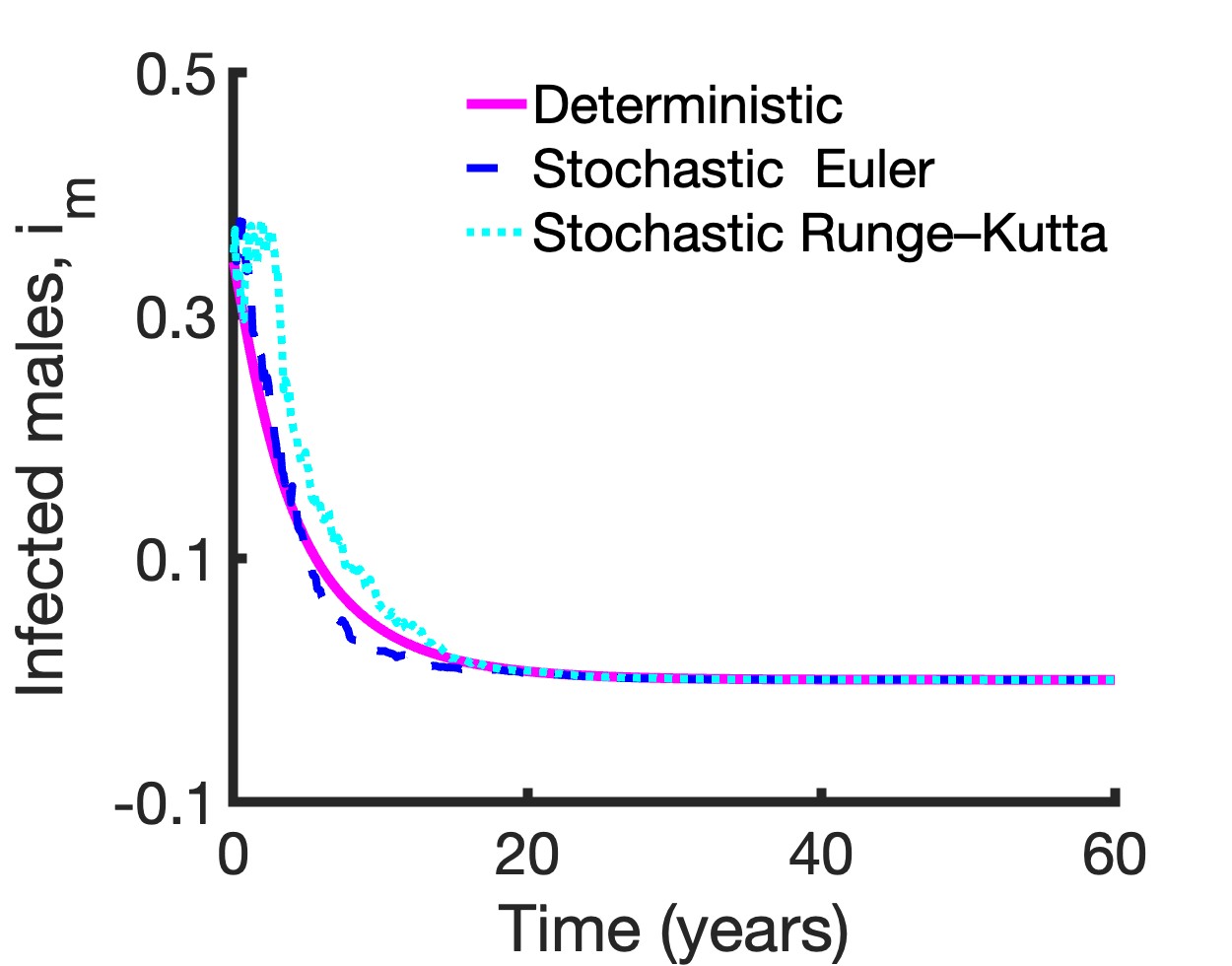}\label{fig:sub621}}
	\subfigure[$i_f$ vs $t$.]{\includegraphics[width=0.4\linewidth]{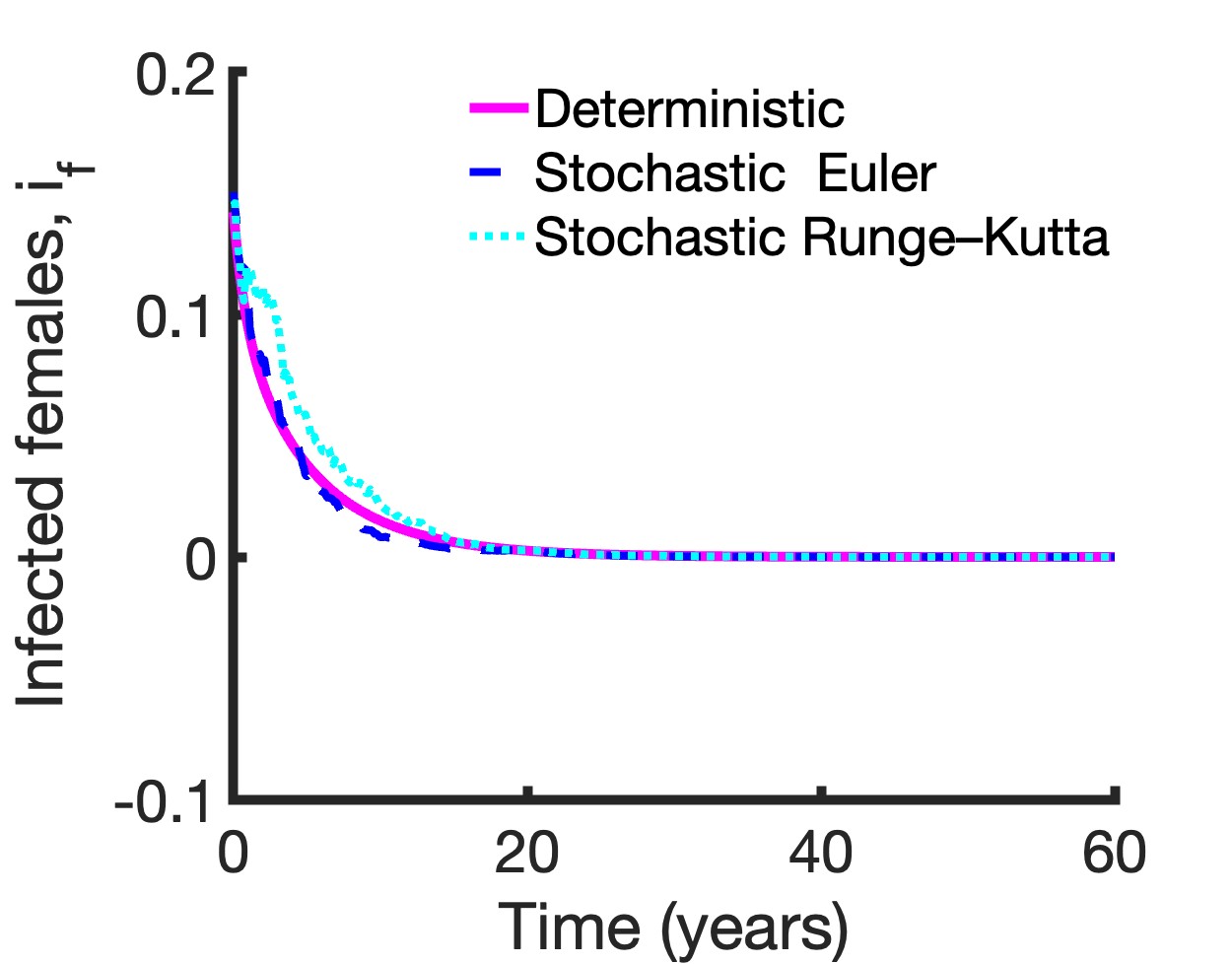}\label{fig:sub622}}
	\subfigure[$t_f$ vs $t$.]{\includegraphics[width=0.4\linewidth]{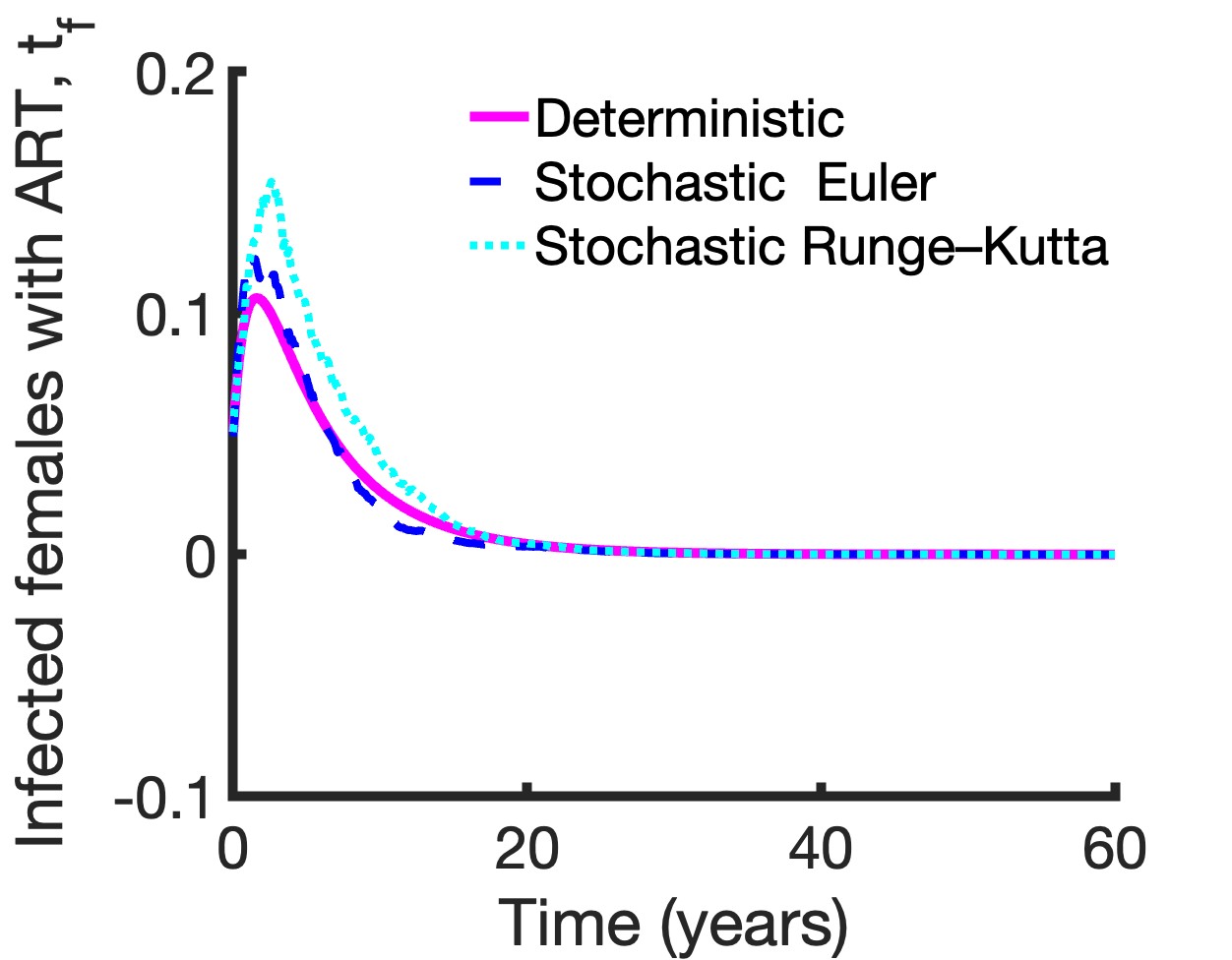}\label{fig:sub623}}
	\caption{Comparing three methods (Disease-Free Equilibrium).}
	\label{fig:comparing_DFE}
\end{figure}

\noindent It is particularly effective at yielding more accurate and detailed outcomes, especially useful in scenarios where precision is crucial. Each method has its advantages, depending on the required precision and the computational resources available.

\subsection{Correlation Coefficients}
\label{sec:cc}
\noindent Correlation coefficients measure the strength and direction of the relationship between two variables \cite{asuero2006correlation}. In this analysis, we explored how the parameter $\delta_f$, which denotes the proportion of infected females receiving antiretroviral treatment (ART), affects the relationships between \textbf{i)} infected males and females ($i_m$ and $i_f$), and \textbf{ii)} infected males and treated females ($i_m$ and $t_f$). For both Endemic and Disease-Free Equilibrium, these relationships are positive, indicating that an increase in $i_m$ corresponds to increases in both $i_f$ and $t_f$. Additionally, in the Endemic Equilibrium scenario, higher correlation coefficients at smaller $\delta_f$ values show that lower values of $\delta_f$ lead to stronger correlations. The pattern is similar in the Disease-Free Equilibrium, with no inverse relationships observed in either case (see figures \ref{fig_a} and \ref{fig_b}).
\subsubsection*{Endemic Equilibrium}
\begin{table}[H]
\centering
\begin{tabular}{cccc}
\toprule
\multicolumn{1}{c}{$\delta_f$} & \multicolumn{2}{c}{Values of Correlation Coefficients} & \multicolumn{1}{c}{Relationship} \\
\cmidrule(lr){2-3}
& \multicolumn{1}{c}{($i_m$, $i_f$)} & \multicolumn{1}{c}{($i_m$, $t_f$)} & \\
\midrule
0.35 & 0.9847 & 0.9817 & Direct \\
0.50 & 0.9737 & 0.9837 & Direct \\
0.70 & 0.9500 & 0.9737 & Direct \\
0.90 & 0.8707 & 0.9221 & Direct \\
\bottomrule
\end{tabular}
\caption{Correlation Coefficients for various values of $\delta_f$.}
\label{delta_f_cc_ee}
\end{table}

\subsubsection*{Disease-Free Equilibrium}
\begin{table}[H]
\centering
\begin{tabular}{cccc}
\toprule
\multicolumn{1}{c}{$\delta_f$} & \multicolumn{2}{c}{Values of Correlation Coefficients} & \multicolumn{1}{c}{Relationship} \\
\cmidrule(lr){2-3}
& \multicolumn{1}{c}{($i_m$, $i_f$)} & \multicolumn{1}{c}{($i_m$, $t_f$)} & \\
\midrule
0.35 & 0.9985 & 0.9651 & Direct \\
0.40 & 0.9990 & 0.9616 & Direct \\
0.70 & 0.9992 & 0.9402 & Direct \\
0.90 & 0.9963 & 0.9319 & Direct \\
\bottomrule
\end{tabular}
\caption{Correlation Coefficients for various values of $\delta_f$.}
\label{delta_f_cc_dfe}
\end{table}

\begin{figure}[H]
	\centering
	\subfigure[Endemic Equilibrium.]{\includegraphics[width=0.45\linewidth]{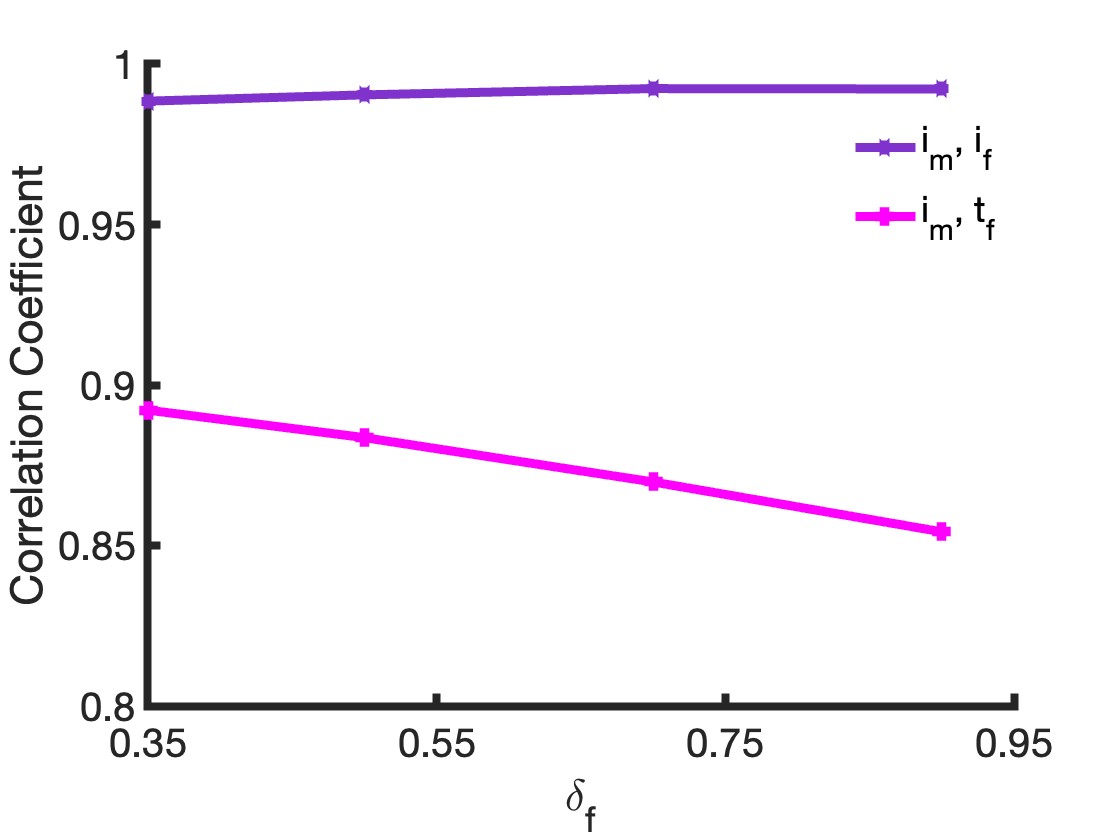}\label{fig_a}}
	\subfigure[Disease-Free Equilibrium.]{\includegraphics[width=0.45\linewidth]{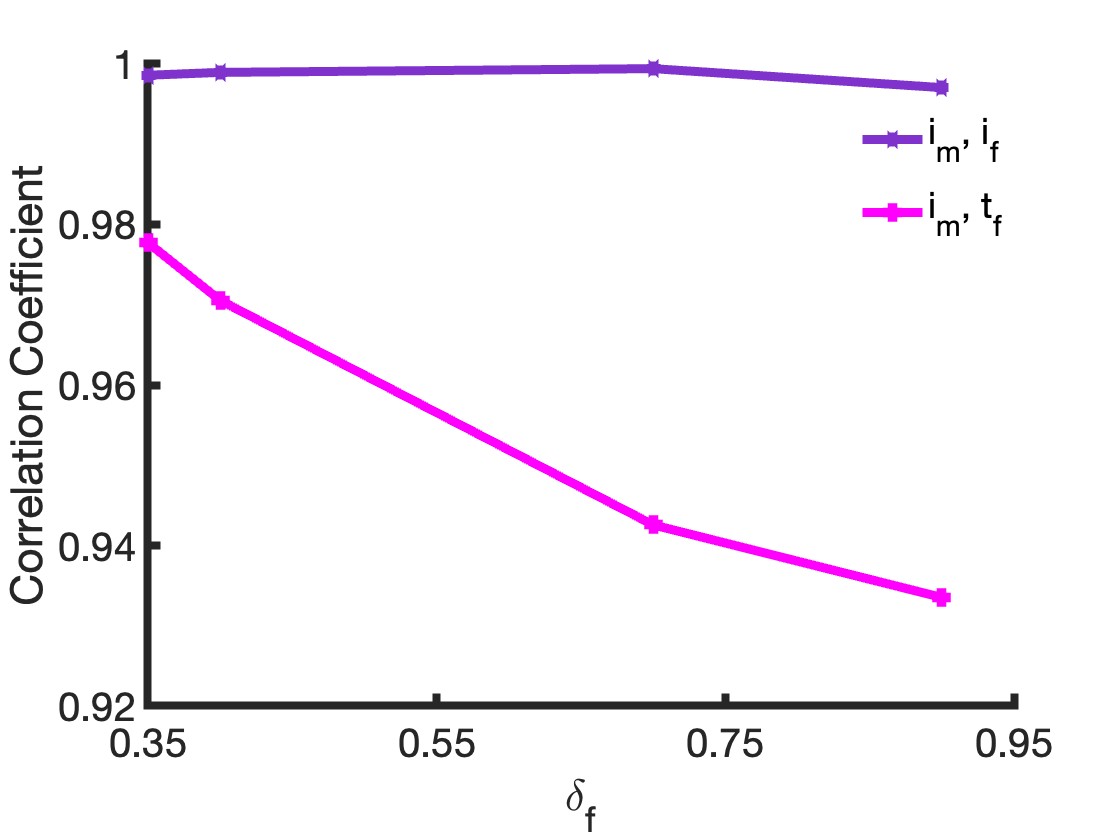}\label{fig_b}}
\caption{Correlation Coefficient vs $\delta_f$.}
\label{fig:both}
\end{figure}

\section{Results and Discussions}
\label{sec:results}
The findings of this study underscore the significant advantages that stochastic models offer over deterministic models in accurately representing the complex dynamics of the HIV epidemic. Stochastic models, particularly those based on the stochastic Runge-Kutta technique, are more effective at capturing the inherent variability and unpredictability present in real-world scenarios. These models account for random fluctuations in transmission and treatment processes, providing a more complex understanding of epidemic dynamics. This study revealed that reductions in transmission probabilities for both men and women led to a significant decline in the number of infected individuals, a consistent pattern across both the deterministic and stochastic models. Crucially, increasing the proportion of infected women receiving antiretroviral therapy (ART) resulted in a marked reduction in the number of infected men and women. This highlights the essential role that ART plays in controlling the spread of the virus and managing the epidemic. One of the key insights from the study is the superior performance of stochastic models, particularly those using the Runge-Kutta method, in capturing minor fluctuations and trends that deterministic models may overlook. Stochastic models provide a finer level of detail, which is vital for understanding the subtle dynamics of the epidemic. This level of granularity allows for a better appreciation of how small perturbations in the system can lead to significant changes in disease transmission, making stochastic models a powerful tool for public health planning. By integrating stochastic elements, these models offer a more realistic portrayal of the epidemic’s behavior, which can enhance the accuracy and effectiveness of public health interventions. The graphical results presented in figures \ref{fig:sub314}, \ref{fig:sub324}, and \ref{fig:sub334} illustrate the impact of increasing ART coverage among women on the overall infected population. These figures clearly demonstrate the benefit of higher ART uptake, as it leads to a substantial reduction in infection rates for both men and women. Furthermore, figures \ref{fig:DFE_StoRK_gamma_f}, \ref{fig:DFE_StoRK_gamma_m}, and \ref{fig:DFE_StoRK_delta_f} highlight the critical influence of transmission probabilities on infection rates, underscoring the importance of targeted interventions aimed at reducing these probabilities. These visual representations reinforce the study’s conclusion that ART, coupled with well-targeted interventions, is a powerful tool in epidemic management.

A thorough comparative analysis of the deterministic and stochastic models revealed that stochastic models, particularly those using the stochastic Runge-Kutta method, offer a higher resolution of dynamic changes within the epidemic. This higher resolution improves forecast accuracy and deepens the understanding of minor perturbations that can influence the course of the epidemic. By incorporating random variability, stochastic models better simulate the unpredictability of real-world disease transmission and treatment efficacy, providing a more accurate framework for developing public health strategies. The results of this study emphasize the critical role of stochastic models in designing effective public health interventions. By more accurately representing the uncertainties inherent in disease transmission, these models can help policymakers develop more responsive strategies to control the spread of HIV. This is particularly important in resource-limited settings, where efficient resource allocation can have a significant impact on the success of epidemic control efforts. Additionally, the study highlights the importance of ART in reducing infection rates and suggests that expanding ART coverage should be a key component of global health strategies aimed at reducing HIV prevalence and improving the quality of life for people living with AIDS. The findings of this study align with global health objectives to reduce HIV prevalence and enhance patient outcomes. By demonstrating the advantages of stochastic modeling, this research provides a comprehensive framework for managing the HIV epidemic more effectively. Future research should focus on further refining computational techniques for stochastic models and exploring their application to other infectious diseases. Such advancements could contribute to even more effective public health measures, allowing for better management of epidemics and improving the overall quality of health care in resource-constrained environments. Additionally, future work could explore how these models can be adapted to account for the evolving nature of viral epidemics, including the emergence of drug-resistant strains and the changing patterns of transmission due to behavioral, social, and economic factors \cite{khan2023vaccine}.

\section*{Acknowledgments}
The author M. Kamrujjaman acknowledged to the University Grants Commission (UGC), 	and  the  University of Dhaka, Bangladesh.

\section*{Conflict of interest}
The authors declare no conflict of interest.


\section*{Data sharing}
No human data used in this study. 

\section*{Ethical approval}
No consent is required to publish this manuscript.

\section*{CRediT authorship contribution statement}
Nuzhat Nuari Khan: Conceptualization, Data curation, Formal analysis, Methodology, Software, Validation, Writing – original
draft.\\
Md. Kamrujjaman: Conceptualization, Formal analysis, Funding
acquisition, Software, Supervision,  Writing – original draft, Writing – review \& editing.\\
Shohel Ahmed: Validation, Resources, Investigation,  Software, Writing – review \& editing.

\section*{Ethics Statement}
None. 

\clearpage
\appendix
\section{Appendix}\label{appen}
\subsection*{Stochastic Euler Method}
The equations for the stochastic Euler method are as follows \cite{raza2019reliable,arif2019reliable}:
\begin{align}
	i_{m}^{n+1}(t) &= i_{m}^{n}(t) + h [c_{m}\gamma_{f}i_{f}^{n}(t)(1 - i_{m}^{n}(t)) - b i_{m}^{n}(t) + \sigma_{0}i_{m}^{2n}(t) +\delta_{1}\Delta B_{n}i_{m}^{n}(t)] \tag{15} \\
	i_{f}^{n+1}(t) &= i_{f}^{n}(t) + h [c_{f}\gamma_{m}i_{m}^{n}(t)(1 - i_{f}^{n}(t) - t_{f}^{n}(t)) - (b + \delta_{f})i_{f}^{n}(t) + \sigma i_{f}^{n}(t)i_{f}^{n}(t) \notag & \\
	&\quad +\sigma_0 i_{f}^{2n}(t) + \sigma_{2}\Delta B_{n}i_{f}^{n}(t)] \tag{16} \\
	t_{f}^{n+1}(t) &= h_{f3}^{n}(t) + h [\sigma_{f}i_{f}^{n}(t) - (b + \sigma)t_{f}^{n}(t) + \sigma_{0}i_{f}^{n}(t)t_{f}^{n}(t) + \sigma t_{f}^{2n}(t) + \delta_{3}\Delta B_{n}t_{f}^{n}(t)] \tag{17}
\end{align}
where $h$ stands for the time step size and \(\Delta B_{n}\) is usually distributed between the drift and diffusion coefficients, that is, \(\Delta B_{n} \sim N(0, 1)\).

\subsection*{Stochastic Runge-Kutta Method}
The equations for the stochastic Runge-Kutta method are as follows \cite{raza2019reliable,arif2019reliable}:\\

\noindent \textit{Stage 01:}
\vspace{-1mm}
\begin{align*}
	P_1 &= h\left[c_{m}\gamma_{f}i_{f}^{n}(t)(1 - i_{m}^{n}(t)) - b i_{m}^{n}(t) 
	+ \sigma_{0}i_{m}^{2n}(t) + \delta_{1}\Delta B_{n}i_{m}^{n}(t)\right] \\
	Q_1 &= h\left[c_{f}\gamma_{m}i_{m}^{n}(t)(1 - i_{f}^{n}(t) - t_{f}^{n}(t)) - (b + \delta_{f})i_{f}^{n}(t) + \sigma i_{f}^{n}(t) t_{f}^{n}(t)+ \sigma_0 i_{f}^{2n}(t) + \delta_{2}\Delta B_{n}i_{f}^{n}(t)\right] \\
	R_1 &= h\left[\delta_{f}i_{f}^{n}(t) - (b + \sigma)t_{f}^{n}(t) + \sigma_{0}i_{f}^{n}(t)t_{f}^{n}(t) + \sigma t_{f}^{2n}(t)+ \delta_{3}\Delta B_{n}t_{f}^{n}(t)\right]
\end{align*}
\textit{Stage 02:}
\vspace{-1mm}
\begin{flalign*}
	&P_2 = h \bigg[c_{m}\gamma_{f}\left(i_{f}^{n}(t) + \frac{Q_1}{2}\right)\left(1 - \left(i_{m}^{n}(t) + \frac{P_1}{2}\right)\right) - b\left(i_{m}^{n}(t) + \frac{P_1}{2}\right)+ \sigma_{0}\left(i_{m}^{n}(t) + \frac{P_1}{2}\right)^{2} &\\
	&\quad + \delta_{1}\Delta B_{n}\left(i_{m}^{n}(t) + \frac{P_1}{2}\right)\bigg] &\\
	&Q_2 = h \bigg[c_{f}\gamma_{m}\left(i_{m}^{n}(t) + \frac{P_1}{2}\right)\left(1 - \left(i_{f}^{n}(t) + \frac{Q_1}{2}\right) - \left(t_{f}^{n}(t) + \frac{R_1}{2}\right)\right)- (b + \delta_{f})\left(i_{f}^{n}(t) + \frac{Q_1}{2}\right) &\\
	&\quad + \sigma \left(i_{f}^{n}(t) + \frac{Q_1}{2} \right) \left(t_{f}^{n}(t) + \frac{R_1}{2} \right)   + \sigma_{0}\left(i_{f}^{n}(t) + \frac{Q_1}{2}\right)^{2} + \delta_{2}\Delta B_{n}\left(i_{f}^{n}(t) + \frac{Q_1}{2}\right)\bigg] &\\
	&R_2 = h \bigg[\delta_{f}\left(i_{f}^{n}(t) + \frac{R_1}{2}\right) - (b + \sigma)\left(t_{f}^{n}(t) + \frac{R_1}{2}\right) + \sigma_{0}\left(i_{f}^{n}(t) + \frac{R_1}{2}\right)\left(t_{f}^{n}(t) + \frac{R_1}{2}\right) &\\
	&\quad +\sigma \left( t_{f}^{n}(t)+ \frac{R_1}{2}\right)^{2} +\delta_{3}\Delta B_{n}\left(t_{f}^{n}(t) + \frac{R_1}{2}\right)\bigg] &
\end{flalign*}
\textit{Stage 03:}
\vspace{-1mm}
\begin{flalign*}
	P_3 &= h \bigg[c_m\gamma_f \left(i_{f}^{n}(t) + \frac{Q_2}{2}\right)\left(1 - i_{m}(t) + \frac{P_2}{2}\right) - b\left(i_{m}^{n}(t) + \frac{P_2}{2}\right)+ \sigma_0\left(i_{m}^{n}(t) + \frac{P_2}{2}\right)^2 &\\
	&\quad + \delta_1\Delta B_n \left(i_{m}^{n}(t) + \frac{P_2}{2}\right)\bigg] &\\
	Q_3 &= h \bigg[c_f\gamma_m \left(i_{m}^{n}(t) + \frac{P_2}{2}\right)\times \left(1 - \left(i_{f}^{n}(t) + \frac{Q_2}{2}\right) - \left(t_{f}^{n}(t) + \frac{R_2}{2}\right)\right) - (b + \delta_f)\left(i_{f}^{n}(t) + \frac{Q_2}{2}\right)&\\
	&\quad + \sigma\left(i_{f}^{n}(t)+ \frac{Q_2}{2}\right)   \left(t_{f}^{n}(t) + \frac{R_2}{2}\right) + \sigma_0\left(i_{f}^{n}(t) + \frac{Q_2}{2}\right)^2+ \delta_2\Delta B_n \left(i_{f}^{n}(t) + \frac{Q_2}{2}\right) &\\
	R_3 &= h \bigg[\delta_f\left(i_{f}^{n}(t) + \frac{Q_2}{2}\right) - (b + \sigma)\left(t_{f}^{n}(t) + \frac{R_2}{2}\right) + \sigma_0\left(i_{f}^{n}(t) + \frac{Q_2}{2}\right)\left(t_{f}^{n}(t) + \frac{R_2}{2}\right)&\\
	&\quad + \sigma\left(t_{f}^{n}(t) + \frac{R_2}{2}\right)^2 + \delta_3\Delta B_n \left(t_{f}^{n}(t) + \frac{R_2}{2}\right)\bigg]  &
\end{flalign*}
\textit{Stage 04:}
\vspace{0.5mm}
\begin{flalign*}
	&P_4 = h [c_m \gamma_f (i_{f}^{n}(t) + Q_3)(1 - (i_{m}^{n}(t) + P_3)) - b(i_{m}^{n}(t) + P_3)]+ \sigma_0(i_{m}^{n}(t) + P_3)^2 + \delta_1\Delta B_n(i_{m}^{n}(t)  &\\
	&\quad + P_3] &\\
	&Q_4 = h [c_f \gamma_m (i_{m}^{n}(t) + P_3)(1 - (i_{f}^{n}(t) + Q_3) - (t_{f}^{n}(t) + R_3))- (b + \delta_f)(i_{f}^{n}(t) + Q_3) &\\
	&\quad + \sigma(i_{f}^{n}(t)) + Q_3)(t_{f}^{n}(t) + R_3) + \sigma_0 (i_{f}^{n}(t)) + Q_3)^2 + \delta_2\Delta B_n(t_{f}^{n}(t) + Q_3)] &\\
	&R_4 = h [\delta_{f}i_{f}^{n}(t) + Q_3) - (b + \sigma)(t_{f}^{n}(t) + R_3)+ \sigma_0 i_{f}^{n}(t) + Q_3)(t_{f}^{n}(t) + R_3)+ \sigma(t_{f}^{n}(t) + R_3)^2 &\\
	&\quad + \delta_3\Delta B_n (t_{f}^{n}(t) + R_3)] &
\end{flalign*}
\textit{Final Stage}
\vspace{0.5mm}
\begin{equation*}
	\begin{aligned}
		i^{n+1}_{m}(t) &= i^n_{m}(t) + \frac{1}{6}[P_1 + 2P_2 + 2P_3 + P_4] \\
		i^{n+1}_{f}(t) &= i^n_{f}(t) + \frac{1}{6}[Q_1 + 2Q_2 + 2Q_3 + Q_4] \\
		t^{n+1}_{f}(t) &= t^n_{f}(t) + \frac{1}{6}[R_1 + 2R_2 + 2R_3 + R_4] \\
	\end{aligned}
\end{equation*}
where the time step size is denoted by h and $\Delta B_n \sim N(0, 1)$.

\subsection*{Stochastic Non-Standard Finite Difference Method}
The equations for the stochastic Euler method are as follows \cite{raza2019reliable,arif2019reliable}:
\begin{align*}
	\displaystyle
	i^{n+1}_{m}(t) &= \frac{i^n_{m}(t) + \phi(h) \left[ c_{m} \gamma_f i^n_{f}(t) + \sigma_0 i^n_{m}(t)i^n_{m}(t) + \delta_{1} i^n_{m}(t) \Delta B_n \right]}{1 + \phi(h) \left[ 1 + c_{m}\gamma_f i^n_{f}(t) + b \right]} \\
	i^{n+1}_{f}(t) &= \frac{
		i^n_{f}(t) + \phi(h) \left[c_{f} \gamma_m i^{n+1}_{m} \left(1 - i^{n+1}_{f}\right) + \sigma i^n_{f} t^n_{f} + \sigma_0 i^n_{f} i^n_{f}\right] + \phi(h) \delta_2 i^n_{f}(t) \Delta B_n
	}{
		(1 + \phi(h) c_{f} \gamma_m i^{n+1}_{m} + \phi(h) (b + \delta_f))
	} \\
	{t^{n+1}_{f}(t)} &= \frac{
		t^n_{f}(t) + \phi(h) \left[
		\delta_{f} i^{n+1}_{f}(t) + \sigma_0 i^{n+1}_{f} t^{n}_{f}(t) + \sigma t^{n}_{f}(t) t^{n}_{f}(t) + \delta_3 t^{n}_{f}(t) \Delta B_n
		\right]
	}{
		\left[ 1 + \phi(h)(b + \sigma) \right]
	}
\end{align*}
where, $\phi(h)= 1- e^{-h}$ and $h$ is any time step size.\\

\noindent\textbf{Stability of Stochastic Non-Standard Finite Difference (SNSFD) Method:}\\
We consider \( K \), \( M \), and \( N \) as follows:
\begin{align*}
	K &= i_{m}(t) + \frac{h\left[c_{m}\gamma_f i_{f}(t) + \sigma_0 i^2_{m}(t)\right] + \delta_1 i_{m}(t)\Delta B_n}{\left[1 + hc_{m}\gamma_f i_{f}(t) + hb\right]} \\
	M &= i_{f}(t) + \frac{h\left[c_{f}\gamma_m i_{m}(t)(1 - t_{f}(t)) + \sigma i_{f}(t)t_{f}(t) + \sigma_0 i^2_{f}(t) + \delta_1 i_{f}(t)\Delta B_n\right]}{\left[1 + hc_{f}\gamma_m i_{m}(t) + h(b + \delta_f)\right]} \\
	N &= t_{f}(t) + \frac{h\left[\delta_{f}i_{f}(t) + \sigma_0 i_{f}(t)t_{f}(t) + \sigma t^2_{f}(t) + \delta_3 t_{f}(t)\Delta B_n\right]}{\left[1 + h(b + \sigma)\right]}
\end{align*}

\noindent The Jacobian matrix is:
\[
\renewcommand{\arraystretch}{1.3} 
J =
\left[
\begin{array}{ccc}
	\frac{\partial K}{\partial i_m} & \frac{\partial K}{\partial i_f} & \frac{\partial K}{\partial t_f} \\
	\frac{\partial M}{\partial i_m} & \frac{\partial M}{\partial i_f} & \frac{\partial M}{\partial t_f} \\
	\frac{\partial N}{\partial i_m} & \frac{\partial N}{\partial i_f} & \frac{\partial N}{\partial t_f}
\end{array}
\right]
\]
In order to achieve disease-free equilibrium, the Jacobian matrix can be obtained as follows:
\[(i_m(t),i_f(t),t_f(t))=(0,0,0).\]
The Eigenvalue ($\lambda$) of the Jacobian matrix is as follows:
\[
\lambda = \frac{1 + h\delta_3\Delta B_n}{1 + h(b + \sigma)} < 1.
\]
Minor noise disturbances, denoted by \(\delta_1, \delta_2,\) and \(\delta_3\), introduce stochasticity into the HIV/AIDS model. These disturbances are linked to Brownian motions \(B_k(t)\) (where, \(k = 1, 2, 3\)) across various compartments. Consequently, each stochastic term \(\delta_i\) (for \(i = 1, 2, 3\)) is less than \(b\) and \(\sigma\), where \(b\) represents the natural birth rate and \(\sigma\) signifies the human death rate, as detailed in \cite{raza2019reliable}.

\begin{lemma}
	For the quadratic equation $\lambda^2 - C\lambda + D = 0$, $|\lambda_i| < 1, i = 1, 2;$ if and only if succeeding conditions are satisfied \cite{chowell2009basic}:
	\begin{enumerate}
		\item[(i)] $1 - C + D > 0$
		\item[(ii)] $1 + C + D > 0$
		\item[(iii)]$D < 1$.
	\end{enumerate}
	where,\\ $C = \text{Jacobian matrix's Trace}$. \\
	$D = \text{Jacobian matrix's Determinant}.$   
\end{lemma} 
\begin{proof}
	\textbf{(i)}
	\[
	1 - C + D > 0
	\]
	\begin{flalign*}
		& (1 + hb)(1 + h(b + \delta_f)) - 2 - h\delta_1\Delta B_n - h\delta_2\Delta B_n - h^2 b \delta_1\Delta B_n - h^2\delta_f \delta_f\Delta B_n -h^2 b \delta_2\Delta B_n- 2hb  &\\
		& - h\delta_f + 1 + h\delta_1\Delta B_n + h\delta_2\Delta B_n + h^2\delta_1\delta_2\Delta B_n\Delta B_n - h^2 c_m c_f \gamma_m \gamma_f > 0 &\\
		&\Rightarrow h^2 \left[b^2 - b(\delta_1\Delta B_n + \delta_2\Delta B_n - \delta_f) - \delta_1\delta_f\Delta B_n + \delta_1\delta_2\Delta B_n\Delta B_n - c_m c_f \gamma_m \gamma_f\right] > 0 &\\
		&\Rightarrow h^2 > 0
	\end{flalign*}
	For any time step size $h$, the condition $h > 0$ holds.\\
	
	\textbf{(ii)}
	\[
	1 + C + D > 0
	\]
	Since \(1 > 0\) and \(C > 0\) so it is enough to show \(D > 0\)
	\begin{flalign*}
		&\Rightarrow \frac{1 + h\delta_1\Delta B_n + h\delta_2\Delta B_n + h^2\delta_1\delta_2\Delta B_n\Delta B_n- h^2 c_m c_f \gamma_m \gamma_f}{(1 + hb)(1 + h(b + \delta_f))} > 0 &\\
		&\Rightarrow 1+h\delta_1\Delta B_n+h\delta_2\Delta B_n+h^2\delta_1\delta_2\Delta B_n\Delta B_n-h^2 c_m c_f \gamma_m \gamma_f>0&\\
		&\Rightarrow h^2 \left[c_m c_f \gamma_m \gamma_f - \delta_1\delta_2\Delta B_n\Delta B_n\right] < 1 + h\left[\delta_1\Delta B_n + \delta_2\Delta B_n\right] &\\
		&\Rightarrow h^2 \left[c_m c_f \gamma_m \gamma_f - \delta_1\delta_2\Delta B_n\Delta B_n\right] - 1 - h\left[\delta_1\Delta B_n + \delta_2\Delta B_n\right]<0 &\\
		&\Rightarrow h^2 - \frac{h\left[\delta_1\Delta B_n + \delta_2\Delta B_n\right]} {c_m c_f \gamma_m \gamma_f - \delta_1\delta_2\Delta B_n\Delta B_n} - \frac{1}{c_m c_f \gamma_m \gamma_f - \delta_1 \delta_2 \Delta B_n \Delta B_n} < 0 &\\
		&\Rightarrow h^2 - 2h\left[\frac{\left(\delta_1\Delta B_n + \delta_2\Delta B_n\right)}{2(c_m c_f \gamma_m \gamma_f - \delta_1\delta_2\Delta B_n\Delta B_n)}\right] + \left[\frac{\left(\delta_1\Delta B_n + \delta_2\Delta B_n\right)}{2(c_m c_f \gamma_m \gamma_f - \delta_1\delta_2\Delta B_n\Delta B_n)}\right]^2  &\\
		&\qquad > \left[\frac{\left(\delta_1\Delta B_n + \delta_2\Delta B_n\right)}{2(c_m c_f \gamma_m \gamma_f - \delta_1\delta_2\Delta B_n\Delta B_n)}\right]^2 + \frac{1}{c_m c_f \gamma_m \gamma_f - \delta_1 \delta_2 \Delta B_n \Delta B_n} &\\
		&\Rightarrow \left[\frac{\left(\delta_1\Delta B_n + \delta_2\Delta B_n\right)}{2(c_m c_f \gamma_m \gamma_f - \delta_1\delta_2\Delta B_n\Delta B_n)} - h\right]^2 >  \left[\frac{\left(\delta_1\Delta B_n + \delta_2\Delta B_n\right)}{2(c_m c_f \gamma_m \gamma_f - \delta_1\delta_2\Delta B_n\Delta B_n)}\right]^2 &\\
		&\qquad + \frac{1}{c_m c_f \gamma_m \gamma_f - \delta_1\delta_2\Delta B_n\Delta B_n}  &
	\end{flalign*}
	For any time step size $h$, the condition $h > 0$ holds.\\
	
	\textbf{(iii)}
	\[
	D < 1
	\]
	\begin{flalign*}
		&\Rightarrow \frac{1 + h\delta_1\Delta B_n + h\delta_2\Delta B_n + h^2\delta_1\delta_2\Delta B_n\Delta B_n - h^2 c_m c_f \gamma_m \gamma_f} {(1 + hb)(1 + h(b + \delta_f))}< 1 &\\
		&\Rightarrow 1 + h\delta_1\Delta B_n + h\delta_2\Delta B_n + h^2\delta_1\delta_2\Delta B_n\Delta B_n - h^2 c_m c_f \gamma_m \gamma_f < (1 + hb)(1 + h(b + \delta_f))&\\
		&\Rightarrow h^2(\delta_1\delta_2\Delta B_n\Delta B_n + c_m c_f \gamma_m \gamma_f - b^2 - \delta_f b_0)+ h(\delta_1\Delta B_n + \delta_2\Delta B_n - 2b - \delta_f) > 0 &\\
		&\Rightarrow h^2 + h\left(\frac{\delta_1\Delta B_n + \delta_2\Delta B_n - 2b - \delta_f}{\delta_1\delta_2\Delta B_n\Delta B_n + c_m c_f \gamma_m \gamma_f - b^2 - \delta_f b}\right) > 0 &\\
		&\Rightarrow h^2 + 2h\left(\frac{\delta_1\Delta B_n + \delta_2\Delta B_n - 2b - \delta_f}{2(\delta_1\delta_2\Delta B_n\Delta B_n + c_m c_f \gamma_m \gamma_f - b^2 - \delta_f b)}\right) &\\
		&\qquad + \left(\frac{\delta_1\Delta B_n + \delta_2\Delta B_n - 2b - \delta_f}{2(\delta_1\delta_2\Delta B_n\Delta B_n + c_m c_f \gamma_m \gamma_f - b^2 - \delta_f b)}\right)^2  &\\
		&\qquad > \left(\frac{\delta_1\Delta B_n + \delta_2\Delta B_n - 2b - \delta_f}{2(\delta_1\delta_2\Delta B_n\Delta B_n + c_m c_f \gamma_m \gamma_f - b^2 - \delta_f b)}\right)^2 &\\
		&\Rightarrow \left(h + \frac{\delta_1\Delta B_n + \delta_2\Delta B_n - 2b - \delta_f}{2(\delta_1\delta_2\Delta B_n\Delta B_n + c_m c_f \gamma_m \gamma_f - b^2 - \delta_f b)}\right)^2 
		> \left(\frac{\delta_1\Delta B_n + \delta_2\Delta B_n - 2b - \delta_f}{2(\delta_1\delta_2\Delta B_n\Delta B_n + c_m c_f \gamma_m \gamma_f - b^2 - \delta_f b)}\right)^2 &\\
		&\Rightarrow h + \frac{\delta_1\Delta B_n + \delta_2\Delta B_n - 2b - \delta_f}{2(\delta_1\delta_2\Delta B_n\Delta B_n + c_m c_f \gamma_m \gamma_f - b^2 - \delta_f b)} 
		> \frac{\delta_1\Delta B_n + \delta_2\Delta B_n - 2b - \delta_f}{2(\delta_1\delta_2\Delta B_n\Delta B_n + c_m c_f \gamma_m \gamma_f - b^2 - \delta_f b)}
	\end{flalign*}
	For any time step size $h$, the condition $h > 0$ holds.
\end{proof}

\newpage
\bibliographystyle{unsrt} 
\bibliography{references} 

\end{document}